%% file: main.tex
\documentclass[acmsmall,screen,natbib=false,review=false, anonymous=false,nonacm]{acmart}
\input{datalabmacros} %
\input{individualmacros}

\renewcommand\footnotetextcopyrightpermission[1]{} %
	\setcopyright{none}
\usepackage[style=trad-abbrv,backend=bibtex,doi=true,isbn=false,url=true,eprint=false,maxbibnames=99]{biblatex} 
\begin{document}

\title[Any-k Algorithms for Enumerating Ranked Answers to Conjunctive Queries]{Any-k Algorithms for Enumerating Ranked Answers to Conjunctive Queries}
 
\author{Nikolaos Tziavelis}
\affiliation{%
    \institution{Northeastern University}
    \city{Boston}
    \state{Massachusetts}
    \country{USA}
}
\email{tziavelis.n@northeastern.edu}
\orcid{0000-0001-8342-2177}

\author{Wolfgang Gatterbauer}
\affiliation{%
    \institution{Northeastern University}
    \city{Boston}
    \state{Massachusetts}
    \country{USA}
}
\email{w.gatterbauer@northeastern.edu}
\orcid{0000-0002-9614-0504}

\author{Mirek Riedewald}
\affiliation{%
    \institution{Northeastern University}
    \city{Boston}
    \state{Massachusetts}
    \country{USA}
}
\email{m.riedewald@northeastern.edu}
\orcid{0000-0002-6102-7472}

\begin{abstract} 
We study \emph{ranked enumeration for Conjunctive Queries} (CQs) where the answers are ordered by a given ranking function (e.g., an ORDER BY clause in SQL).
We develop ``\emph{any-k}'' algorithms, which, without knowing the
number $k$ of desired answers,
\emph{push down the ranking into joins} by carefully ordering the computation of intermediate tuples and avoiding materialization of join answers until they are needed. 
For this to be possible, the ranking function needs to obey a particular type
of monotonicity.
Supported ranking functions include the common sum-of-weights, where answers are compared by the
sum of input-tuple weights, as well as any commutative selective dioid.
Our results extend 
a well-known
unranked-enumeration
dichotomy,
which states that only free-connex CQs are tractable
(under certain hardness hypotheses and for CQs without self-joins).
For this class of queries and with $n$ denoting the size of the input, the data complexity of our ranked enumeration approach 
for the time to the $k^\textrm{th}$ CQ answer 
is $\O(n + k \log k)$,
which is
only a logarithmic factor slower than
the $\O(n + k)$ unranked-enumeration
guarantee.

A core insight of our work is that ranked enumeration for CQs is closely related to 
Dynamic Programming
and
the fundamental task of \emph{path enumeration in a weighted DAG}.
We uncover a previously unknown tradeoff, both 
for this problem and
for CQs, under the lens of combined complexity where the query size is not considered a constant:
one any-$k$ algorithm has lower complexity when the number of returned answers
is small, the other when their number is large. 
This tradeoff is eliminated under a stricter monotonicity property that we define and exploit 
for a \emph{novel algorithm that asymptotically dominates all previously known alternatives},
including Eppstein's algorithm for sum-of-weights path enumeration.
We empirically demonstrate the findings of our theoretical analysis in an experimental study
that highlights the superiority of our approach over the join-then-rank approach that existing database systems follow.

\end{abstract}

\maketitle

\input{content.tex}

\begin{acks}
We are grateful to Hung Q. Ngo for reading drafts of this paper and providing valuable feedback. 
This work was supported in part by %
the Office of Naval Research
(Grant\#: N00014-21-C-1111),
the National Institutes of Health (NIH) under award number R01 NS091421, by
the National Science Foundation (NSF) under award numbers CAREER IIS-1762268
and IIS-1956096,
and by a Google PhD fellowship for author Nikolaos Tziavelis.
Any opinions, findings, and conclusions or recommendations expressed in this paper are those of the authors and do not necessarily reflect the views of the funding agencies.
\end{acks}

\balance

\printbibliography	%

\input{appendix.tex}

\end{document}

%% file: datalabmacros.tex
\usepackage[utf8]{inputenc} 
\usepackage{amsmath}
\usepackage{mathtools}
\usepackage{graphicx} %
\usepackage{url}
\usepackage{enumitem} %

\usepackage{doi}
 \usepackage{bm} %
\usepackage{stmaryrd} %

\usepackage{upgreek} %
\usepackage{bbding} %

\usepackage{listings} %

\usepackage{tabularx} %
\usepackage{booktabs} %
\usepackage{multirow}
\usepackage{hhline} %
\usepackage{diagbox}
\usepackage{pgfplotstable} %

\usepackage{xcolor,colortbl}    %

\usepackage{tikz} %
\usetikzlibrary{matrix}
\usetikzlibrary{calc}
\usetikzlibrary{math}
\usetikzlibrary{arrows.meta,positioning}
\usetikzlibrary{intersections, pgfplots.fillbetween}

\usepackage{easybmat} %

\usepackage{adjustbox}
\def\eox{\unskip\kern 10pt{\unitlength1pt\linethickness{.4pt}$\diamondsuit${}}} %

\usepackage{rotating}		%

\usepackage{xspace} %

\usepackage{subcaption} %
\newcommand{\hide}[1]{}

\usepackage[ruled,noend,linesnumbered]{algorithm2e} %

\DontPrintSemicolon     %
\SetNlSty{}{}{}                %
\SetAlgoInsideSkip{smallskip}   %

\SetAlFnt{\small}			%
\SetAlCapFnt{\small}		%
\SetAlCapNameFnt{\small}

\usepackage{hyperref}   %
\hypersetup{                                                            %
	pdfpagemode=UseNone,               %
    pdfstartview=Fit,
    pdfpagelayout=SinglePage,       %
    bookmarks=false,
    bookmarksnumbered=true,
    bookmarksopen=false,             %
    pdftex=true,
    unicode=true,
    colorlinks=true,       %
    linkcolor=blue,          %
    citecolor=cyan,  %
    filecolor=magenta,      %
     urlcolor=blue           %
}

\usepackage[capitalise,nameinlink]{cleveref} %

\usepackage{aliascnt}  	

\newtheorem{theorem}{Theorem} %

\newaliascnt{corollary}{theorem}

\aliascntresetthe{corollary}

\newaliascnt{example}{theorem}
\newtheorem{example}[example]{Example}
\aliascntresetthe{example}

\newaliascnt{definition}{theorem}
\newtheorem{definition}[definition]{Definition}
\aliascntresetthe{definition}

\newaliascnt{proposition}{theorem}
\newtheorem{proposition}[proposition]{Proposition}
\aliascntresetthe{proposition}

\newaliascnt{lemma}{theorem}
\newtheorem{lemma}[lemma]{Lemma}
\aliascntresetthe{lemma}

\newaliascnt{conjecture}{theorem}

\aliascntresetthe{conjecture}

\newtheorem{observation}{Observation}

\newtheorem{questionW}{Question}
\newtheorem{resultW}{Result}

{\end{itshape}
}
\AtEndPreamble{%
    \hypersetup{colorlinks,
      linkcolor=purple,
      citecolor=blue,
      urlcolor=ACMDarkBlue,
      filecolor=ACMDarkBlue}}

\usepackage{tcolorbox}		%
\tcbuselibrary{breakable,skins}		%

\tcbset{examplestyle/.style={
		enhanced jigsaw,	%
		colback=blue!08,	%
		colframe=blue!08,	%
		arc=3mm,
		boxrule=0pt,		%
		left=1mm,
		right=1mm,
		left skip= 0mm,  %
		right skip= 0mm, %
		top=1mm,		%
		bottom=1mm,		%
		breakable,		%
		parbox = false,		%
		before={\par\pagebreak[0]\vspace{1mm}\parindent=0pt},		
		after={\par\pagebreak[0]\vspace{1mm}\parindent=0pt},				
		bottomrule = 0mm,
		boxsep = 0mm,					%
		topsep at break=0pt,			%
		bottomsep at break=0pt,			%
		pad at break=0mm,
		pad before break=1mm,		
		pad after break=1mm,		
		bottomrule at break=0mm,
		toprule at break=0mm,		
		}}

\usepackage{footnote}

\tcolorboxenvironment{example}{examplestyle}
\newcommand{\resultbox}[1]{
\begin{tcolorbox}[
	enhanced jigsaw,		%
	colback=red!5,
	colframe=red!75!black,	
	arc=0mm,
	left skip=-1mm,
	right skip=-1mm,	
	left=0mm,
	topsep at break=1mm,			%
	right=0mm,
	top=0mm,
	bottom=0mm,		%
	breakable,		%
	parbox = false		%
]
\emph{#1}
\end{tcolorbox}
}

\makeatletter
\DeclareRobustCommand*\uell{\mathpalette\@uell\relax}
\newcommand*\@uell[2]{
  \setbox0=\hbox{$#1\ell$}
  \setbox1=\hbox{\rotatebox{10}{$#1\ell$}}
  \dimen0=\wd0 \advance\dimen0 by -\wd1 \divide\dimen0 by 2
  \mathord{\lower 0.1ex \hbox{\kern\dimen0\unhbox1\kern\dimen0}}
}
\makeatother

\usepackage{marginnote}

\newcommand{\smallsection}[1]{\vspace{2mm}\noindent\textbf{#1.}} %
\newcommand{\introparagraph}[1]{\textbf{#1.}} %

\renewcommand{\epsilon}{\varepsilon} %

\newcommand{\define}{:=} %

\newcommand{\datarule}{{\,:\!\!-\,}} %
\renewcommand{\vec}[1]{\boldsymbol{\mathbf{#1}}}

\newcommand{\N}{\mathbb{N}} %
\newcommand{\R}{{\mathbb{R}}} %
\newcommand{\bigO}{{\mathcal{O}}} %
\renewcommand{\O}{{\mathcal{O}}} %

\usepackage{scalerel}
\usepackage{tikz}
\usetikzlibrary{svg.path}

\definecolor{orcidlogocol}{HTML}{A6CE39}
\tikzset{
  orcidlogo/.pic={
    \fill[orcidlogocol] svg{M256,128c0,70.7-57.3,128-128,128C57.3,256,0,198.7,0,128C0,57.3,57.3,0,128,0C198.7,0,256,57.3,256,128z};
    \fill[white] svg{M86.3,186.2H70.9V79.1h15.4v48.4V186.2z}
                 svg{M108.9,79.1h41.6c39.6,0,57,28.3,57,53.6c0,27.5-21.5,53.6-56.8,53.6h-41.8V79.1z M124.3,172.4h24.5c34.9,0,42.9-26.5,42.9-39.7c0-21.5-13.7-39.7-43.7-39.7h-23.7V172.4z}
                 svg{M88.7,56.8c0,5.5-4.5,10.1-10.1,10.1c-5.6,0-10.1-4.6-10.1-10.1c0-5.6,4.5-10.1,10.1-10.1C84.2,46.7,88.7,51.3,88.7,56.8z};
  }
}

\RequirePackage{etoolbox}
\DeclareRobustCommand\orcidicon[1]{\href{https://orcid.org/#1}{\mbox{\scalerel*{
\begin{tikzpicture}[yscale=-1,transform shape]
\pic{orcidlogo};
\end{tikzpicture}
}{|}}}}

%% file: individualmacros.tex
\newenvironment{proofoutline}{%
  \renewcommand{\proofname}{Proof Outline}\proof}{\endproof}

\def\polylog{\operatorname{polylog}} %

\newcommand{\serialw}{w_{\textrm{sd}}}

\newcommand{\minel}{\min}
\newcommand{\smonotone}{s-monotone\xspace}
\newcommand{\smonotonicity}{s-monotonicity\xspace}
\newcommand{\ssmonotone}{ss-monotone\xspace}
\newcommand{\ssmonotonicity}{ss-monotonicity\xspace}

\newcommand{\diam}{\texttt{diam}}
\newcommand{\sparseBMM}{\textsc{sparseBMM}}
\newcommand{\hyperclique}{\textsc{Hyperclique}}
\newcommand{\witness}{\texttt{wit}} 
\newcommand{\out}{r}

\newcommand{\dom}{{\mathtt{dom}}}

\newcommand{\nodenum}{N}
\newcommand{\SortedSuf}{\mathtt{SortedSuff}}
\newcommand{\aggr}{+}
\newcommand{\aggrsum}{\sum}

\newcommand{\Choices}{\mathtt{Choices}}
\newcommand{\suc}{\mathtt{suc}}
\newcommand{\Suc}{\mathtt{Suc}}
\newcommand{\stages}{\ell}

\newcommand{\sgiter}{i}

\newcommand{\Dec}{E}
\newcommand{\DecR}{\mathbb{E}}
\newcommand{\Sset}{V}
\newcommand{\SsetR}{\mathbb{V}}

\newcommand{\concat}{\circ}
\newcommand{\Ch}{C}
\newcommand{\parent}{pr}
\newcommand{\chnum}{d}
\newcommand{\solW}{\pi}
\newcommand{\sol}{\Pi}
\newcommand{\weight}{w}	%
\newcommand{\tree}{\concat}	

\newcommand{\calH}{\mathcal{H}}

\newcommand{\QUICK}{\textsc{Quick}\xspace}
\newcommand{\PSQL}{\textsc{PSQL}\xspace}
\newcommand{\SYSX}{\textsc{System X}\xspace}

\newcommand{\BATCH}{\textsc{JoinFirst}\xspace}
\newcommand{\LAZY}{\textsc{Lazy}\xspace}
\newcommand{\EAGER}{\textsc{Eager}\xspace}

\newcommand{\MIN}{\textsc{All}\xspace}
\newcommand{\HEAP}{\textsc{Take2}\xspace}

\newcommand{\ANYKREC}{\textsc{anyK-rec}\xspace}
\newcommand{\ANYKPART}{\textsc{anyK-part}\xspace}
\newcommand{\ANYKPARTP}{\textsc{anyK-part+}\xspace}

\newcommand{\Bitcoin}{\textsc{Bitcoin}\xspace}
\newcommand{\Twitter}{\textsc{Twitter}\xspace}
\newcommand{\Friendship}{\textsc{Friendship}\xspace}
\newcommand{\Foodweb}{\textsc{Foodweb}\xspace}

\usepackage{booktabs} %
\graphicspath{{./figs/}}
\usepackage{numprint}
\usepackage{multirow}
\usepackage{makecell}

\usepackage{balance}
\usepackage{tikz}
\usepackage{tikz,tikz-qtree,tikz-qtree-compat,tikz-3dplot}
\usetikzlibrary{shapes,decorations.markings,arrows.meta,fit}

\definecolor{dkgreen}{rgb}{0,0.6,0}
\newcommand{\cc}{olive}
\newcommand{\algocomment}[1]{\textcolor{\cc}{{//#1}}}

\newcommand{\TTL}{\ensuremath{\mathrm{TTL}}\xspace}
\newcommand{\TT}{\ensuremath{\mathrm{TT}}}
\newcommand{\Del}{\ensuremath{\mathrm{Delay}}}
\newcommand{\MEM}{\ensuremath{\mathrm{MEM}}}

\newcommand{\Cand}{\mathtt{Cand}} %
\newcommand{\free}{\texttt{free}}

\newcommand{\Union}{\mathtt{Union}} %

\newcommand{\minweight}{min-weight-projection\xspace} 
\newcommand{\Minweight}{Min-weight-projection\xspace} 
\newcommand{\allweights}{all-weight-projection\xspace}
\newcommand{\Allweights}{All-weight-projection\xspace}

\newcommand{\varset}{\textup{\texttt{var}}}

\newcommand{\0}{\bar{0}}
\newcommand{\1}{\bar{1}}

\newcommand{\subw}{\textsf{subw}\xspace}

\makeatletter
\newcommand{\markZwicky}[1][]{\pgfutil@ifnextchar({\mark@Zwicky{#1}}{\mark@Zwicky{#1}()}}
\def\mark@Zwicky#1(#2)#3{%
   \tikz[every Zwicky picture,#1]{%
     \node[every Zwicky node,draw=none,inner sep=+\z@,outer sep=+\z@] {#3};
     \def\tikz@Mark@name{#2}%
     \ifx\tikz@Mark@name\pgfutil@empty\else
       \tikzset{every Zwicky node/.append style={name={#2}}}%
     \fi
     \node[every Zwicky node,overlay] {\phantom{#3}};
   }%
}
\newcommand{\tikzZwicky}[1][]{%
  \def\tikz@Zwicky@args{#1}%
  \let\tikz@Zwicky@list\pgfutil@gobble
  \let\tikz@Zwicky@first\pgfutil@empty
  \pgfutil@ifnextchar(\tikz@Zwicky@collect\tikz@Zwicky@finish
}
\def\tikz@Zwicky@collect(#1){%
  \ifx\tikz@Zwicky@first\pgfutil@empty
    \edef\tikz@Zwicky@first{#1}%
  \else
    \edef\tikz@Zwicky@list{\tikz@Zwicky@list,#1}%
  \fi
  \pgfutil@ifnextchar(\tikz@Zwicky@collect\tikz@Zwicky@finish
}
\def\tikz@Zwicky@finish{%
  \tikz[remember picture,overlay]
    \draw[every Zwicky connector,/expanded=\tikz@Zwicky@args]
      (\tikz@Zwicky@first) [/expanded={@Zwicky@list/.list={\tikz@Zwicky@list}}] [every Zwicky connect finish/.try];
}
\pgfkeys{/expanded/.code={\edef\pgfkeys@temp{{#1}}\expandafter\pgfkeysalso\pgfkeys@temp}}
\makeatother
\tikzset{
  @Zwicky@list/.style={insert path={to[every Zwicky connector how/.try] (#1)}},
  every Zwicky picture/.style={
    baseline,
    remember picture,
  },
  every Zwicky node/.style={
    remember picture,
    anchor=base,
    inner sep=+2pt
  },
  every Zwicky connector/.style={
    ultra thick,
    red!80!black,
    draw opacity=.5,
    line cap=round,
    line join=round
  }
}

\newtheorem{hypothesis}{Hypothesis}

%% file: content.tex
\section{Introduction}

Joins are an essential building block of queries in relational and graph databases.
They are notoriously critical for performance, as they can
produce huge intermediate or final results.
Enumeration~\cite{bagan07constenum,carmeli21ucqs,DBLP:journals/sigmod/Segoufin15}
is a query-answering paradigm that circumvents this by returning
the answers as a stream
as quickly as possible, even if the full output is too large to compute. 
Among other key problems at the forefront of recent research~\cite{tziavelis22tutorial},
such as 
worst-case optimal joins~\cite{ngo2018worst}, 
(hyper)tree decompositions~\cite{GottlobGLS:2016}, 
and factorized representations~\cite{olteanu16record},
\emph{ranked enumeration} has been identified as an important open problem \cite{dagstuhl19enumeration}:
enumerate answers in an order determined by a given ranking function.
This augments the enumeration paradigm by imposing an order on the output; the top-ranked answers are returned first, 
followed by lower-ranked ones in quick succession.

Ranked enumeration shares a similar motivation with top-$k$ queries~\cite{ilyas08survey}, yet has two crucial differences:
(1) For top-$k$, the number of desired answers, $k$, is provided in advance.
In contrast, a ranked enumeration (or ``\emph{any-$k$}'') algorithm needs to perform well
for all possible values of $k$.
Any-$k$ algorithms can conceptually be seen as a fusion of
top-$k$~\cite{ilyas08survey} and
anytime algorithms \cite{Zaimag96}
that gradually improve their results over time.
(2) Algorithms for top-$k$ joins, including the celebrated
Threshold Algorithm~\cite{fagin03},
were developed for a ``middleware'' cost model that 
accounts only for accesses to external data sources,
but not for memory accesses and intermediate-result size~\cite{tziavelis20tutorial}.
Our goal are strong guarantees 
for time and space complexity in the RAM model of computation, for every value $k$.
We refer to those as $\TT(k)$ (i.e., Time-To-$k$) and $\MEM(k)$ respectively.

\begin{example}[4-path query]\label{ex:4path}
Let $R(A, B, W_1)$, $S(B, C, W_2)$,
$T(C, D, W_3)$, and $U(D, E, W_4)$
be relations where the last column stores tuple weights.
The following query joins the relations and orders the answers by ascending sums of weights.
\begin{verbatim}
  SELECT   *, R.W1 + S.W2 + T.W3 + U.W4 AS SUMW
  FROM     R, S, T, U
  WHERE    R.B = S.B AND S.C = T.C AND T.D = U.D
  ORDER BY SUMW ASC
\end{verbatim}
If the relations have $n$ tuples each, 
the query can produce $\Omega(n^3)$ answers
in the worst case~\cite{AGM}.
Thus, any ``join-then-rank'' algorithm that first joins the relations and
then applies the ranking (e.g., by sorting)
needs $\Omega(n^3)$ time to find even the first answer.
This is due to the sheer number of possible answers that have to be compared
and is true for any join strategy, such
as sort-merge, hash-join, or the optimal Yannakakis algorithm~\cite{DBLP:conf/vldb/Yannakakis81}. 
Our any-$k$ algorithms push down the ranking into the join
and achieve $\TT(k) = \O(n + k \log k)$.
This means that the top-ranked answer is returned in linear time,
and thereafter, for $k$ returned answers, 
the additional cost is the same as sorting $k$ elements.
One may wonder how the situation changes if instead of a 4-path query, we have a 4-cycle (i.e., `\texttt{AND R.A = U.E}'),
or if we project away the A attribute,
or if we replace the $\textup{sum}$ ($+$) with $\max$.
Our work answers all these questions in a principled way.
\end{example}

\introparagraph{Contributions}
(1) 
We develop a theory of \emph{ranked enumeration for Dynamic Programming} (DP) that applies to CQs.
It reveals the deeper common foundations between isolated prior works that 
only partially address the problem: 
$k$-shortest paths~\cite{eppstein1998finding,jimenez99shortest,lawler72}, 
ranked retrieval of graph patterns~\cite{chang15enumeration,yang2018any}, and 
ranked enumeration for CQs~\cite{deep21,KimelfeldS2006}.
While interesting in its own right, this general framework allows us to 
apply our ranked enumeration techniques not only to these problems but also to
any problem whose top-ranked solution can be found via DP, such as DNA sequence alignment~\cite{dpv08book}
or Viterbi decoding~\cite{seshadri94viterbi}.
The ranking function needs to have certain properties to allow efficient ranked enumeration,
which is not surprising given that DP does not apply to all optimization problems.
We identify \emph{subset-monotonicity} 
(\smonotonicity)
as a sufficient condition 
and show how it is related to
algebraic structures called commutative selective dioids~\cite{GondranMinoux:2008:Semirings}, 
which are a special case of the more well-known semirings~\cite{green07semirings}.

(2) For a large class of CQs, our any-$k$ algorithms achieve $\TT(k) = \O(n + k \log k)$,
in data complexity~\cite{DBLP:conf/stoc/Vardi82}, where $n$ is the size of the database
and query size is treated as a constant.
This is close to the $\Omega(n + k)$
lower bound and the additional logarithmic factor is expected since the 
$k$ answers are returned sorted.
We mainly focus on \emph{acyclic} CQs because
cyclic CQs can be handled by first decomposing them to acyclic ones~\cite{GottlobGLS:2016} and then applying our techniques.
Since acyclic CQs can have a more general tree structure (compared to the more restrictive path structure of standard DP),
we first extend our Dynamic Programming framework
to a class of problems that we call \emph{Tree-DP (T-DP)}. 
and apply it to \emph{full} acyclic CQs.
When the query has projections, an established 
unranked-enumeration
dichotomy~\cite{bagan07constenum} states that, 
under plausible assumptions in fine-grained complexity, 
the only (self-join-free) CQs 
that admit $\TT(k) = \O(n + k)$ are those that are free-connex.\footnote{The original dichotomy is phrased in terms of preprocessing and delay.}
We establish that the frontier of tractability remains the same with an \smonotone ranking function
(minus logarithmic factors);
free-connex CQs can be handled with $\TT(k) = \O(n + k \log k)$ and no other (self-join-free) CQ admits this running time.

(3) We compare different any-$k$ algorithms in a more refined analysis of \emph{combined complexity},
where query size is not treated as a constant, and propose 
a new algorithm that asymptotically dominates all others under a stronger monotonicity property. 
In particular, we find that the best algorithm based on the Lawler-Murty procedure~\cite{lawler72},
which we refer to as \ANYKPART,
is in general asymptotically better 
than approaches based on the Recursive Enumeration Algorithm (REA) \cite{deep21,jimenez99shortest},
which we refer to as \ANYKREC.
However, \ANYKREC smartly reuses comparisons 
and this can pay off as $k$ increases:
there exist inputs for which it produces the full
sorted output with lower time complexity than \ANYKPART, and even faster than sorting the materialized output.
We propose \ANYKPARTP, an algorithm that combines the best of both worlds.
However, it requires a stricter monotonicity property from the ranking function,
which we call \emph{strong-subset-monotonicity}
(\ssmonotonicity).
We give examples of ranking functions with this property,
and show
how it follows from elementary algebraic properties like cancellation~\cite{GondranMinoux:2008:Semirings}.
\ANYKPARTP is asymptotically faster than Eppstein's algorithm~\cite{eppstein1998finding} for path enumeration in a DAG, 
assuming that the returned paths are in an explicit listing representation
(i.e., the path size is proportional to the number of its edges)
instead of an implicit representation that Eppstein leverages.

(4) 
We provide the \emph{first empirical study} that directly compares these ranked-enumeration algorithms.
We show a tradeoff between \ANYKPART and \ANYKREC, 
with each one winning for different queries or values of $k$.
As in theory, \ANYKPARTP combines their best
qualities and is in most cases close or better than the fastest of the other two.
Importantly, our study clearly illustrates the advantage of any-$k$ for queries with large join output
over the join-then-rank approach.
Existing database systems follow the latter, and our algorithms outperform them by orders of magnitude.

\introparagraph{Conference version}
This article is an extended version of an earlier conference paper \cite{tziavelis20vldb},
significantly extending its scope and depth.
($i$) First, we propose the \ANYKPARTP algorithm that combines the best features of the two previous algorithms (\Cref{sec:part_plus}).
($ii$) Second, we include the study of CQs with projections (\Cref{sec:projections}) and
give the corresponding dichotomy result.
($iii$) Third, we provide an analysis of the supported ranking functions, their properties, and how they relate to other monotonicity definitions and algebraic structures (\Cref{sec:ranking_functions}).
($iv$) Finally, this article makes several other additions, such as a more extensive experimental study that better evaluates the scalability of the approach
and includes a comparison with a commercial database system,
a more detailed review of the literature that includes
a comparison to unranked and lexicographic order enumeration,
pseudocode for the algorithms, and proofs of correctness.

The project web page at \url{https://northeastern-datalab.github.io/anyk/}
contains code, slides, videos, and further information.

\section{Preliminaries}
\label{sec:prelim}

We use
$[m]$ to denote the set of integers $\{1, \ldots, m\}$
and $[m]_0$ for $\{0, \ldots, m\}$.

\subsection{Basic Notions}
\label{sec:def}

\introparagraph{Graphs and Hypergraphs}
A directed graph $G(V, E)$ is weighted if a weight function $w: E \rightarrow W$ assigns weights from a domain $W$ to 
its edges.
The graph size $|G| = |V| + |E|$ is the number of its nodes and edges.
A \emph{path} (or walk) $p$ of length $\lambda$ is a sequence of nodes $\langle v_0, \ldots, v_\lambda \rangle$ such that
$(v_{i-1}, v_i) \in E, i \in [\lambda]$.
Note that the nodes or edges do not have to be distinct.
Similarly, a path of length $\lambda$ in a hypergraph $\calH(V, E)$
is a sequence $\langle v_0, e_0, v_1, \ldots, e_\lambda, v_\lambda \rangle$
such that
$v_{i-1}, v_i \in e_i, e_i \in E, \forall i \in [\ell]$.
The distance between two nodes $u, v$ is the minimum length of any path with $u = v_0, v = v_\lambda$.
The \emph{diameter} of a hypergraph is the maximum distance between any pair of nodes.
If we fix a source node $s$ and a target node $t$ we call a \emph{prefix} of a node $v_{i}$ any path
$\langle s, v_1, \ldots, v_{i} \rangle$ to $v_{i}$
that starts at $s$ and \emph{suffix} of $v_{i}$ any path 
$\langle v_i, v_{i+1}, \ldots, v_{\lambda}, t \rangle$ from $v_{i}$
that ends in $t$.
We use $\concat$ as a concatenation operator for paths.
Given a prefix and a suffix,
$\langle s, v_1, \ldots, v_{i} \rangle \concat \langle v_i, v_{i+1}, \ldots, v_{\lambda}, t \rangle$ is
$\langle s, v_1, \ldots, v_{\lambda}, t \rangle$.
We also write $v_i \concat \langle v_{i+1}, \ldots, v_{\lambda}, t \rangle$ for some $(v_i, v_{i+1}) \in E$ to mean $\langle v_i, v_{i+1}, \ldots, v_{\lambda}, t \rangle$.
When $s$ and $t$ are fixed and it is clear from the context, we sometimes omit $s, t$ from prefixes or suffixes, i.e., we write $\langle v_1, \ldots, v_\lambda \rangle$ 
instead of $\langle s, v_1, \ldots, v_\lambda, t \rangle$.

\introparagraph{Conjunctive Queries (CQs)}
A CQ $Q$ is a first-order formula 
$\exists \vec Y (R_1(\vec{X_1}) \wedge \ldots \wedge R_\ell(\vec{X_\ell}))$, 
written as
$Q(\vec Z) \datarule R_1(\vec{X_1}) ,\ldots, R_\ell(\vec{X_\ell})$ 
in Datalog notation, 
where $\vec{Y}, \vec{Z}, \vec{X}_i, i \in [\ell]$ are lists of {variables} where
$\vec{Y} \subseteq \bigcup_{i \in [\ell]} \vec{X}_i$ and
$\vec{Z} = \bigcup_{i \in [\ell]} \vec{X}_i \setminus \vec Y$
if interpreted as sets.
Each \emph{atom} $R_i(\vec{X_i})$ refers to a relation with $|\vec{X_i}|$ columns (or attributes).
The variables in $\mathbf{Z}$ 
are called \emph{free} and denoted by $\free(Q)$,
while the rest of the variables $\mathbf{Y}$ are called \emph{existential}.
A \emph{Boolean} CQ has no free variables (i.e., $\mathbf{Z} = \emptyset$)
and only asks for the satisfiability of the formula.
In a \emph{full} CQ, all variables are free, i.e., $\mathbf{Z} = \bigcup_{i \in [\ell]} \vec{X}_i$.
The occurrence of the same variable in different atoms encodes an \emph{equi-join} condition,
implying that the values of the corresponding attributes in a query answer need to be equal. 
Different atoms can refer to the same relation, in which case we have a
\emph{self-join}.
A \emph{self-join-free} query has no self-joins.
Without loss of generality, for our algorithms we
can assume that (1) CQs are self-join-free 
since tables can be copied,
and (2) selections on individual relations (like $R(x,1)$ or $R(x,x)$)
have been removed in a linear-time preprocessing step.

\introparagraph{Query semantics}
CQs are evaluated over a database $D$ of relations that
draw values from a domain $\dom$, such as $\N$.
An \emph{output tuple} or query \emph{answer} to $Q$ 
is a mapping $\free(Q) \to \dom$
such that the first-order formula is satisfied. 
The set of all query answers is $Q(D)$, and we use $q \in Q(D)$ for a query answer.
A semantic evaluation strategy
to compute $Q(D)$ is to
($i$) materialize the Cartesian product of the $\ell$ relations,
($ii$) select tuples that satisfy the equi-joins,
and ($iii$) project on the $\vec{Z}$ attributes. 
A \emph{witness}~\cite{buneman01provenance} 
of a query answer is a 
size-$\ell$ vector of input tuples, 
one from the relation of each atom, 
that can join to produce the answer.
We denote the set of witnesses of an answer $q$ by $\witness(q)$.
For full CQs, we can equivalently represent an answer by its unique witness.

\introparagraph{Join trees}
A CQ is associated with a hypergraph where variables form the nodes and
atoms form the hyperedges.
We say that a CQ is \emph{acyclic} if its hypergraph is alpha-acyclic~\cite{baron16acyclic},
which means that we can construct a \emph{join tree}.
A join tree is a rooted tree where the atoms form the nodes
and the \emph{running intersection property} holds:
For every variable $x$, all tree nodes 
containing $x$
form a connected subtree.
The acyclicity of a CQ can be tested, and a corresponding join tree can be constructed in
linear time in the query size by the 
GYO reduction \cite{tarjan84acyclic,yu79gyo}.
A CQ that is not acyclic is called \emph{cyclic}.
The diameter $\diam(Q)$ of a CQ $Q$ is the diameter of its hypergraph.

\begin{example}[$\ell$-path and $\ell$-cycle queries]
\label{ex:path_cycle}
Let $E(\texttt{FROM}, \texttt{TO})$ 
be a relation that stores the directed edges of a graph.
The following $\ell$-path CQ computes paths of length $\ell$ between any pair of nodes in the graph:
$Q_{P\stages }(\vec x) \datarule E(x_0, x_1), E(x_1, x_2), \ldots, E(x_{\stages-1}, x_{\stages})$.
These CQs are acyclic; to construct a join tree, we can organize the atoms in a path with $E(x_0, x_1)$ as the root.

Length-$\stages$ cycles can be expressed as
$Q_{C\stages }(\vec x) \datarule E(x_1, x_2), E(x_2, x_3), \ldots, E(x_\stages, x_1)$.
These CQs are cyclic, and no join tree exists; 
for example, the path with $E(x_1, x_2)$ as the root does not have the
running intersection property;
the first and the last atom both contain $x_1$ but are disconnected.

\hide{	The corresponding SQL queries are
	\begin{verbatim}
	    SELECT R.V1, R.V2, S.V2, T.V2, W.V2
	    FROM R, S, T, W
	    WHERE R.V2=S.V1 AND S.V2=T.V1 AND T.V2=W.V1
	
	    SELECT R.V1, R.V2, S.V2, T.V2
	    FROM R, S, T, W
	    WHERE R.V2=S.V1 AND S.V2=T.V1
	      AND T.V2=W.V1 AND W.V2 = R.V1
	\end{verbatim}
}
\end{example}

\subsection{Ranked Enumeration}
\label{sec:ranked_enumeration}

Ranked enumeration assumes a given \emph{ranking function} that orders the query answers
by mapping them to a domain $W$ equipped 
with a total order $\preceq$.
We denote by $\minel(S)$ the smallest element in $S \subseteq W$ according to $\preceq$, where $S$ is allowed to be a multiset.

\introparagraph{Weight aggregation}
In this work, we focus on \emph{aggregate ranking functions} where weights are assigned to input tuples and
the mapping to $W$ is done by aggregating the weights of input tuples forming the witness of an output tuple.
A common example is the \emph{sum-of-weights} case as in \cref{ex:4path}:
Real-valued weights are assigned to input tuples and
the weight of a (full) CQ answer is computed 
by adding up 
the weights of tuples in its witness.
More generally, we are given an input-weight function $w_I$ that associates each input tuple with some weight
in $W$
and an aggregate function $w_A: \N^W \rightarrow W$ 
that returns a unique weight from a multiset of elements in $W$.\footnote{For simplicity, we assume that the input weights and the answer weights have the same domain, but a generalization is straightforward.}
Aggregate functions are not sensitive to the order that the input weights are provided, 
which is captured by the fact that their input is a 
multiset~\cite{jesus15aggregation}.
We further assume that $w_A(\emptyset) = \minel(W)$.

For CQs with projections, 
a query answer can have multiple witnesses with differing weights.
In that case, we adopt the \emph{\minweight} semantics: its weight is the minimum (according to $\preceq$) among the witnesses.
Formally, the weight of $q \in Q(D)$ is
$w(q) = \minel_{(t_1,\ldots, t_\stages) \in \witness(q)}(w_A(\{w_I(t_i) | i \in [\ell]\}))$.
In \Cref{sec:projections}, we discuss other possible semantics.
For the remainder of this paper,
we refer to aggregate ranking functions simply as
ranking functions,
and, when it is clear from the context,
we may use $w$ instead of $w_I$ or $w_A$.

\introparagraph{Monotonicity}
Let $\uplus$ be the multiset union operator. Our algorithms exploit the following monotonicity properties of common ranking functions:

\begin{definition}[Subset-Monotonicity~\cite{KimelfeldS2006}]
\label{def:smonotone}
A ranking function $w$ is \smonotone if 
$w_A(Y_1) \preceq w_A(Y_2) \Rightarrow w_A(X \uplus Y_1) \preceq w_A(X \uplus Y_2)$ for all $X, Y_1, Y_2$
$\in \N^W$.
\end{definition}

Intuitively, \smonotonicity allows us to infer the ranking of complete solutions from the ranking of
partial solutions.
Next, we define a stronger notion that allows us to
rank complete solutions from
other related complete solutions.
This stronger property allows us to develop an algorithm with the best known asymptotic guarantees in \cref{sec:part_plus}.

\begin{definition}[Strong-Subset-Monotonicity]
\label{def:strongsubset}
A ranking function $w$ is \ssmonotone if 
$w_A(X_1 \uplus Y_1) \preceq w_A(X_1 \uplus Y_2)
\wedge
w_A(X_1) \preceq w_A(X_2)
\Rightarrow
w_A(X_2 \uplus Y_1) \preceq w_A(X_2 \uplus Y_2)$ for all 
$X_1, X_2, Y_1, Y_2 \in \N^W$.
\end{definition}

Notice that \ssmonotonicity implies \smonotonicity by setting $X_1 = \emptyset$
(recall that $w_A(\emptyset) = \minel(W)$).
The common sum-of-weights case satisfies both properties.
An example of a ranking function
that is \smonotone but not
\ssmonotone
is
given in \cref{ex:weaksubsetmonotone}.
Also, note that the condition $w_A(X_1) \preceq w_A(X_2)$ in our definition does not restrict the set of ranking functions that satisfy
\ssmonotonicity, but expands it.
For example, min/max ranking is \ssmonotone:
Although $\max \{2, x\} \leq \max \{1, x \}$ is not true for all values of $x$,
it is true for $x \geq 3$.

\introparagraph{Incremental aggregation}
Another issue related to the ranking function is whether we can compute $w_A(X \uplus Y)$
in $\O(1)$ given two partial weights $w_A(X)$ and $w_A(Y)$ 
(i.e., without performing the aggregation from scratch).
Such an aggregate ranking function has been called \emph{algebraic}~\cite{DBLP:journals/datamine/GrayCBLRVPP97}
or \emph{decomposable}~\cite{jesus15aggregation}.
We adopt the former term.
Examples include SUM, COUNT, MIN, and AVERAGE.
In contrast,
\emph{holistic} functions like MEDIAN require $\omega(1)$ space for intermediate results.

\introparagraph{Problem definition}
The central problem in this paper is the following:

\begin{definition}[Ranked enumeration for CQs]
\label{def:ranked_enum_cqs}
Given a CQ $Q$ over a database $D$ and a 
ranking function,
ranked enumeration returns
the query answers $Q(D)$ one-at-a-time
in ascending $\preceq$ order without duplicates.
\end{definition}

As we will see, the above problem 
can be modeled as an instance of ranked enumeration for Dynamic Programming
(DP),
which is in turn intimately related to the problem of \emph{path enumeration} in a DAG.
Thus, we also consider the problem of ranked enumeration for such paths.
The ranking function there
is defined naturally
over the edges forming a path $p$:
$w(p) = w_A(\{ w_I(e) | e \in p\})$.
Crucially, we consider the variant of the problem where paths have to be enumerated \emph{explicitly}
by returning the list of edges of each path.
It is known that just enumerating the weights
(or the weights and an implicit pointer representation of the paths)
can remove the dependency on the path length from the running time, as shown in the seminal work of Eppstein~\cite{eppstein1998finding}.

\begin{definition}[Ranked enumeration of paths]
\label{def:ranked_enum_paths}
Given a weighted DAG $G$ with
a source node $s$ and target node $t$, and a ranking function,
(explicit) ranked enumeration
returns the paths from $s$ to $t$ one-at-a-time
in ascending $\preceq$ order 
represented as lists of edges
and without duplicates.
\end{definition}

In \Cref{sec:ranking_functions}, we delve deeper into issues related to the ranking function including the relationship of monotonicity properties to algebraic structures.
Until then, we phrase our algorithms using the sum-of-weights model, which is commonly used in standard DP formalism and $k$-shortest paths.

\subsection{Complexity Measures and Hypotheses}
\label{sec:complexity_measures}

\introparagraph{Model of computation}
We consider in-memory computation only and
analyze all algorithms in the standard
Random Access Machine (RAM)
model with uniform cost measure
where every operation or memory access costs $\O(1)$.
We do not assume any given indexes or sorting in the database relations.
In line with previous work~\cite{berkholz19submodular,carmeli21ucqs,GottlobGLS:2016,ngo2018worst},
we also assume that it is possible to build a lookup table in linear
time to support tuple lookups in constant time. 
In practice, this is virtually guaranteed by hashing.

\introparagraph{Problem parameters}
We use $n$ to refer to the maximum cardinality
of any database relation referenced in $Q$
and $\out$ for $|Q(D)|$.
The \emph{size} of a CQ is measured in terms of $\ell$ (i.e., the number of atoms) 
and also the maximum arity among atoms, which we denote by $\alpha$.
We use two different notions of complexity for CQs~\cite{DBLP:conf/stoc/Vardi82}.
The first is \emph{data complexity}, where the size of the query
is treated as a constant.
The second is the more detailed \emph{combined complexity}, where query size is
treated as a variable.
We apply the latter analysis to full acyclic CQs, revealing interesting differences between our algorithms
and aiding in the explanation of our experimental results.
For this analysis, we further assume that the ranking function is algebraic.
This is not a strict requirement and
our analysis can easily be extended to cases where this does not hold by keeping track of all the elements even after aggregating them.

\introparagraph{Measures of success}
We measure the complexity of enumeration by
the time required until the $k^\textrm{th}$ answer 
is returned ($\TT(k)$
for Time-To-$k$)
for all values of $k$.
For a lower bound, note that it takes
$\Omega(n)$ just to look at each input tuple and $\Omega(k)$ to return $k$ output tuples.
Sorting $k$ independent items takes $\Omega(k \log k)$ with a comparison-based algorithm.
We denote the special case of the \emph{Time-To-Last} $\TT(\out)$ by $\TTL$.
Additionally, we measure the space complexity $\MEM(k)$.

\begin{figure*}[tb]
\centering
\hfill
\begin{subfigure}[t]{.45\linewidth}
    \centering
    \includegraphics[height=3.5cm]{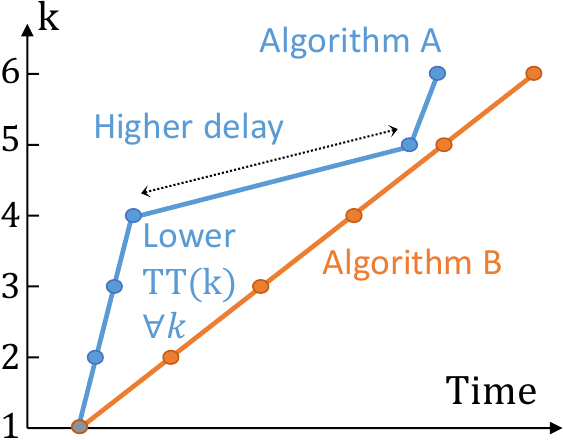}
    \caption{}
    \label{fig:ttk_vs_delay_1}
    \end{subfigure}%
\hfill
\begin{subfigure}[t]{.45\linewidth}
    \centering
    \includegraphics[height=3.5cm]{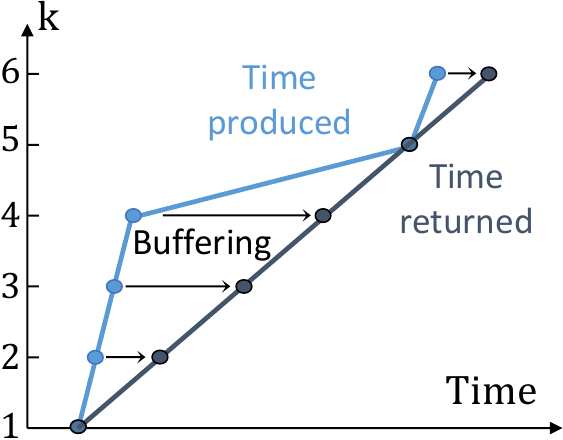}
    \caption{}
    \label{fig:ttk_vs_delay_2}
\end{subfigure}
\hfill
\caption{Time-to-$k$ ($\TT(k)$) is a more practical measure of success than delay.
\Cref{fig:ttk_vs_delay_1} illustrates how an algorithm $A$ with minimal $\TT(k)$ can return each answer faster than an algorithm $B$ even if the maximum delay of $A$ is higher (between answers 4-5 in the figure).
\Cref{fig:ttk_vs_delay_2} shows that the maximum delay of an algorithm can be lowered by buffering the answers and returning them later, but this slows down the algorithm.}
\label{fig:ttk_vs_delay}
\end{figure*}

\introparagraph{Time-to-$k$ versus delay}
The goal of enumeration is to provide each answer as quickly as possible,
which is why we adopt $\TT(k)$ as our measure of success.
A large body of work on enumeration
\cite{bagan07constenum,DBLP:journals/mst/CarmeliK20,idris20dynamic_theta,olteanu15dtrees,DBLP:journals/sigmod/Segoufin15}
has instead focused on measuring the \emph{delay} between answers after a preprocessing phase, typically striving for constant delay after linear-time preprocessing $\O(n)$
(in data complexity).
This is desirable because it guarantees that every $k^\textrm{th}$ answer is returned in time
$\TT(k) = \O(n + k)$, which is optimal (without ranking).
However, low delay is \emph{sufficient but not necessary} to achieve low $\TT(k)$.
As \Cref{fig:ttk_vs_delay_1} shows, a lower delay does not necessarily imply a faster algorithm.

Given enough buffer space (not necessarily main memory),
the delay of an algorithm 
can be made uniform across all $k$ values, e.g., by buffering the answers
and returning them at regular intervals.
We illustrate this in \Cref{fig:ttk_vs_delay_2}.
Techniques in that spirit have been developed~\cite{capelli19delay,carmeli21ucqs,deep21projections,deep22ranked} to guarantee low delay for algorithms that already have low $\TT(k)$.
From a practical point of view, this is usually undesirable~\cite{capelli21delay,carmeli21ucqs} because it slows down the algorithm.
Aiming for low delay might still be relevant if (1) a downstream application requires uniform interarrival times and 
(2) at the same time
there is not enough space to buffer the answers. 
We are currently not aware of any such scenario in practice.

\introparagraph{Hardness hypotheses}
Works in the wider area of enumeration~\cite{bagan07constenum,carmeli21ucqs,carmeli23direct,brault13thesis}
have used certain hypotheses on the hardness of problems to prove lower bounds 
and achieve dichotomies for CQs.
In this work, we extend some of these results to the case of ranked enumeration using the same set of hypotheses.
We list these below and we refer the reader to Berkholz et al.~\cite{Berkholz20tutorial} for an excellent discussion on their plausibility.

\begin{hypothesis}[\sparseBMM{}~\cite{Berkholz20tutorial}]
	Two Boolean matrices $A$ and $B$, represented as lists of non-zeros, 
	cannot be multiplied in time $m^{1+o(1)}$\footnote{An example of a function that is in $m^{1+o(1)}$ is $m^{1+1/m}$. Note that this is higher than $\O(m \polylog m)$.}, 
	where $m$ is the number of non-zeros in $A$, $B$, and $AB$.
\end{hypothesis}

The hypothesis states that the above problem cannot
be solved in near-linear time in input and output size.
As Berkholz et al.~\cite{Berkholz20tutorial} argue, this is a weaker (i.e., more plausible) assumption than 
the original assumption of Bagan et al.~\cite{bagan07constenum} that 
the problem cannot be solved in quadratic time in size.

\begin{hypothesis}[\hyperclique{}~\cite{abboud14conjectures,DBLP:conf/soda/LincolnWW18}]
For every $k \geq 2$, there is no
$O(m \polylog m)$ algorithm to decide the existence of a
$(k{+}1,k)$-hyperclique in a $k$-uniform hypergraph with $m$ hyperedges,
where a \emph{$(k{+}1,k)$-hyperclique} is a set of $k{+}1$ vertices
such that every subset of $k$ elements is a hyperedge
and a $k$-uniform hypergraph is one where all hyperedges have size $k$.
\end{hypothesis}

For $k=2$, \hyperclique{} says that detecting a triangle in a graph cannot be
done in quasilinear time.
The fastest known algorithm for triangle detection~\cite{Alon1997} uses fast matrix multiplication
and runs in time $\Omega(n^{4/3})$, even if the matrix multiplication exponent is $2$, that is, the lowest possible.

\subsection{Known Results on Unranked and Lexicographic Enumeration}
\label{sec:known}

We now review a key dichotomy result for \emph{unranked} enumeration,
where there is no ranking function and the task is to 
enumerate query answers in no particular order.
We also touch upon results for the special case of ranking according to lexicographic orders.

\introparagraph{Free-connex CQs}
Bagan et al.~\cite{bagan07constenum} show that unranked enumeration for full CQs is always possible 
with linear-time preprocessing and constant delay in data complexity.
When the CQ has existential variables (i.e., projections), then this is only possible
(under the hypotheses of the previous section)
for the CQs that are \emph{free-connex}.
A CQ is free-connex if it is acyclic and additionally,
it remains acyclic if we add an atom that contains all free variables~\cite{brault13thesis}.
Trivially, every full acyclic CQ is free-connex.
We restate the dichotomy in terms of $\TT(k)$\footnote{
The lower bound on $\TT(k)$ stated here is slightly stronger than the one mentioned
in the original papers
(which bound the delay 
without logarithmic factors),
but is covered by the hypotheses in a straightforward way.
} below:

\begin{theorem}[\cite{bagan07constenum,brault13thesis}]
\label{theorem:known-enumeration}
Let $Q$ be a CQ. If $Q$ is free-connex, then unranked enumeration
(in arbitrary order) 
is possible with $\TT(k) = \O(n + k)$.
Otherwise, if $Q$ is also self-join-free, then it is not possible 
with $\TT(k) = \O(n + k \polylog k)$,
assuming \sparseBMM{} and \hyperclique{}.
\end{theorem}

\introparagraph{Enumeration by lexicographic orders}
A closer look at the algorithm of Bagan et al.~\cite{bagan07constenum}
reveals that the query answers
are actually returned in \emph{some lexicographic order} of the variables.
In a lexicographic order $\langle x_1, x_2, \ldots \rangle$, two output tuples are first compared on the
$x_1$ value, and if equal then on their $x_2$ value, and so on.
The order of the variables crucially depends on the structure of the query
(i.e., it needs to be an \emph{alpha elimination order} \cite{baron16acyclic})
and there are certain lexicographic orders that cannot be realized by the algorithm.
For example, for the full 2-path query $Q_{P2}(x, y, z) \datarule R_1(x, y), R_2(y, z)$,
the lexicographic order $\langle x, z, y \rangle$ is not allowed.
Intuitively, this is because after choosing $x$ in $R_1$, 
we need to fix $y$ to determine the possible options for
$z$ in $R_2$.

Later work by Bakibayev et al. on factorized databases~\cite{bakibayev13fordering} 
shows how to achieve constant-delay enumeration according to any lexicographic order.
However, lexicographic orders that do not agree with the ``factorization order'' used in the construction of their data structure
require an expensive restructuring operation.
For $Q_{P2}$ with lexicographic order $\langle x, z, y \rangle$,
their approach would need to construct a representation of $\O(n^2)$ size,
which implies $\TT(k) = \O(n^2 + k)$.
(see \Cref{appendix:fdb_order}).

In our ranked-enumeration framework,
\emph{any lexicographic order} is supported 
with only an additional logarithmic factor
(i.e., $\TT(k) = \O(n + k \log k)$) for any free-connex 
CQ,
including the above example.
As we explain in more detail in \Cref{sec:lex},
lexicographic orders are at least as hard as the sum-of-weights case.

\section{Overview of Results and Outline}
\label{sec:overview}

We are now in a position to state the main results of this work more formally. 
The first result concerns the data complexity of ranked enumeration for free-connex CQs, 
accompanied by a conditional lower bound,
similar to the lower bound of unranked enumeration.
We note that our approach does not only apply to free-connex CQs, but also to any CQ with weaker guarantees.

\begin{theorem}[Any-$k$ CQs, Data Complexity]
\label{theorem:cq_data_comp}
Let $Q$ be a CQ.
If $Q$ is free-connex, then ranked enumeration with an \smonotone ranking function
is possible with $\TT(k) = \O(n + k \log k)$ 
and $\MEM(k) = \O(n + k)$.
Otherwise, if it is also self-join-free, 
then it is not possible with $\TT(k) = \O(n + k \polylog k)$,
assuming \sparseBMM{} and \hyperclique{}.
\end{theorem}
\begin{proofoutline}
We start with the simplest case of full path-structured CQs
and show the (serial) DP structure of the problem in \Cref{sec:paths_and_dp}.
All three algorithms presented in \Cref{sec:DPalgorithms} achieve the data complexity bounds of the theorem.
In \Cref{sec:tdp}, we modify our techniques for tree-structured problems 
which include all full acyclic CQs.
To ease the presentation, we adopt the sum-of-weights ranking but carefully note which properties of the ranking function we use. 
Free-connex CQs are reduced to full CQs in \Cref{sec:projections}.
In \Cref{sec:rankingfctgeneralizing}, we generalize to all \smonotone ranking functions, showing algorithm correctness for these cases.
\end{proofoutline}

In a more detailed analysis, we show that the \ANYKPARTP algorithm that we develop achieves the best known guarantees in combined complexity
for the class of full acyclic CQs and \ssmonotone ranking functions.

\begin{theorem}[Any-$k$ CQs, Combined Complexity]
\label{theorem:cq_combined_comp}
For a full acyclic CQ $Q$ with $\ell$ atoms of
maximum arity $\alpha$,
ranked enumeration with an algebraic and
\smonotone ranking function
is possible with $\TT(k) = \O(n \ell \alpha + k (\log k + \ell \alpha))$.
Ranked enumeration with an algebraic and
\ssmonotone ranking function
is possible with $\TT(k) = \O(n \ell \alpha + k (\log(\min\{k, n^{\ell-\diam(Q)+1}\}) + \ell \alpha))$.
In both cases, $\MEM(k) = \O(n \ell \alpha + k \ell \alpha)$.
\end{theorem}
\begin{proofoutline}
The two time bounds are achieved, respectively, by the \ANYKPART and \ANYKPARTP algorithms,
as we show in \Cref{sec:complexity}
(for path-structured CQs)
and in \Cref{sec:tdp}
(for tree-structured CQs).
Compared to \Cref{theorem:cq_data_comp}, 
proving these time complexity bounds requires
(1) the special handling of cases where the ranking function does not have an inverse (\Cref{sec:rankingfctgeneralizing})
and (2) the assumption that the ranking function is algebraic (which we adopt in \Cref{sec:complexity}). 
\end{proofoutline}

According to the theorem, an \ssmonotone ranking function such as sum-of-weights allows us to reduce the $k \log k$ term in the running time
to $k \log(\min\{k, n^{\ell-\diam(Q)+1}\})$.
This means that the higher the diameter of the CQ,
the smaller the logarithmic factor.
The largest difference is observed for
path CQs $Q_{P\ell}$ where $\diam(Q) = \ell$, and hence
$\TT(k) = \O(n \ell \alpha + k \log (\min\{k, n\} + \ell \alpha))$.
To see why this
is an improvement, note that $k$ can be very large, e.g., $\Omega(n^ { \lceil \ell / 2 \rceil })$ for $Q_{P\ell}$~\cite{AGM}.
In the other extreme, for star queries $Q_{S\ell}(x_1, \ldots, x_{2 \stages - 2}) \datarule R_1(x_1, \ldots, x_{\ell-1}), R_2(x_1, x_{\ell}), \ldots, R_\ell(x_{1}, x_{2 \stages - 2})$,
we have
$\diam(Q_{S\ell}) = 3$
and hence we obtain
$\TT(k) = \O(n \ell \alpha + k (\ell \log n + \ell \alpha))$,
i.e., an additional factor $\ell$ in the logarithm term.
We note that our upper bound is pessimistic, since
a high diameter is a sufficient but not necessary condition for our algorithm to reduce the $k \log k$ term in the running time.
This ultimately depends on the height of the constructed join tree.
For example, the arity-2 version of the star queries 
$Q_{2S\ell}(x_0, \ldots, x_\ell) \datarule E(x_0, x_1), E(x_0, x_2), \ldots, E(x_0, x_{\stages})$ 
have $\diam(Q_{2S\ell}) = 2$, however 
they admit a join tree that is a path
(similarly to the path CQs $Q_{P\ell}$) and, as a result,
admit the same running time as $Q_{P\ell}$.

As a side benefit, our algorithm can also be used for ranked enumeration of paths in a directed acyclic graph and thus applies to
a wide class of DP problems.

\begin{theorem}[Any-$k$ DP]
\label{theorem:paths}
In a weighted DAG $G$ with $\nodenum$ nodes and maximum path length $\ell$,
explicit ranked enumeration with an algebraic and \ssmonotone ranking function
is possible with $\TT(k) = \O(|G| + k(\log N + \ell))$
and $\MEM(k) = \O(|G| + k \ell)$.\footnote{Our algorithm actually achieves $\O(|G| + k(\log(\min \{ k, n \}) + \ell)$ but asymptotically,
this is the same as the bound of the theorem as we explain in \Cref{sec:complexity}.}
\end{theorem}
\begin{proofoutline}
The bounds are achieved by the \ANYKPARTP algorithm and its analysis in \Cref{sec:complexity}.
The correctness for ranking functions beyond sum is established in \Cref{sec:rankingfctgeneralizing}.
\end{proofoutline}

Our asymptotic time complexity
dominates the previously best-known algorithm due to Eppstein~\cite{eppstein1998finding}
with $\TT(k) = \O(|G| + k(\log k + \ell))$.
We strictly improve over this since $k$ can be as high as $2^{n-2}$ in a DAG
and for $k < n$ the first term dominates the running time.

\section{Path-Structured CQs and Dynamic Programming (DP)}
\label{sec:paths_and_dp}

We start with the subset of full path-structured CQs.
Note that this does not include only the CQs $Q_{P\ell}$ (\Cref{ex:path_cycle}),
but also any CQ whose join tree is a path (i.e., every node in the join tree has at most one child).
We show that the problem of computing
the single, top-ranked query answer is solvable by Dynamic Programming.
As a consequence, the weighted query answers can be represented as weighted paths in a DAG,
and their ranked enumeration (\cref{def:ranked_enum_cqs}) can be achieved by 
ranked enumeration of paths in that DAG (\cref{def:ranked_enum_paths}).

\subsection{Dynamic Programming Formulation}
\label{sec:dp}

We now describe a framework for Dynamic Programming (DP) as a DAG that captures a wide range of problems~\cite{bertsekas05dp,dpv08book}.
Our ranked enumeration approach described in the next section
applies to any problem expressible in this framework.

\introparagraph{DP as a DAG}
The DAG $G(V, E)$ of a DP problem
captures the dependencies between
different subproblems.
The nodes $V$ represent \emph{states}, which contain
local information for decision making~\cite{bertsekas05dp}.
Among them, there is a \emph{source} state $s$ and 
a \emph{terminal} or \emph{target} state $t$.
In each state $v \in V$, we have to make a {\em decision}
that leads to another state $v'$. 
These decisions are encoded as edges $E$.
Each decision $(v, v')$ is associated with a \emph{weight} (or cost) $\weight(v, v')$,
thus the DAG is weighted on its edges.
\footnote{``Cost'' is 
more common in optimization problems, ``weight'' in shortest-path problems.
We use the latter throughout the paper.
}
A \emph{solution} $\sol = \langle v_1, \ldots, v_\lambda \rangle$ 
to the DP problem is a sequence of $\lambda$ states that together with
$s$ and $t$
form an $s-t$ path in the graph,
i.e., $v_0 = s, v_{\lambda+1} = t$ and $(v_{i}, v_{i+1}) \in E$,
$\forall i \in [\lambda+1]_0$.
Notice that we do not include $s$ and $t$ when we write a solution, since they appear in all solutions.
The \emph{objective function} is the total cost of a solution,
\begin{align}
	\weight(\sol) = \aggrsum_{\sgiter=0}^{\lambda} \weight(v_{\sgiter}, v_{\sgiter+1}), \label{eq:costDP}
\end{align}
and DP finds the minimum-cost solution $\sol_1$. 
This corresponds to a shortest $s-t$ path in the DAG.\footnote{We use the term ``shortest''
for the path that has the smallest weight, regardless of the number of edges.
}
In our notation, the index denotes the rank,
i.e., $\sol_k$ is the $k$-th best solution.

By \emph{serial DP} we refer to the special case of a multi-stage graph: 
The states are partitioned into $\stages$ \emph{stages} and every decision 
from a state of stage $S_i$ can only lead to a state of stage $S_{i+1}$.
As a result, all solutions have the same size $\ell$. 
This is precisely the case for path CQs, as we shall see next.

\introparagraph{Principle of optimality}~\cite{bellman1954}
The core property of DP is that a solution can be efficiently derived from
solutions to subproblems.
In the shortest-path view of DP, the subproblem in \emph{any} state $v \in V$
is the problem of finding the shortest path from $v$ to $t$. 
With $\sol_1(v)$ and $\solW_1(v)=w(\sol_1(v))$
denoting a
shortest path from $v$ and its weight respectively, 
the DP algorithm computes the minimum of the objective function as follows:
\begin{equation}
\begin{aligned}
    \!\!\!\!\!\!\!\solW_1(t) &= 0 \textrm{ for the target $t$} \\
    \!\!\!\!\!\!\!\solW_1(v) &= 
		\!\!\min_{(v, v') \in \Dec} 
		\{\weight(v, v') \aggr \solW_1(v') \},
		\textrm{ for }
		v \in V \setminus \{t\} .
		\!\!
    \label{eq:DP_recursion}
\end{aligned}
\end{equation}
The cost of the optimal DP solution is then $\solW_1(s)$, i.e., the weight of the shortest path
from $s$ to $t$.
For convenience, we define the set of paths compared in \cref{eq:DP_recursion} 
(i.e., those that start at $v$, continue to a neighbor of $v$, and then reach $t$ optimally)
as
$\Choices_1(v)$.
With $\concat$ denoting concatenation, 
$\Choices_1(v) = \{v \concat \sol_1(v') \ |\ (v, v') \in \Dec \}$.

\introparagraph{DP algorithm}
\cref{eq:DP_recursion} can be computed for all states in time $\O(|\Sset| + |\Dec|)$ ``bottom-up'',
i.e., in a reverse topological sort of the DAG.
To compute $\Choices_1(v)$ for state $v \in V$,
the algorithm retrieves all decisions $(v, v') \in \Dec$
from $v$ to any state $v' \in V$, looks up $\solW_1(v')$,
and keeps track of the minimum total weight $\weight(v, v') \aggr \solW_1(v')$.
If no such edge is found, then the weight is set to $\infty$.
When computing $\solW_1(v)$, the algorithm also adds pointers to keep track of optimal solutions.
In this way, the
corresponding paths can be reconstructed by tracing the pointers back ``top-down'' from 
$s$~\cite{bertsekas05dp}. 

Whenever the bottom-up phase encounters a state $v$ without outgoing edges (i.e., $\Choices_1(v) = \emptyset$), 
then $v$ and all its adjacent edges can be
removed without affecting the solution space. 
If we apply this to all states, there will be no ``dead ends'' in the graph;
every node or edge reachable from $s$ will be part of some solution.
We use $\SsetR \subseteq V$ and $\DecR  \subseteq \Dec$ to denote the
\emph{sets of remaining states and decisions}, respectively. 
Note that the DP algorithm we describe corresponds to variable elimination~\cite{DBLP:journals/ai/Dechter99}
with the \emph{tropical semiring}~\cite{pin98tropical}
and
the removal of states and decisions is the same 
as the \emph{semi-join reductions by Yannakakis}~\cite{DBLP:conf/vldb/Yannakakis81}.

\begin{figure*}[tb]
\centering
\begin{subfigure}[t]{.35\linewidth}
    \centering
    \includegraphics[height=4.5cm]{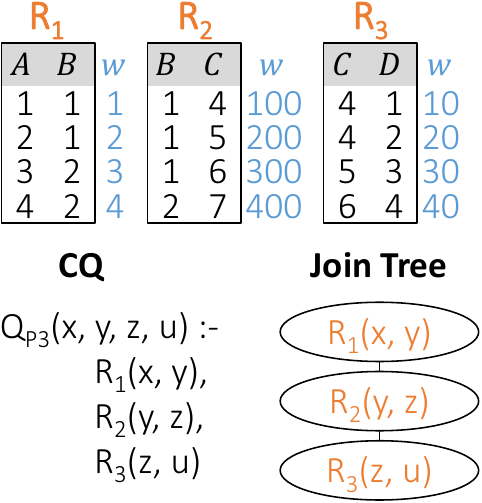}
    \caption{An example query $Q_{P3}$ together with an example database and a path-structured join tree.}
    \label{fig:tables_join_tree}
    \end{subfigure}%
\hfill
\begin{subfigure}[t]{.6\linewidth}
    \centering
    \includegraphics[height=4.5cm]{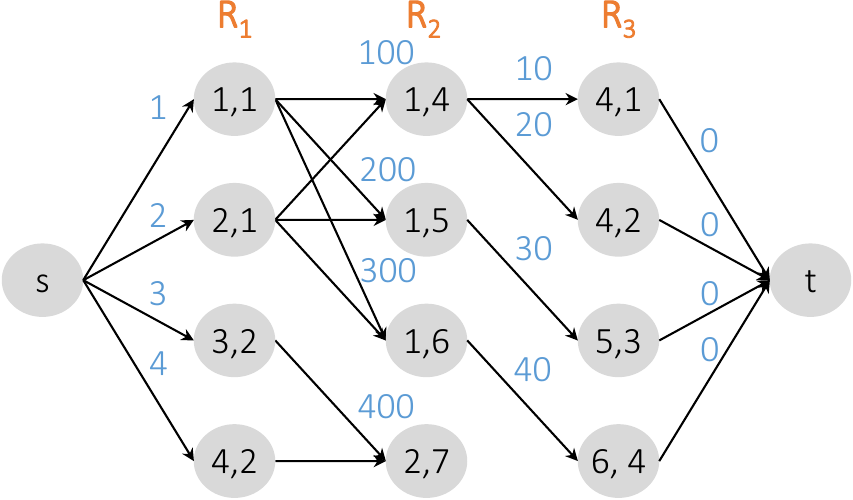}
    \caption{The DP instance represented as a weighted DAG. 
	Nodes represent tuples and edges joining pairs. 
	The edges are weighted according to the weight of the target tuple.}
    \label{fig:dp_naive}
\end{subfigure}
\caption{Constructing a DP instance from an input query and database.}
\label{fig:dp_from_tables}
\end{figure*}

\subsection{DP Instance for Path CQs}
\label{sec:cq_to_dp}

Computing the top-ranked answer to a full path CQ $Q_{P}$
can be modeled as an instance of serial DP.
To create it, we use the structure of the join tree.
Assume that the join tree of $Q_P$ is a path with the node
at depth $i$ referencing the relation $R_i, i \in [\ell]$.
The stages $S_i, i \in [\ell+1]_0$ of the DP instance are as follows:
(1) atom $R_i, i \in [\ell]$
corresponds to stage $S_i$ and each
tuple in $R_i$ maps to a unique state in $S_i$,
(2) there is an edge
between $v \in S_i$ and $v' \in S_{i+1}$ iff the corresponding
input tuples join and the weight of the edge is the weight of the tuple corresponding to $v'$,
(3) there is an edge from $s$ to each state in $S_1$ whose weight is
the weight of the corresponding $R_1$-tuple, and
(4) each state in $S_\stages$ has an edge to $t$ of weight 0.
Clearly, there is a 1:1 correspondence between
paths from $s$ to $t$ and query answers.
Due to the commutativity of sum,
the weight of such a path is equal to the weight of the corresponding query answer,
regardless of the stage order imposed by the join tree. 
Consequently, the $k^\textrm{th}$-best 
query answer
corresponds to the $k^\textrm{th}$-shortest path in the DP instance.

\begin{example}[Mapping of path CQ to DP]\label{ex:dp_naive}
We use the problem of finding the minimum-weight answer to the 3-path $Q_{P3}$
as a running example. 
\Cref{fig:tables_join_tree} shows an example database and the weights assigned to each tuple.
The figure also shows a possible join tree that has a path structure with
$R_1$ as the root.
The corresponding DP instance is depicted in \Cref{fig:dp_naive}.
It has 5 stages: Three of them correspond to relations and encode input tuples,
while the remaining two correspond to the source and terminal nodes $s, t$.
Every edge encodes two joining tuples.
For example, node $( 1, 1 )$ is connected to $(1, 5)$ because $y$ is a common variable between $R_1(x, y)$ and $R_2(y, z)$ and both tuples assign $1$ to $y$.
The edges are weighted according to target-tuple weight,
e.g., edge $((1, 1), (1, 5))$ has weight $100$ because $w((1, 5))=100$.
The DP algorithm visits the stages right-to-left (i.e., from $R_3$ to $R_1$)
and computes the minimum weight path to target $t$
for every node.
For example, $\solW_1((1, 1)) = \min \{ 100 + 10, 200 + 30, 300 + 40 \} = 110$
where $10, 30, 40$ are $\solW_1((1, 4)), \solW_1((1, 5)), \solW_1((1, 6))$, respectively,
which have already been computed when we visit $(1, 1)$.
For the source node, we have $\solW_1(s) = 111$ which corresponds precisely to the minimum-weight query answer.
A top-down traversal from $s$ can retrieve the exact edges that were chosen for $\solW_1(s)$,
from which we reconstruct the witness of the top-1 query answer $((1, 1), (1,4), (4,1))$.
\end{example}

\begin{figure}[tb]
\centering
\includegraphics[height=4.2cm]{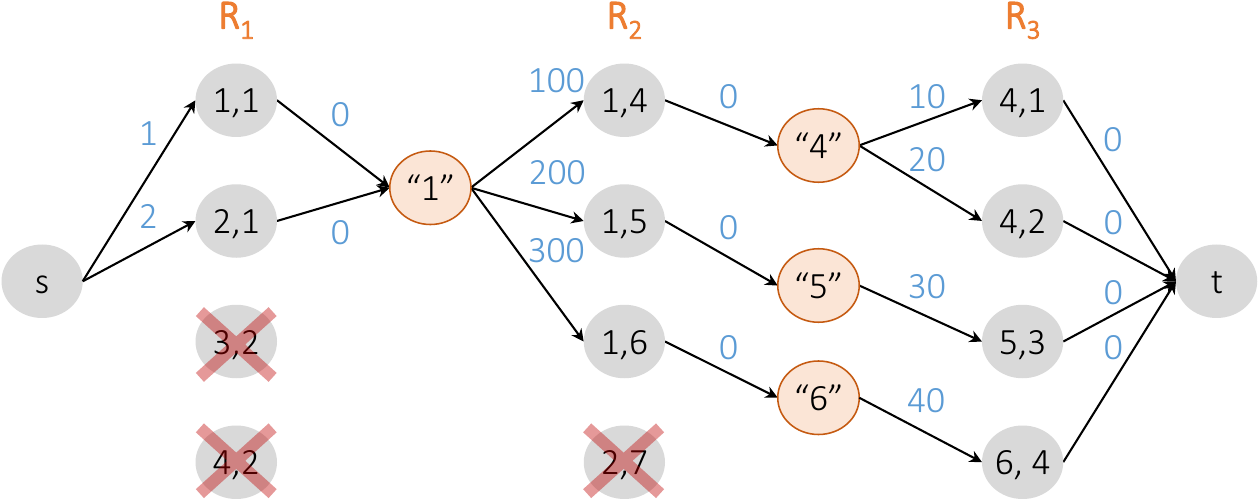}
\caption{\emph{Equi-join} representation from $\O(n^2)$ to $\O(n)$.
Intermediate nodes are introduced (in color), which correspond to values of the join attribute(s).
Additionally, nodes that have no outgoing edges are removed (equivalent to a semi-join reduction).}
\label{fig:equiJoinGraph}
\end{figure}

\introparagraph{Encoding equi-joins efficiently}
For an equi-join, the naive construction we describe above has $\O(\stages n)$ states and
$\bigO(\stages n^2)$ edges,
therefore the DP algorithm would have quadratic time complexity in the number of tuples.
This is due to values that are ``heavy hitters'',
i.e., occur with high frequency and result in a dense connection pattern
(e.g., value 1 between $R_1$ and $R_2$ in \cref{fig:dp_naive}).
We reduce the graph size to $\O(\stages n)$ with an 
\emph{equi-join-specific graph transformation} illustrated in \cref{fig:equiJoinGraph}.

Consider the join between $R_1$ and $R_2$, representing stages $\Sset_1$ and $\Sset_2$, respectively.
For each join-attribute value, the corresponding 
states in $R_1$ and $R_2$ 
form a fully connected bipartite graph. For each state,
all \emph{incoming} edges have the
same weight, as edge weight is determined by tuple weight. 
Therefore, we can represent the subgraph equivalently
with a single node ``in-between'' the matching states in $\Sset_1$ and $\Sset_2$,
assigning zero weight to edges starting from states in $\Sset_1$ and the target-tuple weight to those leading to states in $\Sset_2$.
The transformed representation has only $\O(\stages n)$ edges
and preserves all the connections that exist in the original graph.
Its total size is $\O(n \stages \alpha)$ since the states that correspond to input tuples need to store $\O(\alpha)$ attribute values.
At its core, our encoding relies on the conditional independence of the non-joining
attributes given the join attribute value, a property also exploited in
factorized databases \cite{olteanu16record}.
Here, we provide a different perspective on it as a graph transformation of bipartite cliques~\cite{feder91cliques}.

This efficient representation can be constructed directly from the input tables and the join tree in $\O(n \stages \alpha)$.
Similarly to the DP algorithm, we visit the input relations bottom-up, and for each stage,
we create the edges to the previous stage by using a lookup table (i.e., a hash table) on the join attribute(s).\footnote{If indexes already exist in the database, then they can be exploited here.
Note that using a B-tree or a sort-based method will incur an additional logarithmic factor in the analysis.
}
The time to construct each lookup table and then perform lookups for all tuples of a relation is $\O(n \alpha)$.
Nodes that have no outgoing edges can already be removed in this step.

\begin{example}[Equi-join Transformation]
\label{ex:equi-join}
The transformed DP instance for \cref{fig:dp_from_tables} is shown in \cref{fig:equiJoinGraph}.
The $6$ edges between $R_1$ and $R_2$ for join value $1$ are replaced by one intermediate node
and $2+3$ edges.
This reduces the number of edges from quadratic to linear.
Constructing the representation bottom-up, we first remove $(2, 7)$ because $7$ does not appear in the first column of $R_3$,
and as a consequence, we also remove $(3, 2)$ and $(4, 2)$ of $R_1$.
\end{example}

\section{Any-$k$ Algorithms for DP}
\label{sec:DPalgorithms}

We defined a class of DP problems that can be described in terms of a DAG,
where every solution is equivalent to a path from $s$ to $t$ in a ``reduced'' graph $(\SsetR, \DecR)$. 
In addition to the minimum-weight path, ranked enumeration retrieves \emph{all paths}
in weight order (\Cref{def:ranked_enum_paths}). Let $\sol_k(v)$ be the
$k^\textrm{th}$-shortest suffix
from node $v$ to $t$ and $\solW_k(v)$
its weight (i.e., $\solW_k(v) \define w(\sol_k(v))$). 
The goal is to return the sequence $\sol_1(s), \sol_2(s), \ldots$

First, we explore algorithms that follow two different approaches.
The first \emph{partitions} the solution space to find the next best solution 
and traces its roots to the works of Hoffman and Pavely~\cite{hoffman59shortest},
Lawler~\cite{lawler72}, and Murty~\cite{murty1968}.
We call it \ANYKPART and show that it has the lowest $\TT(k)$ for small $k$.
The second, which we call \ANYKREC, formulates the problem 
\emph{recursively}~\cite{bellman60kbest,dreyfus69shortest,jimenez99shortest},
and, as we show,
it has the lowest $\TT(k)$ for large $k$ on certain inputs.
Then we extend \ANYKPART with ideas from \ANYKREC to develop \ANYKPARTP,
a new algorithm that combines the respective advantages of both algorithms to give the best asymptotic $\TT(k)$ for any $k$.

\subsection{Partitioning-based Algorithm (ANYK-PART)}
\label{sec:part}

The \ANYKPART algorithm (1) relies on the \emph{Lawler-Murty procedure}~\cite{lawler72,murty1968},
which is a general ranked-enumeration approach
that applies to a wide range of optimization problems,
and (2) exploits the DP structure of our problem
using the concept of deviations.

\introparagraph{Deviations}
Given the source node $s$ and the target node $t$,
we can impose a deviation structure on the $s-t$ paths.

\begin{definition}[Deviation~\cite{hoffman59shortest}]
A deviation of a path $\sol$ is a path $\Pi'$
that follows the same edges as $\Pi$ from $s$ up to a node $v_i$,
then takes a different edge to a node $v_{i+1}'$,
and then the optimal path $\sol_1(v_{i+1}')$ to $t$.
The edge $(v_i, v_{i+1}')$ is called the deviating edge.
\end{definition}

Already in 1959, Hoffman and Pavely~\cite{hoffman59shortest} showed that every path
is a deviation of some shorter (or equally short) path.
Unfortunately, the precise statement in the original paper is inaccurate in the presence of ties between different paths.
We present the formal statement and its proof more precisely and for more general ranking functions:

\begin{lemma}
\label{lem:deviations}
For an \smonotone ranking function over the $s-t$ paths of a DAG and $k > 1$, 
there exists a valid ordering of the paths
such that 
the $k^{\textrm{th}}$-ranked path $\sol_k(s)$ is a deviation of some higher-ranked path $\sol_j(s)$
(i.e., $j < k$).
\end{lemma}
\begin{proof}
Let $\sol_k(s) = \langle s, v_1, v_2, \ldots, v_\lambda, t \rangle$.
Also, let $v_i$ be the last node of $\sol_k(s)$ with the property that the suffix 
$\sol(v_i) = \langle v_i, v_{i+1}, \ldots, v_\lambda, t \rangle$ 
is
not the same as $\sol_1(v_i)$.
Then, $\sol_k(s)$ is a deviation of $\sol_j(s) = \langle s, v_1, v_2, \ldots, v_i \rangle \concat \sol_1(v_i)$
because, from the way we picked $v_i$, the suffix $\langle v_{i+1}, \ldots, v_\lambda, t \rangle$ is necessarily optimal.
Since $\sol_1(v_i)$ is the optimal suffix from $v_i$, 
we have $w(\sol_1(v_i)) \preceq w(\sol(v_i))$ and by \smonotonicity, $w(\sol_j(s)) \preceq w(\sol_k(s))$.
If the inequality is strict, then $\sol_j(s)$ necessarily appears before $\sol_k(s)$ in the ranked order, i.e., $j < k$.

Otherwise, their weight is the same and either of them can be ranked higher.
For this case, we show that there exists a tie-breaking mechanism so that $\sol_j(s)$ is ranked higher than $\sol_k(s)$.
Let $\sol(v)$ and $\sol(v)'$ be two different equal-weight paths starting from some node $v$ and ending at $t$.
Also let $v'$ be their last common node starting from $v$.
If either $\sol(v)$ or $\sol(v)'$ coincides with the optimal suffix $\Pi_1(v)$, then it is preferred.
Otherwise, we can tie-break arbitrarily.
By this mechanism, $\sol_j(s)$ has to appear before $\sol_k(s)$ since $\sol_j(s)$ contains $\sol_1(v_i)$
and the two paths coincide up to $v_i$.
\end{proof}

Deviations provide an efficient way to explore the solution space, starting from the best 
path $\sol_1(s)$ which is already computed by DP.
\ANYKPART maintains a priority queue of candidates $\Cand$ 
initialized with $\sol_1(s)$. To produce the next-best path,
it pops from $\Cand$,
returns the path,
and pushes its deviations
back to $\Cand$.
For the correctness of this algorithm, note that every path can be constructed through deviations starting from $\sol_1(s)$
and \Cref{lem:deviations} ensures that the order in which
they are produced is correct.

\introparagraph{Avoiding duplicates}
A straightforward application of the idea described above that attempts
all possible deviations of every path produces duplicate candidates.
The Lawler-Murty procedure~\cite{lawler72,murty1968} 
provides a way to choose deviations in order to avoid duplicates.
For a path whose last deviating edge is $(v_i, v_{i+1})$, it does not produce new candidates for:
(1) deviating edges from $s$ up to $v_i$ and
(2) deviating edges to replace $(v_i, v_{i+1})$ that have already been produced with the same prefix.

\begin{figure}[tb]
\centering
\begin{minipage}[t]{.58\textwidth}
    \centering
    \includegraphics[height=2.4cm]{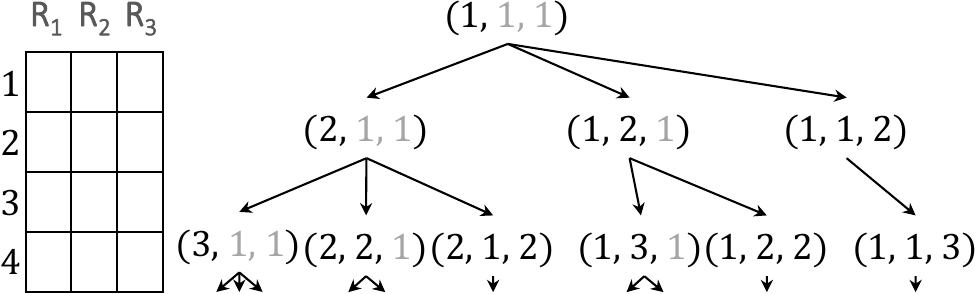}
    \caption{Ranked enumeration of Cartesian Product $R_1 \times R_2 \times R_3$ for three relations $R_1, R_2, R_3$. The numbers $1-4$ are the ranks of the input tuples within each relation.
    Those in gray are optimal suffixes and are not represented in our candidate encoding.}
    \label{fig:cartesian_product}
\end{minipage}%
\hfill
\begin{minipage}[t]{.39\textwidth}
    \centering
    \includegraphics[height=2.4cm]{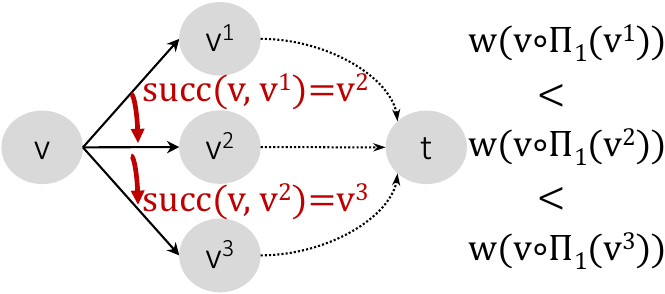}
    \caption{The successor function returns the next best child
    when comparing the optimal paths that reach $t$.}
    \label{fig:successor}
\end{minipage}
\end{figure}

\begin{example}[Cartesian Product]
\label{ex:cartesian}
\Cref{fig:cartesian_product} depicts how a Cartesian Product of three relations $R_1, R_2, R_3$ of size $4$
is explored according to the Lawler-Murty procedure.
Equivalently, this problem can be encoded in the DP framework as a 3-stage fully-connected graph together with a source and a target node (see \Cref{sec:dp}).
The figure shows the new candidates that we generate after returning each solution as its children in a tree.
Every solution is denoted by three numbers, indicating the rank of the input tuples in each relation.
For example, $(1, 2, 3)$ refers to the combination of the 
best $R_1$-tuple with the
$2^\textrm{nd}$-best $R_2$-tuple and the
$3^\textrm{rd}$-best $R_3$-tuple.
Clearly, $(1, 1, 1)$ is the best solution.
The candidates for the $2^\textrm{nd}$-best solution are its three deviations, one for each relation (or stage).

Notice that $(2, 2, 1)$ could potentially be generated by both $(1, 2, 1)$ and $(2, 1, 2)$.
To avoid this, the Lawler-Murty procedure produces deviations only in the third stage for $(2, 1, 2)$; this is because its last deviating edge is in the third stage (rule 1).
Further, notice that according to our definition, 
$(1, 1, 1)$ is a deviation of $(1, 2, 1)$
so we have to avoid producing it as a candidate again after returning $(1, 2, 1)$ (rule 2).
In our example, this is guaranteed by considering the input tuples in increasing order of their ranks.
We will formalize this with the notion of a successor.
\end{example}

\introparagraph{Encoding of candidates}
A convenient and efficient way to encode the candidate paths is to represent them
as prefixes up to the last deviating edge.
Under this convention, every prefix $\langle v_1, v_2, \ldots, v_i \rangle$ corresponds to the
path that we obtain if we expand it optimally up to the terminal node,
i.e., $\langle v_1, v_2, \ldots, v_i \rangle \concat \sol_1(v)$.
This expansion is performed whenever we pop a path from $\Cand$
so that the algorithm returns the path in an explicit form.
A benefit of this encoding is that it makes it easier to identify the
deviations that we need to generate in order to avoid duplicates;
we simply start from the last edge contained in the prefix (see \Cref{fig:cartesian_product}).

\introparagraph{The successor function}
In \Cref{ex:cartesian}, we assumed that the input relations were already sorted, which allowed us, for example, to generate only $(2, 1, 1)$ from $(1, 1, 1)$ in the first stage,
and not $(3, 1, 1)$ or $(4, 1, 1)$.
In general, when we deviate from an edge $(v_i, v_{i+1})$,
we want to generate only the best deviation that we have not yet considered.
The possible deviations we can pick are determined by the optimal suffixes of $\Choices_1(v_i)$.
Therefore, we assume a total order on $\Choices_1(v_i)$
that orders the children $v'$ of $v_i$
according to the weight of their optimal suffix $w(v_{i} \concat \sol_1(v'))$.
This order is given by a successor function $\suc(v_i, v_{i+1})$,
which returns the node $v_{i+1}'$ that is the next-best choice,
illustrated in \Cref{fig:successor}.
The successor function also helps us to apply rule (2) of the Lawler-Mutry procedure,
since we visit the deviations in successor order.

We are now in a position to describe the complete \ANYKPART approach,
as given in \Cref{alg:anyk-part}.
In each iteration, we pop a prefix $\langle v_1, \ldots, v_{i-1}, v_i \rangle$ from $\Cand$,
expand it optimally into $\langle v_1, \ldots, v_{i-1}, v_i, v_{i+1} \ldots, v_\lambda \rangle$,
and then generate the candidates: 
\begin{align*}
& \langle v_1, \ldots, v_{i-1}, \suc(v_{i-1}, v_i) \rangle \\
& \langle v_1, \ldots, v_{i-1}, v_i, \suc(v_i, v_{i+1}) \rangle \\
& \langle v_1, \ldots, v_{i-1}, v_i, v_{i+1}, \ldots, \suc(v_{\lambda-1}, v_\lambda) \rangle
\end{align*}

\begin{figure*}[tb]
\centering
\begin{subfigure}[t]{.55\linewidth}
    \centering
    \includegraphics[height=4cm]{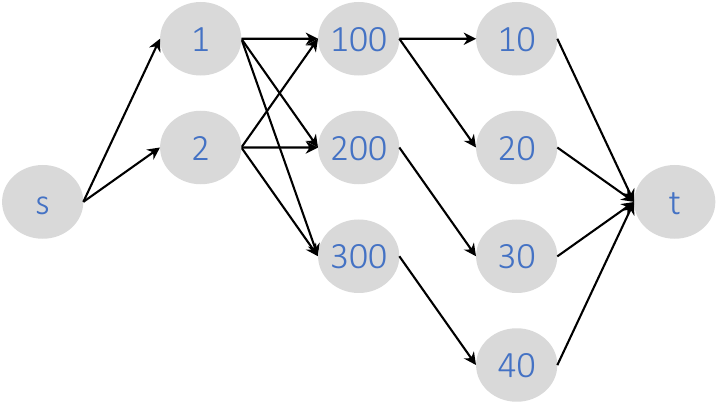}
    \caption{The graph of \Cref{fig:equiJoinGraph} simplified by removing the intermediate nodes and labeling the nodes by the weights of the incoming edges.}
    \label{fig:graph}
    \end{subfigure}%
\hfill
\begin{subfigure}[t]{.43\linewidth}
    \centering
    \includegraphics[height=4cm]{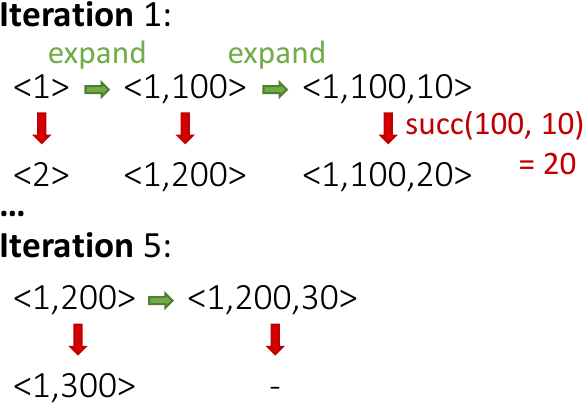}
    \caption{Two iterations of the algorithm. Source node $s$ and target node $t$ are not encoded in the solutions for simplicity.}
    \label{fig:iterations}
\end{subfigure}
\caption{Example iterations of \ANYKPART.}
\label{fig:part_example_run}
\end{figure*}

\begin{algorithm}[t]
\small
\SetAlgoLined
\LinesNumbered
\textbf{Input}: DP problem as a DAG $G(V, E)$ with source $s$ and target $t$\\
\textbf{Output}: solutions in increasing order of weight\\

Execute standard DP to produce for each node $v$: $\sol_1(v)$ and $\solW_1(v)$\;

\algocomment{Initialize candidates with top-1.}\;
\algocomment{The priority of a prefix $\langle v_1, \ldots, v_i \rangle$ in $\Cand$ is the weight of its optimal expansion $w(\langle s, v_1, \ldots, v_i \rangle) + \solW_1(v_i)$.}\;
$\Cand = [\langle v_1^* \rangle]$, where $\sol_1(s) = s \concat \sol_1(v_1^*)$\label{line:initialCand}\;

\Repeat {query is interrupted or $\Cand$ is empty}{\label{line:repeat}

    $\langle v_1, \ldots, v_i \rangle$ = $\Cand.\mathrm{popMin}()$\;

    \algocomment{Expand optimally into full solution.}\;
    $\langle v_1, \ldots, v_i, \ldots, v_\lambda \rangle$ = $\langle v_1, \ldots, v_i \rangle \concat \sol_1(v_i)$\;
    
    \algocomment{Generate new candidates by taking deviations.}\;
    \For {$j$ from $i$ to $\lambda$\label{line:for1}}{

        $\Cand.\mathrm{add}(\langle v_1, \ldots, v_i, \ldots, v_{j-1}, \suc(v_{j-1}, v_j) \rangle)$ \label{line:candidate_generation}\;
    }
    output solution $\langle v_1, \ldots, v_i, \ldots, v_\lambda \rangle$\label{line:output}\;
}
\caption{\ANYKPART}
\label{alg:anyk-part}
\end{algorithm}

\begin{example}[\ANYKPART on \Cref{ex:dp_naive}]
Consider again the DAG of \Cref{ex:dp_naive}.
For ease of presentation, \Cref{fig:graph} shows the same graph
without the equi-join transformation which introduces intermediate nodes (\Cref{ex:equi-join}) 
and with every node identified by the weight of its corresponding tuple.
$\Cand$ initially contains only one candidate $\langle 1 \rangle$ (\Cref{line:initialCand}).
It corresponds to the shortest path $\langle 1, 100, 10 \rangle$ found by DP.
This prefix is popped and expanded in the first iteration of the
repeat loop (\Cref{line:repeat}), leaving $\Cand$ empty for now.
As \Cref{fig:iterations} shows, the for loop in \Cref{line:for1} 
generates three new deviations as candidates.
These are $\langle 2 \rangle, \langle 1, 200 \rangle, \langle 1, 100, 20 \rangle$.
Each one is generated by applying the successor function; for example, we obtain $\langle 1, 200 \rangle$ from $\langle 1, 100 \rangle$
because $\suc(1, 100) = 200$ ($100 + 10 < 200 + 30 < 300 + 40$).
From the three new candidates, $\langle 2 \rangle$ is the one with the minimum weight
$2 + \solW_1(2) = 112$ and will be the one popped next from $\Cand$.
In the fifth iteration, $\langle 1, 200 \rangle$ will be popped. 
Following the Lawler-Murty procedure,
the for-loop (\Cref{line:for1}) for that prefix will be executed only two times,
generating only one new candidate since $30$ has no successor.
\end{example}

\introparagraph{Variants of \ANYKPART}
Different variants of \ANYKPART arise from different
implementations of the successor function.
The earlier version of our work~\cite{tziavelis20vldb} presented, analyzed, and compared
several variants
with ideas that were present in past work\footnote{Some of these variants require a more general definition of the successor 
function, allowing it to return a set of nodes instead of only one.
}:
\EAGER which uses sorting,
\MIN based on Yang et. al~\cite{yang2018any},
\LAZY based on Chang et. al~\cite{chang15enumeration},
and \HEAP as the asympotically best variant.
However, \HEAP is only better in terms of \emph{delay} instead of the more practically relevant $\TT(k)$ (see \Cref{sec:complexity_measures}).
A more careful analysis of $\TT(k)$ shows that \LAZY and \HEAP have the same complexity.
This is a consequence of \Cref{lem:complexity} that we prove in 
\Cref{sec:complexity}.

In the present article, we elect to put less emphasis on the differences
between these variants and defer a more detailed discussion to \Cref{sec:part_variants}.
In practice, we found a variant which we call \QUICK to be the best performer
and this is the one we show in our experimental evaluation
(\Cref{sec:experiments}).
\QUICK uses the Incremental Quicksort algorithm~\cite{paredes06iqs} to sort the children of
every node incrementally.
It is a randomized algorithm and, in expectation, it achieves the same $\TT(k)$ as \LAZY and \HEAP.

To simplify the analysis in \Cref{sec:complexity}, we assume the deterministic \LAZY variant.
After the bottom-up phase of DP, \LAZY constructs a binary heap in linear time for each node that contains one element per outgoing edge.
Successor calls are handled by popping from the heap and moving the popped element to a sorted list.
As the algorithm progresses, the heap gradually empties out, filling the sorted list, allowing
subsequent successor calls to be handled in constant time.

\subsection{Recursive-based Algorithm (ANYK-REC)}
\label{sec:rec}

\begin{figure}
\centering
\begin{subfigure}{.47\textwidth}
    \centering
    \includegraphics[height=4.8cm]{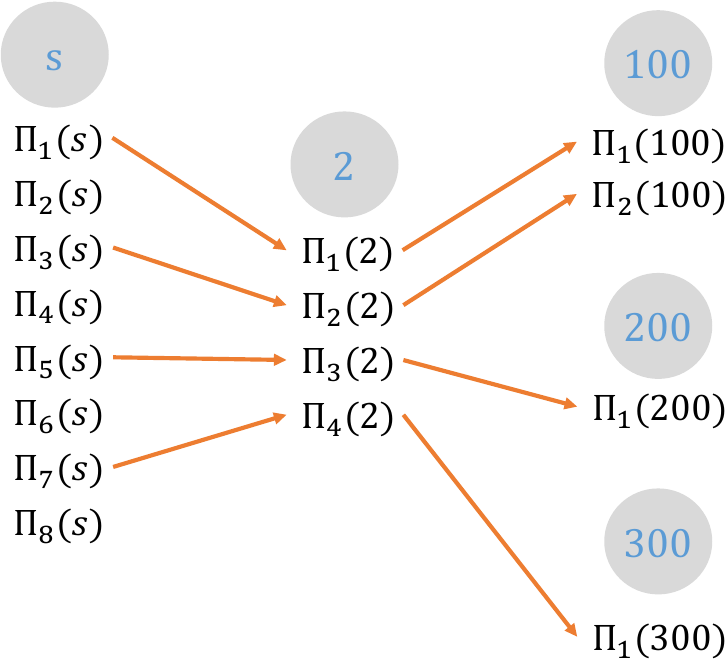}
    \caption{Pointers between solutions from and to node $2$.}
    \label{fig:rec_pointers}
\end{subfigure}%
\begin{subfigure}{.53\textwidth}
    \centering
    \includegraphics[height=4.8cm]{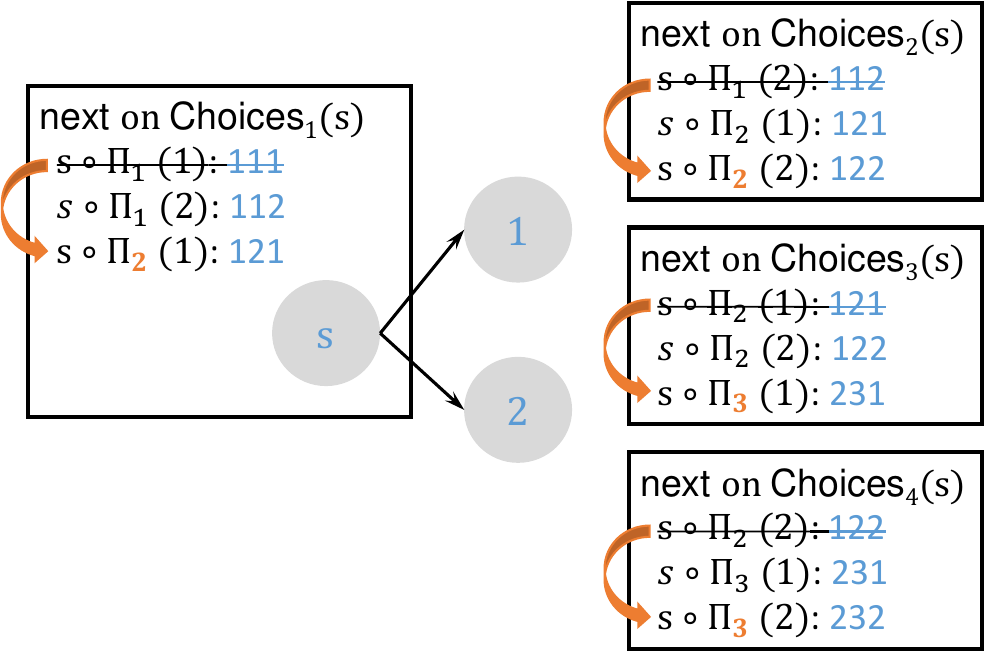}
    \caption{Recursive enumeration 
    of paths starting
    at node $2$.}
    \label{fig:rec_next}
\end{subfigure}
\caption{\Cref{ex:RE-DP}: Recursive enumeration}
\label{fig:REA}
\end{figure}

\begin{algorithm}[t]
\small
\SetAlgoLined
\LinesNumbered
\SetKwFunction{RecFun}{next}
\SetKwProg{Fn}{Function}{:}{}
\textbf{Input}: DP problem as a DAG $G(V, E)$ with source $s$ and target $t$\\
\textbf{Output}: solutions in increasing order of weight\\

Execute standard DP to produce for each state $s$: $\sol_1(s)$, $\solW_1(s)$, and $\Choices_1(s)$\;

\algocomment{Assume that $\sol_0(v)$ for a node $v$ is a special value used for initialization.}\;
$k = 0$ \;
\Repeat {query is interrupted or $\Choices_k(s)$ is empty}
{
    output solution $\sol_{k+1}(s) = \mathrm{\RecFun}(\sol_k(s))$\;
    $k = k + 1$\;
}
\;

\algocomment{Returns the next best solution starting from $v$.}\;
\Fn{\RecFun{$\sol_{k}(v)$}}{
    \algocomment{Base case: Terminal. Assume that null values are ignored throughout.}\;
    \If {$v == t$}{
        \KwRet null\;
    }
    \algocomment{If $\sol_{k + 1}(v)$ has been computed by some previous call, it has been stored at node $v$.}\;
    \If {$\sol_{k + 1}(v)$ is not in the sorted list of $v$ or $k == 0$}
    {
        $\sol_{k + 1}(v) = \Choices_{k}(v)$.popMin()\;
        Store $\sol_{k + 1}(v)$ in the sorted list of $v$\;
        \algocomment{Replace the popped path with the next one that goes to the same node, computed recursively.}\;
        Let $\sol_{k + 1}(v) = v \concat \sol_{k'}(v')$\;
        $\Choices_{k+1}(v) = \Choices_{k}(v)$.add($v \concat \mathrm{\RecFun}(\sol_{k'}(v'))$) \label{rec_line:insert}\;        
    }
    \KwRet $\sol_{k + 1}(v)$\;
}
\caption{\ANYKREC}
\label{alg:rec}
\end{algorithm}

\introparagraph{Generalized principle of optimality}
\ANYKREC relies on a \emph{generalizaton of the principle of
optimality} of DP, which in its usual form states that there exists an optimal path consisting of optimal suffixes.
The generalization gives a similar property to the $k^\textrm{th}$ path in the ranking.
To our knowledge, the first algorithm that uses this principle is due to Bellman and Kalaba~\cite{bellman60kbest},
while the principle has been stated more explicitly in later work by Martins et al.~\cite{martins2001opt}.

\begin{lemma}[Generalized Principle of Optimality]
\label{lem:generalizedopt}
For an \smonotone ranking function over the $s-t$ paths of a DAG $G(V, E)$ and $k \geq 1$, 
there exists a valid ordering of the paths such that for every node $v$, $\Pi_k(v)$ is equal to $v \concat \sol_{k'}(v')$ 
for some $k' \leq k$ and $(v, v') \in E$.
\end{lemma}
\begin{proof}
Let $v'$ be the node after $v$ in $\Pi_k(v)$.
Suppose that $\Pi_k(v)$ is $v \concat \Pi_i(v')$ for some $i > k$.
Consequently, there has to exist an index $j \in [k]$ such that
the path $v \concat \Pi_j(v')$ is not in the top $k$ suffixes from $v$. 
Since $\Pi_j(v')$ precedes $\Pi_i(v')$ in the ranking of suffixes from $v'$,
we have $w(\Pi_j(v')) \preceq w(\Pi_i(v'))$ and by \smonotonicity,
$w(v \concat \Pi_j(v')) \preceq w(v \concat \Pi_i(v'))$.
If the last inequality is strict, then we have a contradiction, since $v \concat \Pi_j(v')$ should have been in the top $k$ paths from $v$.
If their weights are equal, then we can swap them in the ordered list
(producing a different valid ordering)
so that $\Pi_k(v)$ is equal to $v \concat \sol_{k'}(v')$.
\end{proof}

For the shortest path (i.e., $k = 1$), we have $k' = 1$, which is the principle of optimality that
we saw in \Cref{sec:dp}.
The generalized principle of optimality gives rise to an algorithm that computes lower-ranked suffixes for all nodes,
not just the source node $s$.
In particular, to compute $\sol_k(s)$, it is \emph{sufficient} to compute $\sol_k(v)$ for every node $v$~\cite{bellman60kbest}. 
However, it is not \emph{necessary}~\cite{dreyfus69shortest} since fewer than $k$ suffixes might be needed for some of the nodes.
\ANYKREC computes only those that are needed using
a recursive call structure proposed by Jim{\'e}nez et al.~\cite{jimenez99shortest}.

\introparagraph{Recursive calls}
Recall that the shortest path $\sol_1(s)$ from the source node $s$
is found as the minimum-weight path in $\Choices_1(s)$. 
Assume it goes 
through $v'$. 
Through which node does the $2^\textrm{nd}$-shortest path $\sol_2(s)$ go? 
It must be either the $2^\textrm{nd}$-shortest path through $v$, of weight $\weight(s, v) + \solW_2(v)$, or the shortest
path through any of the other nodes adjacent to $s$. 
In general, the $k^\textrm{th}$-shortest path $\sol_k(v)$ from any node $v$
is determined as the minimum-weight path in some later version of the set of choices
$\Choices_k(v) = \{v' \concat \sol_{k'}(v') \ |\ (v, v') \in \DecR\}$,
for appropriate values of $k'$.
Let $\sol_k(v) = v \concat \sol_{k'}(v')$. 
Then the $(k+1)^\textrm{st}$
solution $\sol_{k+1}(v)$ is found
as the minimum over the same set of choices as in iteration $k$, except that
$v \concat \sol_{k'+1}(v')$ replaces
$v \concat \sol_{k'}(v')$.
To find $\sol_{k'+1}(v')$, the same procedure is applied recursively 
at $v'$ \emph{top-down}. Intuitively, an iterator-style \texttt{next} call at source node $s$
triggers a chain of such \texttt{next} calls along the path that was found
in the previous iteration.

\introparagraph{Data structures}
The sets of choices $\Choices_k(v)$ of every node $v$ is implemented with a \emph{priority queue}
so that we can find the minimum-weight choice efficiently.
Additionally, every node maintains a \emph{sorted list} of suffixes.
As the lower-ranked suffixes are computed, they are added to the sorted list so that they can be reused
(when the same node is visited via a different prefix).
This type of \emph{memoization} (in addition to that of standard DP) makes the algorithm faster as the value of $k$ increases.
\Cref{alg:rec} contains the detailed pseudocode.

\begin{example}[\ANYKREC on \Cref{ex:dp_naive}]
\label{ex:RE-DP}
Consider source node $s$ in \cref{fig:graph}.
Since its children are $1$ and $2$, 
the shortest path $\sol_1(s)$ is selected from
$\Choices_1(s) = \{s \concat \sol_1(1), s \concat \sol_1(2)\}$.
The first \texttt{next} call on node $s$ returns
$s \concat \sol_1(1)$, updating the set of choices for $\sol_2(s)$ to
$\{s \concat \sol_2(1), s \concat \sol_1(2) \}$ as shown in the left box in~\cref{fig:rec_next}. 
The subsequent \texttt{next} call on $s$ then returns
$s \concat \sol_1(2)$ for $\sol_2(s)$, causing $s \concat \sol_1(2)$ in
$\Choices_2(s)$
to be replaced by $s \concat \sol_2(2)$
for $\Choices_3(s)$; and so on.
The paths $\sol_2(1)$, $\sol_2(2), \ldots$ are themselves constructed in a similar way
with pointers to paths that start at nodes $100$, $200$, and $300$, as shown in \cref{fig:rec_pointers}.
\end{example}

\begin{algorithm}[t]
\small
\SetAlgoLined
\LinesNumbered
\SetKwFunction{FollowFun}{follow}
\SetKwFunction{StoreFun}{store}
\SetKwProg{Fn}{Procedure}{:}{}
\textbf{Input}: DP problem as a DAG $G(V, E)$ with source $s$ and target $t$\\
\textbf{Output}: solutions in increasing order of weight\\

Execute standard DP to produce for each node $v$: $\sol_1(v)$ and $\solW_1(v)$\;

\algocomment{Initialize candidates with top-1.}\;
\algocomment{The priority of a prefix $\langle v_1, \ldots, v_i \rangle$ in $\Cand$ is the weight of its optimal expansion $w(\langle s, v_1, \ldots, v_i \rangle) + \solW_1(v_i)$.}\;
$\Cand = [\langle v_1^* \rangle]$ where $\sol_1(s) = s \concat \sol_1(v_1^*)$\;

\Repeat {query is interrupted or $\Cand$ is empty}{

    $current$ = $\Cand.\mathrm{popMin}()$\;
    
    \If{current = $\langle v_1, \ldots, v_i \rangle$ $\textsf{\upshape (not a follower)}$}
    {
    
        \algocomment{Expand optimally into full solution.}\;
        $\langle v_1, \ldots, v_i, \ldots, v_\lambda \rangle$ = $\langle v_1, \ldots, v_i \rangle \concat \sol_1(v_i)$\;
        
        \algocomment{Find first non-leading prefix.}\;
        $m = \min \{j \;|\; j \in [i, \lambda] \wedge \SortedSuf(v_j) \; \textrm{initialized} \}$ or $\lambda$ if none \label{line:leading_pref}\;
        
        \algocomment{Generate new candidates by taking deviations until prefix is non-leading.}\;
        \For {$j$ from $i$ to $m$}
        {
            $\Cand.\mathrm{add}(\langle v_1, \ldots, v_i, \ldots, v_{j-1}, \suc(v_{j-1}, v_j) \rangle)$ \label{line:successor_partp}\;
        }
        
        \algocomment{For the rest, find the $2^\textrm{nd}$ best suffix from the leading prefix.}\;
        \If{$m < \lambda$}
        {
            $\mathrm{\FollowFun}(\langle v_1, \ldots, v_m \rangle, 2)$ \label{line:follower}\;
        }

        \algocomment{Update sorted lists.}\;
        \For {$j$ from $1$ to $m-1$ \label{line:store_loop}}
        {
            $\mathrm{\StoreFun}(\langle v_1, \ldots, v_j \rangle, \langle v_j, \ldots, v_\lambda \rangle)$ \label{line:store}\;   
        }
    }
    \ElseIf{current $\textsf{\upshape is a follower}$ $\langle v_1, \ldots, v_m \rangle \concat \langle v_m, \ldots, v_\lambda \rangle$ $\textsf{\upshape with annotation}$ $\{k'\}$}
    {
        \algocomment{Keep following the leading prefix.}\;
        $\mathrm{\FollowFun}(\langle v_1, \ldots, v_m \rangle, k' + 1)$\;
        
        \algocomment{Update sorted lists.}\;
        \For {$j$ from $1$ to $m-1$}
        {
            $\mathrm{\StoreFun}(\langle v_1, \ldots, v_j \rangle, \langle v_j, \ldots, v_\lambda \rangle)$ \label{line:store_suf}\;   
        }
    }
    output solution $\langle v_1, \ldots, v_i, \ldots, v_\lambda \rangle$\;
}
\;

\Fn{\FollowFun{$\langle v_1, \ldots, v_m \rangle$, $k$}}
{
    \If{$\sol_{k}(v_m) \notin \SortedSuf(v_m)$}
    {
        Add $(\langle v_1, \ldots, v_m \rangle, k)$ to subscribers of $\SortedSuf(v_m)$ \label{line:sub}\; 
    }
    \Else
    {
        $\Cand.\mathrm{add}$($\langle v_1, \ldots, v_m \rangle \concat \sol_{k}(v_m)$ with annotation $\{k\}$)\;
    }
}

\Fn{\StoreFun{$\langle v_1, \ldots, v_i \rangle$, $\langle v_i, \ldots, v_\lambda \rangle$}}
{
        $\SortedSuf(v_i).\mathrm{append}(\langle v_i, \ldots, v_\lambda \rangle)$ \label{line:sorted_list_append}\;
        \For {Subscriber $(\langle v_1', \ldots, v_i' \rangle, k)$ of $\SortedSuf(v_i)$}
        {
            Unsubscribe and $\Cand.\mathrm{add}$($\langle v_1', \ldots, v_i', \rangle \concat \langle v_i, \ldots, v_\lambda \rangle$ with annotation $\{k\}$) \label{line:unsubscribe}\;
        }
}

\caption{\ANYKPARTP}
\label{alg:anyk-part-plus}
\end{algorithm}

\subsection{ANYK-PART with Memoization (ANYK-PART+)}
\label{sec:part_plus}

\ANYKREC stores and reuses the order of suffixes $\sol_k(v)$ from every node $v$ as they are computed,
and, as a result, it can be faster for large values of $k$, as we show in \Cref{sec:complexity}.
For \ssmonotone ranking functions, we now propose a new algorithm \ANYKPARTP that is based on \ANYKPART, but uses the same type of memoization as \ANYKREC
and thus achieves the best of both worlds.
\ANYKPARTP achieves the bounds of \Cref{theorem:paths}.

\introparagraph{Reusing the suffix order}
The benefit of storing the sorted suffixes $\sol_1(v), \sol_2(v), \sol_3(v)$, etc.
for a node $v$ is that they can be reused for different prefixes that end at $v$.
For example, suppose that $i < j$ for two suffixes $\sol_i(v), \sol_j(v)$ (i.e., $\solW_i(v) \leq \solW_j(v)$).
Then, by \smonotonicity, we can infer that $w(p \concat \sol_i(v)) \leq w(p \concat \sol_j(v))$ for all prefixes $p$ that end at $v$.
We leverage this to reduce the comparisons among candidate solutions in the priority queue $\Cand$ of $\ANYKPART$.

\introparagraph{Leading prefix and followers}
The idea behind \ANYKPARTP is to let only the first prefix that visits a node $v_m$ explore the solution space as in \ANYKPART. 
The suffix order discovered by that prefix will be stored in
a sorted list $\SortedSuf(v_m)$ and reused by all other prefixes that reach $v_m$.

\begin{definition}[Leading Prefix]
    The prefix $\langle v_1, \ldots, v_m \rangle$ is called the leading prefix of $v_m$ if $\langle v_1, \ldots, v_m, \ldots, v_\lambda \rangle$ is the first solution returned by \ANYKPART (or \ANYKPARTP) that contains $v_m$.
\end{definition}

Recall that after popping a prefix $\langle v_1, \ldots, v_i \rangle$ 
from $\Cand$, \ANYKPART first expands it optimally into 
$\langle v_1, \ldots, v_i, \ldots v_\lambda \rangle$
and then generates a deviation at every step
$v_i, \ldots v_\lambda$.
In \ANYKPARTP, we again traverse the nodes $v_i, \ldots v_\lambda$ in order to find the first non-leading prefix $\langle v_1, \ldots, v_i, \ldots, v_m \rangle$.
To detect it, we can simply check whether the data structure $\SortedSuf(v_m)$ has already been initialized;
if it is, then this means that node $v_m$ has been visited before by some other prefix.
For all prefixes $\langle v_1, \ldots, v_j \rangle, j \in [i, m]$, we create deviations using the successor function as in \ANYKPART.
For the rest of the prefixes
$\langle v_1, \ldots, v_j \rangle, j \in [m+1, \lambda]$,
we do not create any deviations. 
Instead, we use the sorted list $\SortedSuf(v_m)$ to directly find the next-best suffix and create a single new candidate for that subspace.
We will refer to the solutions generated by appending a suffix from $\SortedSuf(v_m)$
to a non-leading prefix as a \emph{follower}.
A follower consists of a prefix and a suffix, together with an annotation $\{k'\}$ that specifies the rank of the suffix.
Whenever we pop a follower $p \concat \Pi_{k'}(v)$ with annotation $\{k'\}$ from $\Cand$,
we simply replace it with a new candidate $p \concat \Pi_{k'+1}(v)$.

The following property is helpful in order to better understand the behavior of the algorithm:

\begin{observation}
\label{obs:leading_pref}
For all iterations of \ANYKPARTP except the first one, 
when we pop from $\Cand$ a prefix $\langle v_1, \ldots, v_{i-1}, v_i \rangle$ that is not a follower,
we have that $\langle v_1, \ldots, v_{i-1} \rangle$ is a leading prefix.
\end{observation}
\begin{proof}
We prove this by (strong) induction.
The base case is the second iteration of the algorithm, where the solution we pop from $\Cand$ is
necessarily a deviation of the optimal DP path, which necessarily consists of leading prefixes.
For the inductive hypothesis, assume that the statement is true for all iterations $k \leq 2$ and we prove it for iteration $k+1$.
For the sake of contradiction, suppose that we pop a non-follower $\langle v_1, \ldots, v_{i-1}, v_i \rangle$ and $v_{i-1}$ has been visited before by a different prefix.
The prefix $\langle v_1, \ldots, v_{i-1}, v_i \rangle$ was generated in some previous iteration $k' < k + 1$ by
computing the successor $v_i = \suc(v_{i-1}, v_i')$ where $v_i'$ was the node following $v_{i-1}$ in that iteration.
If the popped prefix in iteration $k'$ was $\langle v_1, \ldots, v_i' \rangle$, then by the inductive hypothesis $\langle v_1, \ldots, v_{i-1} \rangle$ has to be a leading prefix.
Otherwise, $\langle v_1, \ldots, v_i' \rangle$ was created by optimally expanding a shorter prefix in iteration $k'$.
In that case, since $\langle v_1, \ldots, v_{i-1} \rangle$ is not a leading prefix, the algorithm would create
a follower instead of a deviation,
which is a contradiction.
\end{proof}

\Cref{obs:leading_pref} implies that
for a popped prefix $\langle v_1, \ldots, v_{i-1}, v_i \rangle$,
\ANYKPARTP will always generate the deviation $\langle v_1, \ldots, v_{i-1}, \suc(v_{i-1}, v_i) \rangle$.

\introparagraph{Finding the suffix order}
Keeping track of the suffix order for every node is not as straightforward as in \ANYKREC.
In \ANYKPARTP, we infer it from the output of the algorithm (which consists of complete solutions) using \ssmonotonicity.
Whenever we return a solution $\langle v_1, \ldots, v_\lambda \rangle$, then we can break it up into a prefix $\langle v_1, \ldots, v_j \rangle$
and suffix $\langle v_j, \ldots, v_\lambda \rangle$ for any $j \in [\lambda]$ and record the suffix in $\SortedSuf(v_j)$.
Using \ssmonotonicity, we can show that any later solution that passes through $v_j$
with a different prefix will also follow the same suffix order.
We defer the proof of correctness to \Cref{sec:rankingfctgeneralizing}.

We now discuss how to avoid storing duplicate suffixes. 
In principle, this can be achieved by attempting to store all possible
suffixes obtained by breaking up the solutions returned by the algorithm and
checking if these already exist in our sorted lists.
We describe a more efficient method that does not require checking for duplicates.
The goal is to store a suffix only the first time we encounter it, which is when it is preceded by a leading prefix.
There are two cases:
\begin{enumerate}
    \item \emph{Non-follower}: Suppose that we have returned a solution that is not a follower and let $\langle v_1, \ldots, v_m \rangle$ be its first non-leading prefix.
    Then, we store the suffix of the solution for each node $v_j, j \in [m-1]$.
    This is because from the way the algorithm operates,
    every prefix contained in a leading prefix must be leading.
    \item \emph{Follower}: For a follower solution $\langle v_1, \ldots, v_{m-1}, v_m \rangle \concat \langle v_m, \ldots, v_\lambda \rangle$,
    we also know that the prefix $\langle v_1, \ldots, v_{m-1} \rangle$ is leading from the way we construct followers.
    Thus, we store the suffix for each node $v_j, j \in [m-1]$. 
\end{enumerate}

\introparagraph{Delaying candidate generation}
A complication that arises is that followers might not yet have access to the next-best suffix during the iteration in which we pop them.
Note that it is guaranteed (by \ssmonotonicity) that a follower
$p \concat \sol_{k'}(v)$ will never need to be popped from $\Cand$
before $p^* \concat \sol_{k'}(v)$
where $p^*$ is the leading prefix of $v$.
However, it is not necessarily the case that $p^* \concat \sol_{k'+1}(v)$ will be popped before $p \concat \sol_{k'}(v)$.
If we have not yet popped $p^* \concat \sol_{k'+1}(v)$,
then $\sol_{k'+1}(v)$ is not yet in the sorted list $\SortedSuf(v)$,
so we cannot immediately replace the follower $p \concat \sol_{k'}(v)$ with $p \concat \sol_{k'+1}(v)$.
To handle this scenario, we maintain a list of \emph{subscribers} along with $\SortedSuf(v)$.
A subscriber is a follower that waits for the next best suffix to be discovered and inserted into the sorted list by the leading prefix.
When it does, we push all subscribers to $\Cand$ together with the newly discovered suffix.

\Cref{alg:anyk-part-plus} contains the complete pseudocode of \ANYKPARTP.

\begin{example}[\ANYKPARTP on \Cref{ex:dp_naive}]
We return to the running example of \Cref{fig:graph}.
The first iteration of the algorithm is the same as \ANYKPART since all prefixes in the solution are leading,
but we also store the suffixes for the popped solution (\Cref{line:store_suf}).
$\sol_1(1) = \langle 1, 100, 10 \rangle$ is appended to $\SortedSuf(1)$,
$\sol_1(100) = \langle 100, 10 \rangle$ to $\SortedSuf(100)$,
and $\sol_1(10) = \langle 10 \rangle$ to $\SortedSuf(10)$.
In the second iteration, after we expand $\langle 2 \rangle$ into $\langle 2, 100, 10\rangle$,
we generate two new candidates $\langle 3 \rangle$ and $\langle 2, 200 \rangle$ by using the successor function.
In contrast to \ANYKPART, we stop taking successors after node $100$ because it has been visited before (\Cref{line:leading_pref}).
For all the following nodes, we create a follower with prefix $\langle 2, 100 \rangle$ (\Cref{line:follower}) that will read the sorted list of suffixes 
as found by the leading prefix $\langle 1, 100 \rangle$.
However, the second best suffix $\sol_2(100)$ has not yet been discovered, thus the follower $\langle 2, 100 \rangle$ becomes a subscriber (\Cref{line:sub}).
We also update our sorted lists of suffixes with the current solution.
Notice that $\langle 100, 10 \rangle$ is not stored in $\SortedSuf(100)$ ($m = 2$ in \Cref{line:store_loop})
in order to avoid duplicate suffixes.
The leading prefix $\langle 1, 100 \rangle$ is the only one responsible for updating $\SortedSuf(100)$.
In the third iteration, we pop $\langle 1, 100, 20 \rangle$ from $\Cand$.
This results in storing the suffix $\langle 100, 20 \rangle$ in $\SortedSuf(100)$ and subsequently providing this suffix to all the subscribers of the list.
Thus, $\langle 2, 100 \rangle$ from the previous iteration receives this suffix and the follower $\langle 2, 100 \rangle \concat \langle 100, 20 \rangle$ with annotation $\{2\}$
is added to $\Cand$ (\Cref{line:unsubscribe}).
\end{example}

\subsection{Complexity Analysis}
\label{sec:complexity}

\newcommand\bolden[1]{{\boldmath\bfseries#1}}
\definecolor{colorbest}{RGB}{77, 175, 74}

\begin{figure*}[t]
\centering
\scriptsize
\renewcommand{\tabcolsep}{1.3pt}
\begin{tabular}{|l|l|l|l|l|l|l|}
\hline
\bolden{Algorithm} 	    & \bolden{Ranking fun.} & \bolden{$\TT(k)$} 	& \bolden{$\TT(\nodenum)$} & \bolden{$\TTL$ ($\out = \Omega(|G|)$}) & \bolden{$\TTL$ for CP} & \bolden{$\MEM(k)$} \\ 
\hline
\ANYKPART       & \smonotone 	    
                & $\bigO(|G| + k (\log k + \ell))$
                & \cellcolor{colorbest!20}$\bigO(|G| + \nodenum (\log \nodenum + \ell))$ 
                & $\bigO(\out (\log \out + \stages))$
                & $\bigO(n^\ell \ell \log n)$
				& \cellcolor{colorbest!20}$\bigO(|G| + k \stages)$ \\
\hline
\ANYKREC        & \smonotone 	
                & $\bigO(|G| + k \ell \log \nodenum)$
                & $\bigO(|G| + \nodenum \ell \log \nodenum)$ 
                & $\bigO(\out \ell \log \nodenum)$ 
                & \cellcolor{colorbest!20}$\bigO(n^\ell (\log n + \stages))$ 
				& \cellcolor{colorbest!20}$\bigO(|G| + k \stages)$ \\
		
\hline
\ANYKPARTP		& \ssmonotone 	
                & \cellcolor{colorbest!20}$\bigO(|G| + k (\log \nodenum + \ell))$
                & \cellcolor{colorbest!20}$\bigO(|G| + \nodenum (\log \nodenum + \ell))$
                & \cellcolor{colorbest!20}$\bigO(\out (\log \nodenum + \stages))$ 
                & \cellcolor{colorbest!20}$\bigO(n^\ell (\log n + \stages))$ 
				& \cellcolor{colorbest!20}$\bigO(|G| + k \stages)$ \\
	
\hline
\BATCH     	    & Any
                & $\bigO(|G| + \out (\log \out + \stages))$ 
                & $\bigO(|G| + \out (\log \out + \stages))$ 
                & $\bigO(\out (\log \out + \stages))$ 
                & $\bigO(n^\ell \ell \log n)$
				& $\bigO(|G| + \out \stages)$ \\
\hline
\end{tabular} 
\caption{Complexity of ranked-enumeration algorithms for $s$-$t$ path enumeration in a DAG.
Best performing algorithms are colored in green.
$|G|$ is the graph size, $\nodenum$ is the number of nodes, 
$\out$ is the total number of paths, and
$\ell$ is the maximum length of a path.
The ranking function can be subset-monotone (\smonotone),
strong-subset-monotone (\ssmonotone),
or arbitrary.
For instances created from path CQs (\Cref{sec:cq_to_dp}),
$\ell$ is the number of atoms,
$\nodenum = \O(n \ell)$, with $n$ being the maximum relation size,
and $|G|$ is also $\O(n \ell)$.
CP is a fully-connected multi-stage graph resulting from a Cartesian Product.
}
\label{tab:complexity_dp}
\end{figure*}

We now analyze our any-$k$ algorithms in terms of time and memory.
For the analysis, we assume that
the problem at hand is $s$-$t$ path enumeration in a DAG $G$,
but the same analysis also applies to path-structured CQs
as we explained in \Cref{sec:cq_to_dp}.
In the following, $\nodenum$ is the number of nodes,  
$|G|$ is the number of nodes and edges, 
$\ell$ is the maximum length of a path,
and $\out$ is the total number of paths.
For \ANYKPART and \ANYKPARTP, we assume that $\Cand$ is implemented by a priority queue with a logarithmic-time pop and constant-time insert~\cite{LarkinSenTarjan2004:PQs}.
With \BATCH, we refer to an algorithm that computes all paths
and then sorts them with a comparison-based algorithm.
For full acyclic CQs, this is equivalent to the Yannakakis algorithm~\cite{DBLP:conf/vldb/Yannakakis81} followed by sorting.

For the time to the first solution ($k = 1$),
all any-$k$ algorithms execute DP in time $\O(|G|)$.
All initializations of data structures such as the priority queues $\Choices_1(v)$ of \ANYKREC or the heaps of the \LAZY variant of \ANYKPART and \ANYKPARTP
also take $\O(|G|)$.
Every algorithm pays $\bigO(\ell)$ to return a path of size $\bigO(\ell)$ in each iteration.

\introparagraph{\ANYKPART}
Since at most $\ell$ candidates are generated in each iteration, $|\Cand| \le k \ell$.
Thus, popping the best candidate as well as bulk-inserting all new candidates takes $\bigO(\log(k \ell))$. 
For efficient candidate generation (\Cref{line:candidate_generation} in \Cref{alg:anyk-part}), the new candidates do not copy the previous solution prefix, 
but simply create a pointer to it.
Therefore, a new candidate can be created in $\O(1)$.
To evaluate the successor function, the \LAZY variant may need to pop from a binary heap of size $\O(\nodenum)$.
We can assume that all the second-best choices are already computed in the preprocessing phase in linear time by popping from every heap once.
Since at most one deviation per iteration does not involve the second-best choice from a node, the overall cost of the successor function is $\O(\log \nodenum)$ per iteration.
Putting it all together, we have $\TT(k) = \O(|G| + k(\log(k \ell) + \log \nodenum + \ell)) = 
\O(|G| + k(\log k + \log \nodenum + \ell))$.

\begin{lemma}\label{lem:complexity}
For all $\nodenum \geq 1$ and $k \geq 1$, we have $\nodenum + k \log \nodenum = \O(\nodenum + k \log k)$.
\end{lemma}
\begin{proof}
If $k \geq \nodenum$, then $k \log \nodenum \leq k \log k$, thus the statement is obvious.
For any $1 \leq k \leq \nodenum$, it holds that 
$\nodenum / k \geq \log(\nodenum / k)$ and 
$\log k \geq 0$
and therefore
\begin{align*}
& \frac{\nodenum}{k} \geq \log \nodenum - \log k \geq \log \nodenum - 2 \log k \\
\Rightarrow & \nodenum \geq k \log \nodenum - 2 k \log k \\
\Rightarrow & 2 (\nodenum + k \log k) \geq \nodenum + k \log \nodenum
\end{align*}
This means that there exists an $a > 0$ ($a=2$ here) for which 
$\nodenum + k \log \nodenum \leq a (\nodenum + k \log \nodenum)$ for all values of $\nodenum$,
which completes the proof.
\end{proof}

By \Cref{lem:complexity}, the time complexity of \ANYKPART is 
$\TT(k) =\O(|G| + k(\log k + \ell))$.
If $k$ is equal to the size of the full output $\out$, we obtain 
$\TTL = \O(|G| + \out(\log \out + \ell))$,
which is the same as \BATCH.

\introparagraph{\ANYKREC}
Each \texttt{next} call on source node $s$ triggers 
$\O(\ell)$ \texttt{next} calls (at most one per node on the current path). 
A \texttt{next} call deletes the
top choice at the node and replaces it with the next-best suffix through the same child node.
With a priority queue, these operations together take
time $\bigO(\log \nodenum)$ per node accessed, for a total delay of $\bigO(\stages \log \nodenum)$
between consecutive solutions. In total, we obtain
$\TT(k) = \bigO(|G| + k \stages \log \nodenum)$.
The resulting $\TTL$ bound of $\bigO(|G| + \out \stages \log \nodenum)$
is \emph{tight} in the sense that there exist inputs where \ANYKREC runs in $\bigO(|G| + \out \stages \log \nodenum)$.
However, it does not take into account the effect of memoization that we also exploited in \ANYKPARTP; in later iterations
many \texttt{next} calls will stop early because the corresponding suffixes $\sol_i$
have already been computed by an earlier call.
Taking this into account, we can prove the following:

\begin{proposition}\label{prop:rec_ttl}
There exists a class of inputs where \ANYKREC has asymptotically lower time-to-last ($\TTL$) complexity than \BATCH.
\end{proposition}
\begin{proof}
Regardless of the implementation of \BATCH, before it terminates it has to
(i) process the input in $\Omega(|G|)$,
(ii) enumerate all solutions in $\Omega(\out \cdot \ell)$ and
(iii) use a standard comparison-based sort algorithm on the entire output
in $(\out \log \out)$. In total, it needs
$\Omega(|G| + \out (\log \out + \ell))$.

If \ANYKREC returns the full output, each suffix $\sol_i(v)$
of a node $v$ is
inserted into and removed from the priority queue managing $\Choices$ at $v$ \emph{exactly once}.
Therefore, the total number of priority queue operations,
each costing $\bigO(\log \nodenum)$, is equal to the number
of suffixes. 
Let $\sol_*(i)$ denote the number of suffixes starting at nodes which are at distance $i$ from $s$.
Then the total cost for all priority-queue operations is
$\O(\log \nodenum \sum_{i=1}^{\ell} \sol_*(i))$. 
If
$\sum_{i=1}^{\ell} \sol_*(i) = \bigO(\sol_*(1))$,
then this cost is $\bigO(\out \cdot \log \nodenum)$.
To see this, note that the set of paths starting at distance-1 nodes is
the set of all possible paths, i.e., the full output.
Together with preprocessing time and time to assemble each output
tuple, the total $\TTL$ complexity of \ANYKREC then adds up to
$\O(|G| + \out (\log \nodenum + \ell))$.
To complete the proof, we show inputs where the condition $\sum_{i=1}^{\ell} \sol_*(i) = \bigO(\sol_*(1))$ holds and in which the running time of \BATCH is strictly worse.

Consider multi-stage graphs where $\ell$ consecutive stages of size $n$ are fully connected
($s$ and $t$ are the single nodes contained in two additional stages respectively).
These graphs are obtained from Cartesian Product CQs if we follow the mapping in \Cref{sec:cq_to_dp}.
The output size is $\out = n^\ell$ and the number of suffixes in the first stage is also $\sol_*(1) = n^\ell$.
The ratio between $\sol_*(i)$ and $\sol_*(i + 1)$ for some stage $i \in [\ell-1]$ is $n$.
Therefore, the sum $\sum_{i=1}^{\ell} \sol_*(i)$ is a geometric series and the first term 
$\sol_*(1)$ asymptotically dominates.
Also note that the running time of \BATCH in these instances is $\Omega(n^\ell \cdot \ell \log n)$, which is higher than the complexity $\O(n^\ell (\log n + \ell))$ of \ANYKREC.
\end{proof}

The lower $\TTL$ of \ANYKREC is at first surprising,
given that \BATCH is seemingly optimized for bulk-computing and bulk-sorting the entire output.
The reason why \ANYKREC wins is that it exploits the shared structure of thee solutions, which enables the reuse of shared path suffixes, 
while \BATCH uses general-purpose comparison-based sorting.

\introparagraph{\ANYKPARTP}
In \ANYKPARTP,
the maintenance of additional data structures for memoization does not incur any additional cost compared to \ANYKPART.
In any iteration, the first non-leading prefix can be identified in $\O(\ell)$ 
and the lists $\SortedSuf(v)$ for nodes $v$ contained in a solution
can also be updated in $\O(\ell)$.
The benefit of memoization is evident in the size of the priority queue $\Cand$.
Followers do not increase the size of $\Cand$ because they are replaced by at most one candidate.
Non-followers can generate one follower and a number of other deviations through the successor function.
We can again assume that the generated follower does not increase the size of $\Cand$ due to the pop that occurs in each iteration.
The generated deviations are responsive for increasing the size of $\Cand$, 
but this can only happen once for each leading prefix (i.e., the first time we visit a node).
Since the number of leading prefixes is at most $\nodenum$,
the size of $\Cand$ is also bounded by $\nodenum$, and the cost of popping or inserting elements in $\Cand$ is $\O(\log (\min \{\nodenum, k\}))$.

Furthermore, the list of subscribers affects the worst-case delay, but not $\TT(k)$ (see \Cref{sec:complexity_measures}).
Indeed, since the prefix of every follower is by construction composed of a leading prefix followed by one edge,
we can have $\O(\nodenum)$ subscribers in one list which are all added to $\Cand$
when the next-best suffix becomes available.
However, this delayed candidate generation can only improve $\TT(k)$.
To see this, compare it with an algorithm that has access to an oracle which immediately gives the next-best suffix
without the need of subscriber lists.
For each subscriber that we add to $\Cand$ in a later iteration $k$ than we should,
there exists a previous iteration $k'$ where we added one less candidate.
This means that the number of priority queue insertions up to iteration $k$ is the same as if
it had been added to $\Cand$ in iteration $k'$ and all priority queue operations in iterations in-between can only cost less because $\Cand$ is smaller.

Using \Cref{lem:complexity}, we can simplify the complexity of \ANYKPARTP to
$\TT(k) = \O(|G| + k (\log \nodenum + \ell))$.
This is strictly better than the $\TT(k)$ complexity of both \ANYKPART and \ANYKREC,
and also matches the improved bound of \ANYKREC that we showed in \Cref{prop:rec_ttl}.

\introparagraph{Memory}
All algorithms need $\O(|G|)$ memory for storing the input.
The memory consumption of \ANYKPART depends on the size of $\Cand$.
Since at most $\O(\ell)$
new candidates are generated per iteration, we have $\MEM(k) = \O(|G| + k \ell)$.
For \ANYKREC, the size of a choice set $\Choices_k(v)$ is bounded by the out-degree
of $v$, hence cannot exceed $\nodenum$.
However, we need to store the suffixes $\sol_i(v)$,
whose number is $\O(\ell)$ per iteration, so $\MEM(k) = \O(|G| + k \ell)$.
\ANYKPARTP has the same memory consumption; the size of $\Cand$ is bounded by $\nodenum$
but the sorted lists $\SortedSuf(v)$ occupy $\O(k \ell)$ space.
\BATCH first materializes the output and then sorts it in-place,
therefore has $\MEM(k) = \O(|G| + \out \ell)$, regardless of $k$.

\introparagraph{Summary}
\Cref{tab:complexity_dp} summarizes the analysis for $\TT(k)$, 
highlighting the case of $k = \nodenum$
(this is a case where $k$ is sufficiently large for the enumeration cost to exceed the preprocessing cost).
We also show $\TTL$ for cases where the output is sufficiently large (so that enumeration dominates
preprocessing),
$\TTL$ on Cartesian Product inputs where we can see the advantage of \ANYKREC (\Cref{prop:rec_ttl}),
and $\MEM(k)$.
\ANYKPART is faster than \ANYKREC in general (e.g., for $\TT(\nodenum)$),
but \ANYKREC has the edge for large values of $k$ on certain inputs.
\ANYKPARTP matches the best running times of both and achieves the lowest $\TT(k)$ complexity overall.
All any-$k$ algorithms require minimal space, depending only on input size
and the number of iterations $k$ times the solution size.
\BATCH requires more memory because it materializes the full output.

\section{Extension to General CQs}
\label{sec:general_cqs}

We extend our ranked enumeration framework from path-structured CQs to general CQs.
First, we study DP problems with a tree structure (T-DP),
which allows us to handle any acyclic CQ.
Then, we discuss how to handle projections for non-full CQs and finally, how our techniques can be applied for cyclic CQs.

\subsection{Tree-Based DP (T-DP)}
\label{sec:tdp}

We now consider a class of Dynamic Programming problems where the states (i.e., the nodes)
are organized into stages similarly to serial DP
and the stages are organized in a tree structure.
Compared to
the DP framework of \Cref{sec:dp},
T-DP is more general because it allows tree-structured problems,
but also less general because it does not allow arbitrary ``jumps'' between stages.
Both are instances of Non-Serial Dynamic Programming~\cite{bertele72nsdp} which generalizes the stage-by-stage computation of serial DP;
the DP framework of \Cref{sec:dp} allows ``feedforward loops'',
while T-DP allows ``diverging branches''~\cite{esogbue74nsdp}.
We also note that the top-1 solution for T-DP can be phrased in the framework of Functional Aggregate Queries~\cite{abo16faq}.

In T-DP, we have a set of stages $\Sset_1, \ldots, \Sset_\ell$ that partition the states $V$.
The stages are organized in a \emph{rooted tree}
with $\Sset_0 = \{ s \}$ as the root stage. 
We let the leaves of the tree be a set of terminal stages
$\Sset_{\ell+1}, \ldots, \Sset_{\ell+b}$, each one containing a single terminal node $t_i, i \in [\ell + 1, \ell + b]$.\footnote{
Artificial stages can always be introduced to meet this assumption.}
For every parent stage $\Sset_p$ and child stage $\Sset_c$,
there is a distinct set of decisions (i.e., edges) $E_{pc}$ directed from $\Sset_p$ states to $\Sset_c$ states.
Every root-to-leaf path in the stage tree represents an instance of serial DP (\cref{sec:dp}).
We now extend our approach to such problems.

We establish a \emph{tree order} that serializes the stages and assume that
their indexing follows this order.
Non-leaf stages $\Sset_1, \ldots, \Sset_\ell$ are ordered by a breadth-first traversal of the tree;
if the distance from the root to $\Sset_i$ is greater than $\Sset_j$, then  $i > j$.
Irrespective of their depth, the $t$ leaf nodes are always indexed last in arbitrary order.
We now introduce some helpful notation for trees.
We define $\Ch(\Sset_p)$ 
to be the set of indices of child stages of a stage $\Sset_p$ and $\parent(\Sset_c)$ as the index of the parent stage of a stage $\Sset_c$. 
By $\llceil \Sset_i \rrceil$ we denote all stage indexes in the subtree rooted at $\Sset_i$,
excluding $i$.
Slightly overloading the notation, we also use 
$\Ch(v_i) := \Ch(\Sset_i)$
for a state $v_i\in \Sset_i$,
and analogously for $\parent(v_i)$ and $\llceil v_i \rrceil$.

\begin{figure*}[tb]
\centering
\hfill
\begin{subfigure}[t]{.4\linewidth}
    \centering
    \includegraphics[height=3.9cm]{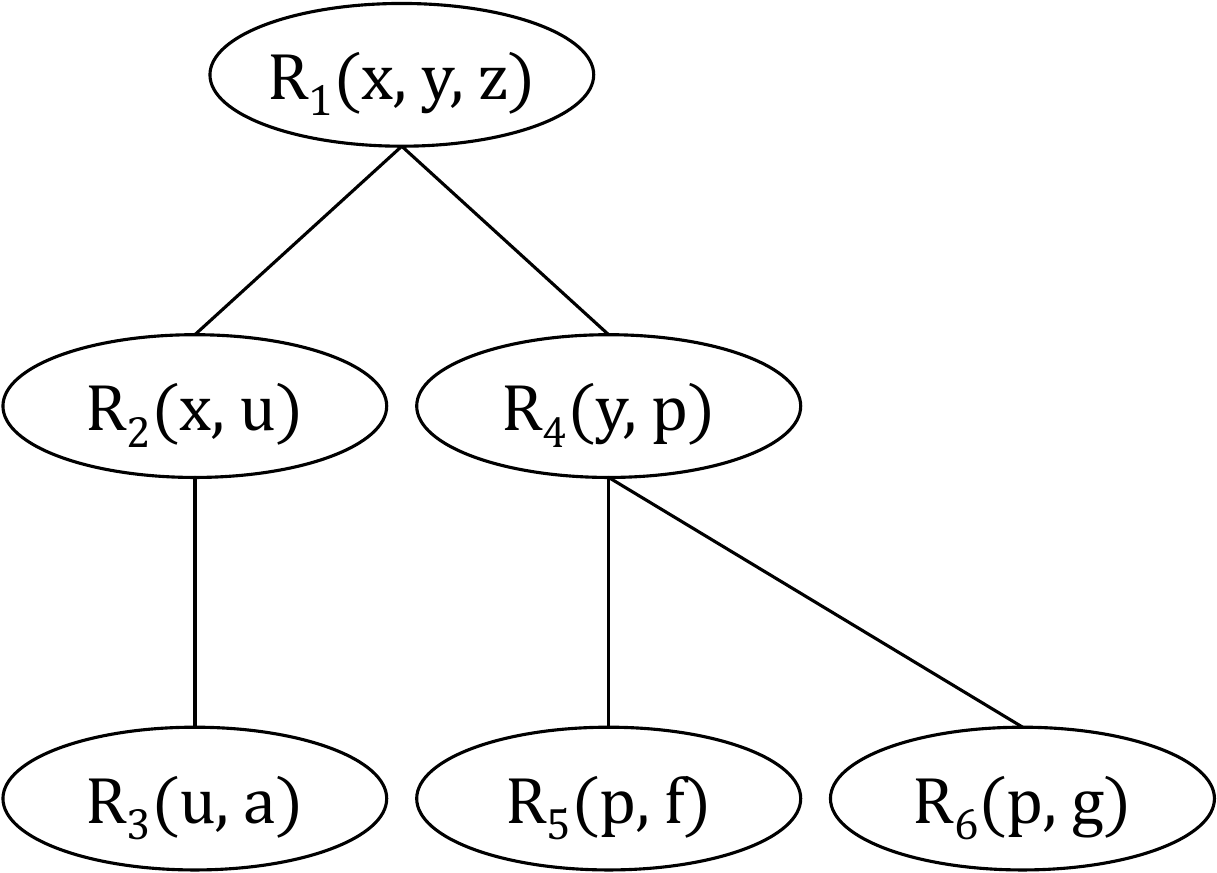}
    \caption{A join tree constructed for the query\\$Q_T(x, y, z, u, a, p, f, g) \datarule 
    R_1(x, y, z),$ $R_2(y, u), R_3(u, a), R_4(y, p), R_5(p, f), R_6(p, g)$.}
    \label{fig:join_tree_2}
\end{subfigure}%
\hfill
\begin{subfigure}[t]{.42\linewidth}
    \centering
    \includegraphics[height=3.9cm]{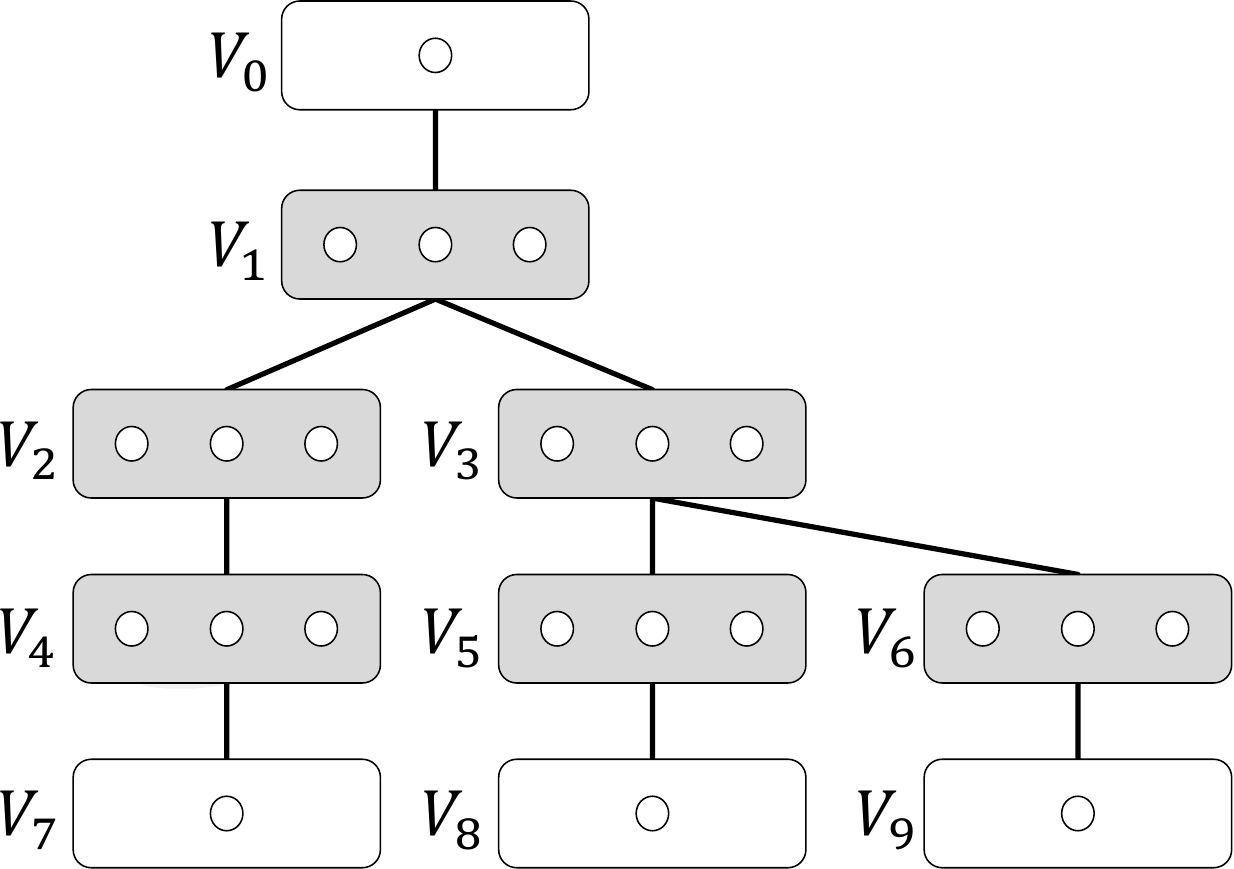}
    \caption{T-DP problem structure. 
    Rounded rectangles are \emph{stages}, small circles are \emph{states}.
    The edges connect the stages in a tree structure.
    }
    \label{fig:treedp}
\end{subfigure}
\hfill
\begin{subfigure}[t]{.15\linewidth}
    \centering
    \includegraphics[height=3.9cm]{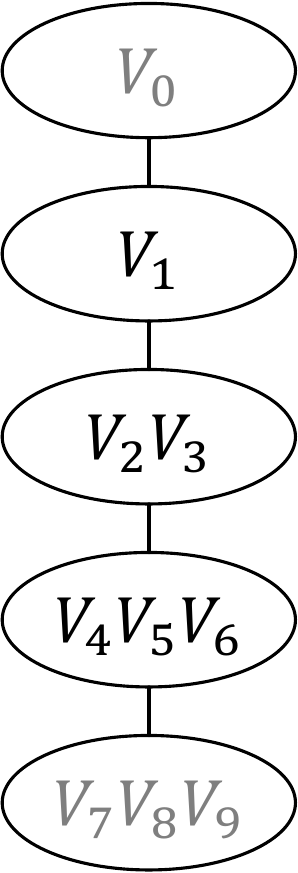}
    \caption{Serial\\decomposition of width $3$.}
    \label{fig:tree_to_path}
\end{subfigure}
\hfill
\caption{\emph{Tree-Based DP} (T-DP) instance from an acyclic CQ.}
\label{fig:tdp}
\end{figure*}

\begin{example}[T-DP Instance]
\Cref{fig:treedp} shows an instance of T-DP with 6 ``internal'' stages $V_1$ to $V_6$, 
a root stage $V_0$ and 3 terminal stages $V_7$ to $V_9$.
Notice that the indexing of stages follows a breath-first order.
If we restrict ourselves to a root-to-leaf path, such as $V_0$ to $V_7$, then the problem degenerates
to serial DP (\Cref{fig:dp_from_tables}).
In this example, $\llceil \Sset_3 \rrceil := \{  5, 6, 8, 9\}$,
$\Ch(\Sset_3) = \{ 5, 6 \}$ and $\parent(\Sset_3) = 1$. 
\end{example}

A T-DP \emph{solution} $\sol = \langle v_1, \ldots, v_\ell \rangle$ 
\footnote{
Notice that as in DP, we do not include the unique root state and the $b$ terminal nodes of the $b$ leaf stages in the solution.
}
is a tree with one state per stage that satisfies
$(v_p, v_c) \in \Dec_{pc}$,
$\forall c \in [\stages+b]$,
where $p = \parent(v_c)$, $v_0 = s$, and $v_{\ell+i} = t_{i+1}, i \in [b]$.
The \emph{objective function} aggregates the weights of decisions across the entire tree structure:
\begin{align}
	\weight(\sol) = \aggrsum_{c=1}^{\stages+b} \weight(v_{\parent(v_c)}, v_{c}) \label{eq:cost_treeDP}
\end{align}

\introparagraph{Principle of optimality and T-DP Algorithm}
The optimal solution is computed bottom-up, following the reverse serial ordering
of the stages.
A bottom-up step for a state $v$ solves a \emph{subproblem} which corresponds to finding an optimal subtree $\sol_1(v)$.
If $\Ch(v) = \{ i_1, \ldots, i_\chnum \}$, then that subtree consists of 
$v$ and a list of other subtrees rooted at its children $v_{i_1}, \ldots, v_{i_{\chnum}}$. 
To solve a subproblem, we can \emph{independently} choose the best decision for each child stage.
The equations 
describing the bottom-up phase in T-DP are
recursively defined for all states and stages
by
\begin{equation}
\begin{aligned}
    \solW_1(v) &= 0, \textrm{ for the $b$ terminals with } \Ch(v) = \emptyset		 \\
    \solW_1(v) &= \!\!\!\aggrsum_{c \in \Ch(v)} 
		\!\min_{(v, v_c) \in \Dec_{pc}} 
		\!\big\{\weight(v, v_c) \aggr \solW_1(v_c)\big\},  
		\textrm{ for } 
			v \in \Sset_p, 
			p \in [\ell]_0
	\label{eq:TDP_recursion}
\end{aligned}
\end{equation}

Similarly to DP, after the bottom-up phase we get reduced sets of states
$\SsetR_i \subseteq \Sset_i$, $\DecR_{pc} \subseteq \Dec_{pc}$
and the top-1 solution $\sol_1(s)$ is found by a top-down
phase that follows optimal decisions.

Comparing the above with serial DP, we now have multiple terminal states (i.e., leaves in the tree) 
that are initialized with zero weight, 
but we still have only one single root node. 
A minimum-weight solution contains other subtree solutions 
that themselves achieve minimum weight for their respective subproblems.

The correctness of the T-DP algorithm is well-established;
for example, it is a special case of the InsideOut algorithm for Functional Aggregate Queries~\cite{abo16faq}.
We give here an independent proof:

\begin{proposition}[T-DP]\label{TH:T-DP}
\Cref{eq:TDP_recursion}
	finds a solution that minimizes \cref{eq:cost_treeDP}.
\end{proposition}
\begin{proof}
We show by induction on the tree stages
in reverse serial order
that for all nodes $v \in V$:
\begin{align*}
 \min_{\sol(v)}
 \big\{
 \aggrsum_{j \in \llceil v \rrceil}
 \weight(s_{\parent(v_i)}, v_{j})
 \big\}
	= \solW_1(v)
\end{align*}
The base case for the (terminal) leaf states follows by definition from
\Cref{eq:TDP_recursion}.
For the inductive step, assume that the above holds for all
descendant states
$v_{d} \in \Sset_d, d \in \llceil v \rrceil$ of a node $v$.

Then for any state $v \in \Sset_p$:
\begin{align*}
&\solW_1(v)
=
	\!\!\!\!\aggrsum_{c \in \Ch(v)}
 \min_{(v, v_c) \in \Dec_{pc}}
	\!\big\{
	\weight(v, v_c) \aggr \solW_1(v_c)
	\!\big\}
	\overset{\mathrm{ind. step}}{=\joinrel=}
	\!\!\!\!\aggrsum_{c \in \Ch(v)}
	\min_{(v, v_c) \in \Dec_{pc}}
	\!\big\{
	\weight(v, v_c) \aggr
	\min_{\sol(v_c)}
	\!\big\{
	\!\!\!\!\aggrsum_{c' \in \llceil v_c \rrceil}
	\weight(v_c, v_{c'})
	\big\}
	\big\}\\
&\overset{  \begin{subarray}{c}
                \mathrm{distributivity}/\\
                \mathrm{subset-monotonicity}
            \end{subarray}
    }{=\joinrel=\joinrel=}
    \min_{\substack{(v_{c_1}, \ldots, v_{c_\chnum}) \\ (c_1, \ldots, c_\chnum) = \Ch(v) \\ (v, v_{c_i}) \in \Dec_{p{c_i}}}}
    \:
	\!\big\{
	\!\!\!\!\aggrsum_{c \in \Ch(v)}
	\!\big(
	\weight(v, v_c) \aggr
	\!\!\!\!\aggrsum_{c' \in \llceil v_c \rrceil}
	\weight(v_c, v_{c'})
	\big)
	\big\}
    =
 \min_{\sol(v)}
 \!\big\{
 \aggrsum_{j \in \llceil v \rrceil}
 \weight(v_{\parent(v_j)}, v_{j})
 \!\big\}
\end{align*}
The important property that allows us to swap the minimization over all solutions in the subtree
with the sum over the mininum solution of each child (second-to-last step)
is the
distributivity of sum over $\min$ or, alternatively, the
subset-monotonicity of the ranking function.

Since the above holds for any node $v$, it also holds for the source $s$, and the statement follows.
\end{proof}

\subsubsection{Any-$k$ for T-DP}

To enumerate lower-ranked solutions for T-DP, we extend the path-based any-$k$ algorithms for serial DP.

\introparagraph{Changes to \ANYKPART}
\ANYKPART is straightforward to extend to the tree case by following the serialized order of the stages.
In particular, the $i^\textrm{th}$ stage in the tree order is treated like the
$i^\textrm{th}$ stage in serial DP, except that the sets of choices are determined
by the parent-child edges in the tree. 
For illustration, assume a tree order as indicated
by the stage indices in \Cref{fig:treedp}. 
Given a prefix $\langle v_1, v_2, v_3 \rangle$,
the choices for $v_4 \in \Sset_4$ are not determined by $v_3$ 
(as they would be in serial DP
with stages $\Sset_1, \Sset_2, \ldots$), 
but by $v_2 \in \Sset_2$, because $\Sset_2$ is the parent
of $\Sset_4$ in the tree. 
In general, at stage $\Sset_c$, we find the successor $\Suc(v_p, v_c)$ 
where $p = \parent(v_c)$. 
Similarly, to optimally expand a prefix
$\langle v_1, \ldots, v_{c-1} \rangle$ by one stage, 
we append $v_{c}$ such that $\sol_1(v_{c})$ is a subtree of $\sol_1(v_p)$.
Thus, we can run \Cref{alg:anyk-part} unchanged as long as we
define the choice sets based on the parent-child relationships in the
tree. 
Hence the complexity analysis in \Cref{sec:complexity} still applies
as summarized in \Cref{tab:complexity_dp}.

\introparagraph{Changes to \ANYKREC}
Recall that in serial DP, each node $v_i$ processes a
\texttt{next} call by recursively calling \texttt{next} on the node $v_{i+1}$ 
that follows in the solution.
In T-DP, the solutions have a tree structure, and therefore, to find the next-best subtree solution at $v_i$,
we have to consider the next-best subtrees for all of its children $v_c$, for $c \in \Ch(v_i)$.
For example, consider a node $v_1 \in \Sset_1$ where the children of $\Sset_1$ are $\Sset_2$ and $\Sset_3$.
A solution rooted at $v_1$ consists of two parts: one solution rooted at some node of
$\Sset_2$ and another rooted at $\Sset_3$. 
Suppose that the current solution rooted at $v_1$ contains the $2^\textrm{nd}$-best
solution from $\Sset_2$ and the $4^\textrm{th}$-best solution from $\Sset_3$
and denote that by $[\Pi_2, \Pi_4]$.
Then the next-best solution from $v_1$ could be either $[\Pi_3, \Pi_4]$ or
$[\Pi_2, \Pi_5]$. Since any combination of child solutions $[\Pi_{j_1}, \Pi_{j_2}]$ is valid
for the parent, the problem is essentially to rank the Cartesian product space of
subtree solutions which agree with $v_1$. 

More generally, let $\sol_j(v, c)$ be the $j^\textrm{th}$ best subsolution starting from $v$ and restricted only to a single branch $c \in \Ch(v)$.
$\sol_j(v, c)$ consists of state $v$, then a child $v_c \in \Sset_c$ and, from there, a list of pointers to other solutions (i.e., subtrees) that have their own ranks $j_1, \ldots, j_\chnum$.
We write that as 
$\sol_j(v, c) = v \tree [ \sol_{j_1}(v_c, i_1), \ldots, \sol_{j_\chnum}(v_c, i_\chnum) ]$ 
for $\Ch(v_c) = \{ i_1, \ldots, i_\chnum \}$.
For example, in \Cref{fig:treedp}, 
$\sol_k(v_1, 3) = v_1 \tree [\sol_{j_1}(v_3, 5), \sol_{j_2}(v_3, 6) ]$ for some ranks $j, j_1, j_2$.
Notice that this definition matches the one in \Cref{sec:rec} for $|\Ch(s_c)| = 1$ and
since we can, without loss of generality, assume that $\Sset_0$ always has a single child $\Sset_1$, we have $\sol_k(s) = \sol_k(s, 1)$ for all values of $k$.
A node $v \in \Sset_p$ maintains one data structure per branch $c \in \Ch(v)$ for storing and comparing solutions $\sol_j(v, c)$.
At the beginning of the algorithm, we initialize it as
$\Choices_1(v, c) = 
\{ v \tree [ \sol_1(v_c, i_1), \ldots, \sol_1(v_c, i_\chnum) ] 
\: | \: 
(v, v_c) \in \DecR_{pc}, \Ch(v_c) = \{ i_1, \ldots, i_\chnum \} \}$. 
To process a \texttt{next} call, we pop the best solution from the data structure but unlike serial DP, we now have to replace it with more than one new candidate.
To compute \texttt{next} of 
$\sol_j(v, c) = v \tree [ \sol_{j_1}(v_c, i_1), \ldots, \sol_{j_\chnum}(v_c, i_\chnum) ]$,
the new candidates are:
\begin{align*}
& v \tree [ \sol_{j_1 + 1}(v_c, i_1), \ldots, \sol_{j_\chnum}(v_c, i_\chnum) ] \\
&\qquad\qquad\qquad\vdots \\
& v \tree [ \sol_{j_1}(v_c, i_1), \ldots, \sol_{j_\chnum + 1}(v_c, i_\chnum) ].
\end{align*}
While this algorithm is correct, it has two drawbacks.
First, there can be up to $\ell$ children, hence up to $\ell$ new candidates, and each one could have size (i.e., number of pointers) up to $\ell$. 
Thus, to create them we would have to pay $\bigO(\stages^2)$, which results in an $\bigO(\ell)$ factor increase in $\TT(k)$ complexity.
Second, this candidate generation process creates duplicates, similarly to the issue we described in \Cref{sec:part}.
We address both of these issues by applying the more efficient candidate generation process of \ANYKPART
for each priority queue $\Choices_j(v, c)$.
As a result, \ANYKREC remains the same as the serial DP case for stages with a single child,
but behaves similarly to \ANYKPART when encountering branches. In the extreme case of star queries
(where a root stage is directly connected to all leaves),
\ANYKREC degenerates to \ANYKPART.

\introparagraph{Changes to \ANYKPARTP}
The expansion and successor-taking phases of \ANYKPARTP are modified in the same way as \ANYKPART.
Unfortunately, the memoization of the order of suffixes does not extend to trees in a direct way.
The reason is that for T-DP, the \emph{independence of suffixes from prefixes} does not hold.
In DP, every prefix $p$ that arrives at a node $v$ can be combined with any suffix that starts from $v$.
This is not true in T-DP; a prefix $\langle v_1, v_2, v_3 \rangle$ may not admit
the same suffixes as a different prefix $\langle v_1', v_2', v_3 \rangle$ if the parent of $\Sset_4$ is not $\Sset_3$.
In the example of \Cref{fig:treedp}, the choice of the node from $\Sset_4$ is restricted by the node chosen at $\Sset_2$, and different $\Sset_2$ nodes allow different choices at $\Sset_4$. 

However, we can still use memoization whenever the prefix-suffix independence holds.
For example, in \Cref{fig:treedp}, if we fix the choice of nodes in \emph{both} $\Sset_2$ and $\Sset_3$,
then the prefixes (the nodes before $\Sset_2$) and the suffixes (the nodes after $\Sset_3$)
are independent.
To capture this intuition, we define the following notion:

\begin{definition}[Serial Decomposition]
A serial decomposition of a T-DP instance is a path where every vertex consists of a set of stages such that 
(1) every stage appears exactly once, 
(2) adjacent stages in T-DP appear in the same or adjacent vertices in the path, and
(3) stage $\Sset_j$ cannot appear strictly before stage $\Sset_i$ in the path if $j > i$.
The width of a serial decomposition $\serialw$ is the maximum number of non-terminal stages contained in a vertex.
\end{definition}

This notion of width is reminiscent of query width for CQs~\cite{chekuri97querydecomp},
but with the additional restrictions that we impose here.
And compared to path decompositions for CQs~\cite{olteanu15dtrees}, the difference here is twofold.
First, the serial decomposition is tied to a particular join tree.
Second, we do not materialize the bags of the decomposition,
but only use the serial decomposition structure to determine which sorted suffix lists to maintain for \ANYKPARTP, as we explain below.

We adapt the suffix memoization of \ANYKPARTP
to the serial T-DP decomposition.
In particular, we maintain a list of sorted suffixes (\Cref{line:sorted_list_append} 
in \Cref{alg:anyk-part-plus}) 
for every \emph{combination of nodes} that belong to stages in the same vertex of the decomposition.
Follower generation (\Cref{line:follower}) and suffix storing (\Cref{line:store}) are only
done when the node being processed belongs to the \emph{last stage of a decomposition vertex}.
For all other nodes, we take their successor (\Cref{line:successor_partp}) as in \ANYKPART.

\begin{example}[\ANYKPARTP for T-DP]
Consider again the example T-DP instance of \Cref{fig:tdp}.
A serial decomposition of width $3$ is shown in \Cref{fig:tree_to_path}.
Note that the terminal stages colored in gray are ignored when computing the width.
Suppose that our current solution after expansion is $\langle v_1, \ldots, v_\ell \rangle$
and we are now processing node $v_2$.
Since $\Sset_2$ is in the same vertex as $\Sset_3$ in the decomposition, we do not perform
any of the additional operations of \ANYKPARTP and only take the successor
$\langle v_1, \suc(v_1, v_2) \rangle$.
When we move on to $v_3$, we have seen all the nodes of the $\Sset_2, \Sset_3$ vertex,
therefore we check if the current prefix $\langle v_1, v_2, v_3 \rangle$ is leading
in the sense that the combination $(v_2, v_3)$ is visited for the first time.
If not, then we generate a follower solution for the rest of the stages.
If it is a leading prefix, then we move on to the next node but also
update the sorted list $\SortedSuf((v_2, v_3))$ with $\langle v_4, \ldots, v_\ell \rangle$.
\end{example}

The time complexity of \ANYKPARTP on T-DP depends on the width of the serial decomposition.
In particular, the number of leading prefixes per level of the decomposition is at most $n^{\serialw}$, where $n$ is the maximum number of nodes in a stage.
The size of the priority queue $\Cand$ is thus bounded by $\O(n^{\serialw} \cdot \ell)$
and the cost of priority queue operations is bounded by $\O(\log n^{\serialw} + \log \ell)$.
Overall, we obtain $\TT(k) = \O(n \ell + k ( \log(\min \{ k, n^{\serialw}\}) + \ell))$.
The memory consumption is the same as the serial DP case since we store less suffixes.

\subsubsection{Any-$k$ for full acyclic CQs}
Given a full acyclic CQ, we can easily map it to a T-DP instance using a join tree.
As in serial DP, the stages correspond to relations, and states correspond to input tuples. 
The connections between states of parent-child stages are created by connecting joining tuples in linear time as in \Cref{sec:cq_to_dp}.
To achieve the best possible bound in combined complexity (\Cref{theorem:cq_combined_comp}),
we cannot afford to try all possible join trees and all possible serial decompositions
to find the one that gives the lowest width.
Instead, we work with an arbitrary join tree computed by one run of the GYO reduction~\cite{graham79gyo,tarjan84acyclic,yu79gyo}.

The width of the serial decomposition of our join tree will determine the running time of \ANYKPARTP, but now we
connect this running time to a \emph{structural property} of the CQ.
For any possible join tree, we can pick as the root node the one that maximizes the depth of the tree; this is easy to find in linear time.
From the definition of the diameter of a CQ, it easily follows that this maximum depth is at least $\diam(Q)$.
For the serial decomposition, we choose one that works ``level-by-level'';
this means that all stages at depth $i$ from the root are placed together in the $i^\textrm{th}$ vertex on the path.
Since the depth of the tree is at least $\diam(Q)$, it is guaranteed that it has width $\serialw \leq \ell - \diam(Q) + 1$.

To derive the final complexity of \Cref{theorem:cq_combined_comp}, we also need to take into account that every state corresponds to a tuple that contains
$\O(\alpha)$ values from the database.
Thus, \ANYKPART yields
$\O(n \ell \alpha + k (\log(\min\{k, n^{\ell-\diam(Q)+1}\}) + \ell \alpha))$.

\begin{example}[Full acyclic CQ to T-DP]
\label{ex:full_acyclic_to_tdp}
A join tree of the query $Q_T(x, y, z, u, a, p, f, g) \datarule$ $R_1(x, y, z),$ $R_2(y, u), R_3(u, a), R_4(y, p), R_5(p, f), R_6(p, g)$
is shown in \Cref{fig:join_tree_2} and the resulting T-DP instance is that of \Cref{fig:treedp}.
The connections between $\Sset_1$ nodes and $\Sset_2$ nodes are constructed by introducing intermediate nodes, 
in the same way as the connections between $R_1$ and $R_2$ in \Cref{fig:equiJoinGraph}.
Similarly for $\Sset_1$ and $\Sset_3$.
The width of the serial decomposition in \Cref{fig:tree_to_path} is $3$ because 
three non-leaf stages $V_4, V_5, V_6$ are placed in the same vertex.
However, the diameter of $Q_T$ is $5$ (this is the shortest distance from variable $g$ to variable $a$.
Therefore, it is possible to achieve a width up to $\ell - \diam(Q) + 1 = 6 - 5 + 1 = 2$ if we choose a better root;
this is indeed the case if we choose $R_6$ as the root of the join tree.
\end{example}

\subsection{CQs with Projections}
\label{sec:projections}

So far, we have only considered \emph{full} CQs, i.e., those that do not have projections.
In this section, we first investigate the different possible semantics of ranked enumeration with projections.
Then, we show that under the semantics we introduced in \Cref{sec:ranked_enumeration},
free-connex CQs can be handled as efficiently as full acyclic CQs (in data complexity).

\introparagraph{Alternative ways to define ranked enumeration}
The set of witnesses $\witness(q)$ for a query answer $q$ is not necessarily unique for non-full CQs.
Thus, there are at least two reasonable semantics for ranked enumeration over CQs with projections.
Consider the $2$-path query $Q_{P2E}(x_1) \datarule R_1(x_1, x_2), R_2(x_2, x_3)$ where we want to return only the values of the first variable $x_1$.
Recall that we assume that input weights have been placed on the relation tuples.
What do we do if the same value $c_1$ of $x_1$ appears in two different witnesses
$((c_1, c_2), (c_2, c_3))$ and $((c_1, c_2'), (c_2', c_3'))$ 
with weights $w$ and $w'$, respectively?
We identify two different semantics:

\begin{enumerate}
\item \emph{\Allweights} 
semantics: The first option is to return $c_1$ twice with both weights $w, w'$ in the correct sequence.
The corresponding SQL query would be:
\begin{verbatim}
    SELECT   R1.X1, R1.W + R2.W as Weight
    FROM     R1, R2
    WHERE    R1.X2=R2.X2
    ORDER BY Weight ASC
\end{verbatim}
In general, for a CQ $Q$, we return the answers and the weights that the corresponding full CQ would return
projected on the variables $\free(Q)$.\footnote{In the case that two answers have the same weight, we still return both of them.}
Thus, it is trivial to extend our approach to \allweights semantics, as it is essentially equivalent to ranked enumeration of full CQs.
We enumerate the full CQ $Q(\vec x)$ as before and then apply a projection $\pi_{\vec \free(Q)}(q)$ 
to every answer $q \in Q$ before returning it.
The guarantees we get in this case are exactly the same as for full CQs.

\item \emph{\Minweight} semantics: 
The second option, which was, to our knowledge, first proposed by Kimelfeld and Sagiv~\cite{KimelfeldS2006}, is to return $c_1$ only once with the best (minimum) of the two weights.
This is precisely the definition of answer weights we gave in \Cref{sec:ranked_enumeration}.
In this case, the SQL query is:
\begin{verbatim}
    SELECT 	 Y.X1, Y.Weight
    FROM  (SELECT   R1.X1, MIN(R1.W + R2.W) as Weight
           FROM     R1, R2
           WHERE    R1.X2=R2.X2
           GROUP BY R1.X1) Y
    ORDER BY Y.Weight
\end{verbatim}
In general, we define the weight of $q \in Q$ as
$w(q) = \min_{(t_1,\ldots, t_\stages) \in \witness(q)} \sum_{i=1}^{\ell} w(t_i)$
and rank the answers by those weights.
Each returned answer has the minimum weight over all answers to the corresponding full CQ $Q(\vec x)$ that agree with $q$ on $\free(Q)$.
We note that an interpretation of \minweight semantics is that the query has a \emph{group-by clause},
and the aggregation in the group-by is the same as the one used in the ranking function.

A simple way to handle \minweight semantics is to
apply the projections on the output of the enumeration and discard all lower-ranked duplicates.
However, the non-trivial $\TT(k)$ guarantees we have proved do not hold in that case.
The reason is that the answers that project to the same values can be as many as $\O(\out)$ in the worst case,
delaying the enumeration of distinct answers.
We next discuss a non-trivial extension that can efficiently handle \minweight semantics for free-connex CQs.
\end{enumerate}

\introparagraph{Handling free-connex CQs}
We modify the techniques that have been developed for \emph{unranked} enumeration (\Cref{theorem:known-enumeration}) in order to
accommodate efficient ranked enumeration under \minweight semantics.
Intuitively, unranked enumeration for free-connex CQs
works by constructing an appropriate join tree
that groups the free variables together.
The tree is first swept bottom-up with
semi-joins as in the
Yannakakis algorithm \cite{DBLP:conf/vldb/Yannakakis81}
and then pruned so that only the free variables remain.
The answers to the query can then be enumerated as if it were full (without projections).
We present a modification of this approach for ranked enumeration under \minweight semantics.
Essentially, we replace the semi-joins with our Dynamic Programming framework.

\begin{figure*}[tb!]
    \centering    
    \hfill
	\subcaptionbox{Hyperedge $R'$ verifies $Q_{FC}$ is free-connex.
		\label{fig:fc_query}}
		[0.22\linewidth]
		{\includegraphics[height=3.8cm]{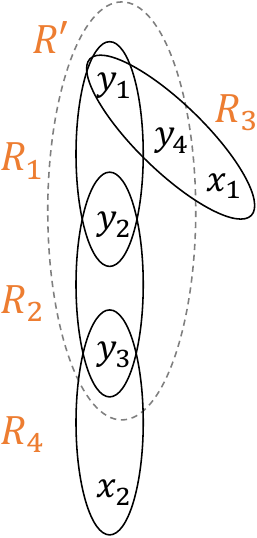}}
    \hfill
	\subcaptionbox{The join tree with a connected subset of nodes $U$ that contain precisely the free variables.
		\label{fig:fc_join_tree}}
		[0.48\linewidth]
		{\includegraphics[height=3.8cm]{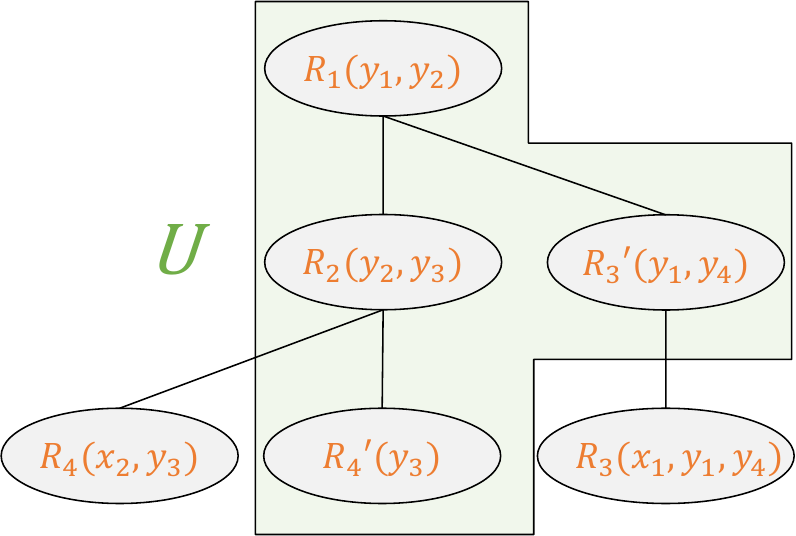}}
    \hfill
	\subcaptionbox{An example database instance.\label{fig:fc_relations}}
		[0.24\linewidth]
		{\includegraphics[height=2.9cm]{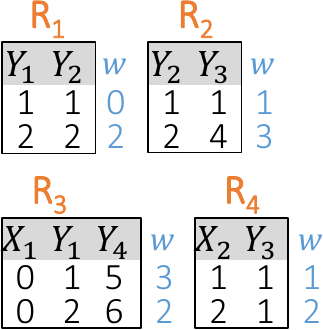}}

	\subcaptionbox{T-DP instance $T$ of the full CQ using the join tree of \cref{fig:fc_join_tree}.
        For every state $v$, $\solW_1(v)$ is depicted on its top-right.
		\label{fig:fc_tdp1}}
		[0.49\linewidth]
		{\includegraphics[height=4.4cm]{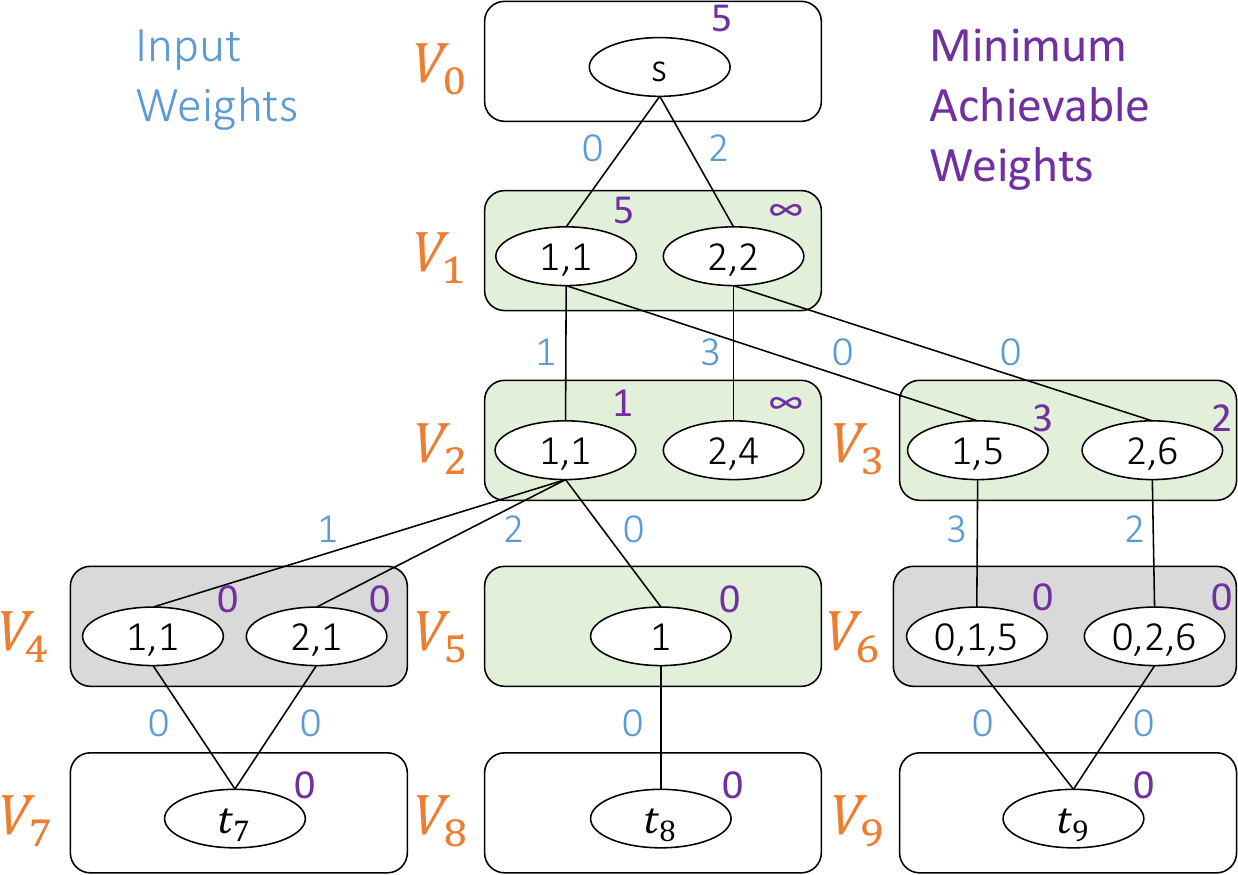}}
    \hfill
	\subcaptionbox{T-DP instance $T'$ used for ranked enumeration.
		Stages corresponding to relations not in $U$ have been removed and edge weights have been modified.
		\label{fig:fc_tdp2}}
		[0.49\linewidth]
		{\includegraphics[height=4.4cm]{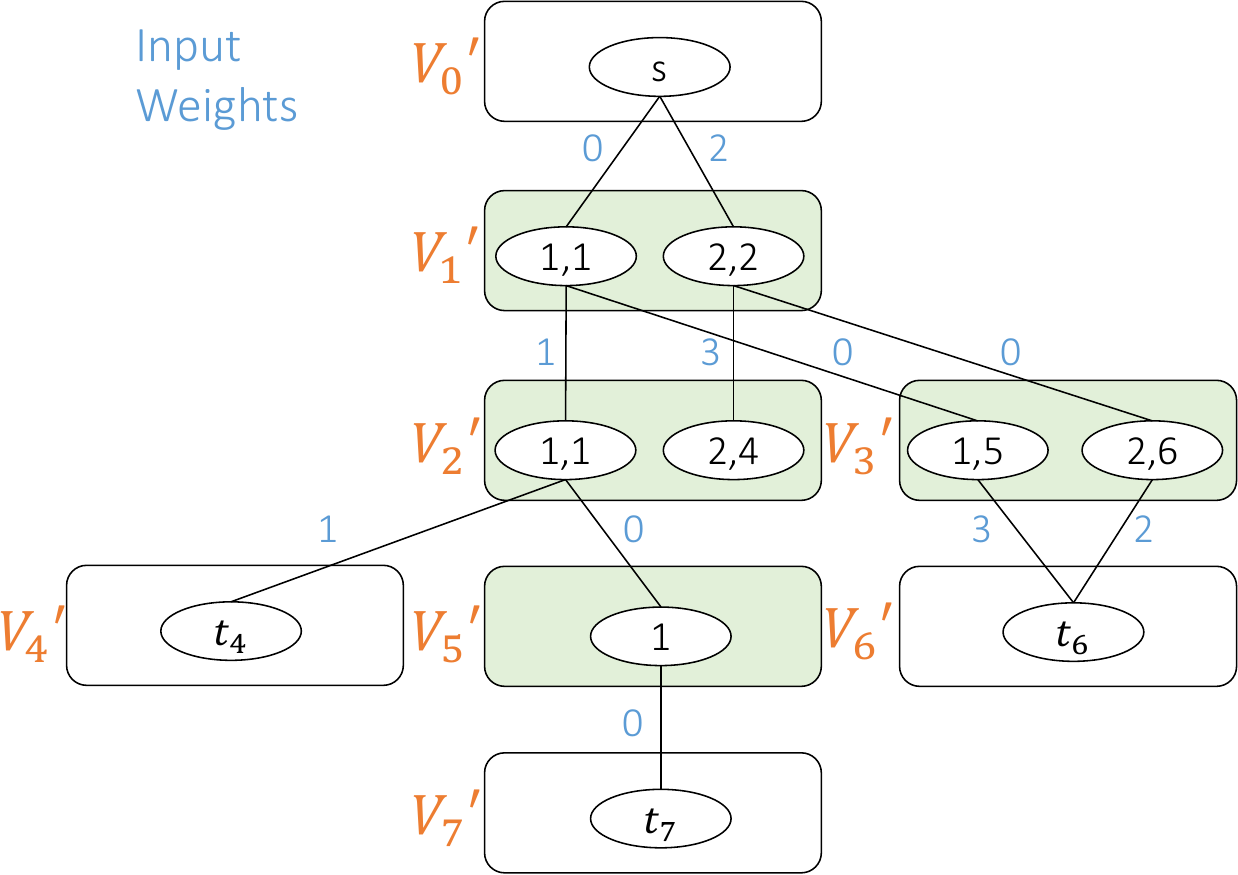}}
	\hfill
    \caption{\Cref{ex:free_connex}: ranked enumeration under \minweight semantics for the free-connex CQ
    $Q_{FC}(y_1, y_2, y_3, y_4) \datarule R_1(y_1, y_2), R_2(y_2, y_3), R_3(x_1, y_1, y_4), R_4(x_2, y_3)$ 
    on an example database.}
    \label{fig:fc}
\end{figure*}

\begin{example}[Free-connex CQ]
Consider the CQ 
$Q_{FC} \datarule R_1(y_1, y_2), R_2(y_2, y_3), R_3(x_1, y_1, y_4),$ $R_4(x_2, y_3)$.
We can verify that it is free-connex if we
add an additional hyperedge $R'$
that encompasses all the free variables
(\cref{fig:fc_query})
and then check that the hypergraph is acyclic
(e.g., by finding a join tree).
Using the algorithm of Brault-Baron~\cite{brault13thesis},
we can construct a join tree for $Q_{FC}$ such that a
connected subset of nodes $U$ contain all the free variables
and no existentially quantified ones (\cref{fig:fc_join_tree}).
In order to achieve that, we have to introduce two additional relations $R_3', R_4'$ 
which are projections of $R_3$ and $R_4$ respectively.
Given this join tree and a database instance (\cref{fig:fc_relations}), 
we can construct the T-DP instance $T$ shown in \cref{fig:fc_tdp1}.
The non-terminal stages correspond to
the nodes of the join tree and are populated by states 
(depicted by white circles)
that correspond to input tuples.
Ranked enumeration directly on $T$ would produce the answers to the full query.
Instead, we only run the bottom-up phase that computes the values $\solW_1(v)$ for all states $v$, 
shown on the top-right of each state.
We proceed by removing the stages that do not belong to $U$ and replacing them with terminal nodes, thereby getting a modified instance $T'$ shown in \cref{fig:fc_tdp2}.
Observe that ranked enumeration on $T'$ will now enumerate the answers to $Q_{FC}$.
To get the correct \minweight semantics, we also have to modify the input weights on $T'$.
Consider state $(1, 1) \in \Sset_2$ on $T$ that has two branches, one towards $\Sset_4$ and one towards $\Sset_5$.
For the first branch, we have to choose between $(1, 1) \in \Sset_4$ and $(2, 1) \in \Sset_4$.
The minimum achievable weight is achieved through $(1, 1)$ 
since $1 + 0 < 2 + 0$.
Therefore, when we remove stage $\Sset_4$ in $T'$ and replace it with a terminal stage 
$\Sset_4' = \{ t_4 \}$, we set the weight of the 
edge $((1,1), t_5)$ to be equal to 1 (i.e., the minimum).
The minimum achievable weights are computed from the bottom-up phase on $T$ (see \cref{eq:TDP_recursion}).
\label{ex:free_connex}
\end{example}

\begin{lemma}[Free-connex CQ to Full CQ]
\label{lem:fc_to_full}
Given a free-connex CQ $Q$ of size $\O(1)$ and a database $D$ of size $\O(n)$,
we can construct a full CQ $Q'$ of size $\O(1)$ and a database $D'$ of size $\O(n)$
such that ranked enumeration
with an \smonotone ranking function
under \minweight semantics for $Q$ 
produces the same answers as ranked enumeration for $Q'$. 
\end{lemma}
\begin{proof}
Let $\vec y$ be the free variables of $Q$.
In a join tree of $Q$, let $\varset(u)$ be the set of variables of the corresponding atom of
a node $u$.
Since the query is free-connex, we can compute 
a join tree with a connected subset of nodes $U$ that satisfy $\bigcup_{u \in U} \varset(u) = \vec y$ 
in $\O(|Q|)$ time using known techniques \cite{Berkholz20tutorial,brault13thesis}.
To achieve this, additional atoms might be introduced;
set the input weights of all tuples materialized from those atoms to $0$.
Also set the root of the tree to be any node $u \in U$.

Next, from the join tree, construct in a bottom-up fashion a T-DP instance $T$ as in \cref{sec:tdp}.
This takes $\O(n)$ time and removes all states that do not participate in any solution.
Every solution of $T$ 
is by construction an answer to the full query $Q_F$ that has the same body as $Q$ (without projections).
Given that (1) $U$ contains all the free variables $y$ needed for answering $Q$ and (2) bottom-up consistency with stages not in $U$\footnote{By slightly abusing the notation, we say that a stage belongs to $U$ if its corresponding node in the join tree belongs to $U$.} has already been enforced, 
the subtree induced by $U$ contains precisely the answers to $Q$.
In more detail, create a copy $T'$ of $T$ 
that only retains the stages
that belong to $U$.
Complete $T'$ with an artificial starting stage as the root of the tree
and terminal stages as the leaves, exactly as in \cref{sec:tdp}.
We argue that 
there is a 1-to-1 correspondence between the T-DP solutions 
of $T'$ and the answers to $Q$.
First, consider a T-DP solution $\sol$ of $T'$.
It has to contain states that belong to $\SsetR$ 
(recall that these are the ones not removed by the bottom-up pass), 
hence they can reach the terminal states of the original T-DP instance $T$.
Thus, there is a way to extend $\sol$ to a solution to the original state-space $T$, 
which corresponds to an answer to the full CQ $Q_F$.
The values assigned to the variables $\vec y$ constitute an answer to $Q$.
Conversely, an answer $q \in Q$ assigns values to the the $\vec y$ variables.
Since the subset $U$ of the join tree contains precisely the $\vec y$ variables, 
we can find tuples in the materialized relations of the join tree or equivalently, states in $T'$ that form a T-DP solution using those values.

To get \minweight semantics, we have to make adjustments to the input weights of $T'$.
In particular, we set the weights of the edges that reach the additional terminal nodes we introduced in
$T'$ according to the weights of $T$ that do not appear in $T'$.
Let $\Sset_c$ be a non-leaf stage in $T$ and $\Sset_p = \parent(\Sset_c)$ its parent such that $\Sset_c \notin U$ and $\Sset_p \in U$.
Also let $\Sset_p'$ be the copy of $\Sset_p$ in $T'$.
The construction of $T'$ added a stage $\Sset_j' = \{ t_j' \}$ and decisions 
$(v_p', t_j')$ for all $v_p' \in \Sset_p'$.
The weight of the edges to reach $t_j'$ from $v_p'$ is set to be the minimum achievable weight that $v_p$ could reach in $T$ from the branch that goes to $\Sset_c$:
$w(v_p', t_j') = \min_{(v_p, v_c) \in \Dec_{pr}}
\big\{ 
\weight(v_p, v_c) \aggr \solW_1(v_c)
\big\}$.
This can be done in time linear in $T$.
For a solution $\sol$ of $T$ and a solution $\sol'$ of $T'$, let $\sol' \subset_U \sol$ if they agree on the subset $U$.
By subset-monotonicity, we have that $w(\sol') = \min_{\sol : \sol' \subset_U \sol} w(\sol)$ for every solution $\sol'$ of $T'$.

The total time spent so far is linear in the size of the database $D$.
From $T'$ we can easily construct a corresponding CQ $Q'$ and database $D'$ by creating a relation for every stage and a tuple for every state.
The weight of a tuples is equal to the weight of the (unique) edge that connects it to its parent
plus the weight of the edges that connect it to terminal children stages (i.e those that are leaves), is such exist.
\end{proof}

\Cref{lem:fc_to_full} says that for free-connex CQs and ranking functions that are \smonotone,
we can achieve ranked enumeration
by first creating a full CQ and a modified database and then applying our any-$k$ algorithms.
By following this approach, we obtain the guarantees of \Cref{theorem:cq_data_comp}.
The corresponding lower bound of the theorem is immediate from \Cref{theorem:known-enumeration} since ranked enumeration is strictly harder than unranked enumeration.

\subsection{Cyclic Queries}
\label{sec:cycles}

Our work mainly targets acyclic CQs, yet the techniques we develop can also be used for cyclic CQs
by leveraging (hyper)tree decompositions~\cite{GottlobGLS:2016}.
The main idea of such \emph{decomposition methods} is to reduce cyclic CQs to acyclic CQs.
Extending the notion of
tree decompositions for graphs~\cite{RobertsonS:1986}, (hyper)tree
decompositions~\cite{GottlobLS:2002} organize the
relations into ``bags'' and arrange those bags into a
tree structure. Each decomposition is associated with
a width parameter that captures the degree of acyclicity in the query
and affects the complexity of subsequent evaluation: smaller width implies
lower time complexity.
\emph{Our approach is orthogonal to the decomposition method used and it adds
ranked enumeration capability virtually ``for free.''}

\introparagraph{Decomposition methods}
The state-of-the-art decompositions rely on the submodular width $\subw(Q)$
of a query $Q$~\cite{Marx:2013:THP:2555516.2535926}. 
There exist decomposition methods that run in time 
$\O(f(|Q|)n^{(2+\delta)\subw(Q)})$ for $\delta > 0$~\cite{berkholz19submodular}
or $\bigO(f_1(|Q|)n^{\subw(Q)}(\log n)^{f_2(|Q|)})$~\cite{khamis17panda}
for query-dependent functions $f, f_1$ and $f_2$.
For example, $\stages$-cycle CQs $Q_{C\stages}$
can be decomposed into trees where each bag materializes a relation of size $\O(n^{2-1/\lceil \ell/2 \rceil})$~\cite{Alon1997,DBLP:conf/sebd/Scarcello18}.
Since this is an active research area,
we expect these algorithms to be improved and we believe our framework is general
enough to accommodate future decomposition methods. Sufficient conditions for our approach
to apply with no additional preprocessing cost are
(1) the full output of $Q$ is the union of the output produced by the trees
in the decomposition and
(2) the number of trees depends only on query size $|Q|$.
Both are satisfied by current decompositions.\footnote{In general, a decomposition
might be designed only for Boolean CQs where it can take ``shortcuts'' because no output tuples are needed, e.g., through
(fast) matrix-multiplication~\cite{Alon1997}.}

\introparagraph{Applying any-$k$}
To use our any-$k$ algorithms on top of an existing decomposition, three issues have to be addressed.
First, the decomposition (e.g., based on the submodular width) might create more than one tree.
To handle that, we define UT-DP as a \emph{union of T-DP problems} 
where a solution to
any of the T-DP problems constitutes a valid solution.
Thus, we are given a set of $u$ functions
$F = \big\{ f^{(i)} \big\}$, 
each defined over a solution space $\sol^{(i)}, i \in [u]$.
The UT-DP problem is then to \emph{find the minimum solution across all T-DP instances}.
The necessary changes to any of our any-$k$ algorithms are straightforward:
We add one more top-level data structure $\Union$ that maintains the
last returned solution of each separate T-DP algorithm in a single priority queue.
Whenever a solution is popped from $\Union$, it gets replaced by the
next best solution of the corresponding T-DP problem.
The second challenge is that we have to properly compute the tuple weights in the bags of the decomposition.
For this, we track the lineage of the relations participating in a bag at the schema level: We only need to know from which input
relation a tuple originates and if that relation's weight values had already been accounted
for by another bag.
Third, we have to deal with possible output duplicates when a decomposition creates
multiple trees. 
We eliminate the duplicates by maintaining a lookup table that contains all the answers produced so far.
Note that many consecutive duplicates can increase the delay of the algorithm,
however $\TT(k)$ is unaffected in data complexity because the number of duplicates for an answer is bounded by $\O(|Q|)$.

\section{Ranking Functions}
\label{sec:ranking_functions}

This section focuses on different aspects related to the ranking functions that are supported by our framework.
We explore how algebraic structures can be used to define ranking functions (\Cref{sec:algebra}),
how the monotonicity properties we have discussed relate to algebraic properties,
as well as other definitions of monotonicity in the literature (\Cref{sec:monotonicity}),
and show in more detail the correctness of our algorithms when we generalize from
the sum-of-weights model we have focused on so far
to other ranking functions (\cref{sec:rankingfctgeneralizing}).

\subsection{Relationship to Algebraic Structures}
\label{sec:algebra}

In many cases, the aggregate ranking function $w_A$ is defined through a binary operator such as
$+, \times$ or $\max$.
Algebraic structures~\cite{GondranMinoux:2008:Semirings} capture the properties of such operators
in an abstract framework.
As a prominent example, semirings have been repeatedly shown to be at the core of many efficient algorithms~\cite{abo16faq,aji00distributive,green07semirings,}.
We now show how these are related to ranking functions and different monotonicity properties.

\introparagraph{Definitions}
A \emph{monoid} is a 3-tuple $(W, \oplus, \0)$
where $W$ is a non-empty set and
$\oplus: W \times W \to W$ is a closed binary operation such that:
\begin{enumerate}
    \item $(x \oplus y) \oplus z = x \oplus (y \oplus z)$ (associativity),
    \item $\0 \in W$ satisfies $x \oplus \0 = \0 \oplus x = x, \forall x \in W$ (neutral or identity element).
\end{enumerate}
A monoid is commutative if it also satisfies (3) $x \oplus y = y \oplus x, \forall x, y \in W$ (commutativity).

A \emph{semiring} is a 5-tuple $(W, \oplus, \otimes, \0, \1)$, where
\begin{enumerate}
    \item $(W, \oplus, \0)$ is a commutative monoid and $(W, \otimes, \1)$ is a monoid,
    \item $\forall x, y, z \in W: 
            (x \oplus y) \otimes z =
            (x \otimes z) \oplus  (y \otimes z)$
            (distributivity of $\otimes$ over $\oplus$),
    \item $\forall x \in W: x \otimes \0 = \0$
            ($\0$ is absorbing or annihilating for $\otimes$).
\end{enumerate}

A \emph{commutative semiring} is one where $(W, \otimes, \1)$ is also commutative.
A binary operator $\oplus$ is called \emph{selective} iff it always returns one of the two operands, i.e.,
$\forall x,y \in W: (x \oplus y = x) \lor (x \oplus y = y)$.
A \emph{selective (commutative) dioid} is a (commutative) semiring in which
$\oplus$ is selective.
A selective commutative dioid has a \emph{lattice} structure if  $\forall x,y \in W: x \oplus y = x \Rightarrow x \otimes y = y$, i.e., $\otimes$ is also selective and follows the reverse order.
An element $a$ of a commutative monoid $(W, \otimes, \1)$ is called \emph{cancellative} iff it can be ``canceled out''
to solve equations, i.e., $\forall x,y \in W: x \otimes a = y \otimes a \Rightarrow x = y$.
Notice that cancellation is a weaker notion than invertibility.
A semiring (or a selective dioid) is cancellative if all the elements 
of the 
product monoid $(W, \otimes, \1)$ are cancellative,
except for the neutral element of the sum monoid $\0$ 
(which is absorbing and thus cannot be cancellative).

For ranked enumeration, a total order on the domain is necessary and it is known that the selective property is enough to induce such an order. 

\begin{lemma}[\cite{GondranMinoux:2008:Semirings}]
\label{lem:total_order}
For any selective dioid $\mathcal{D} = (W, \oplus, \otimes, \0, \1)$, the relation $a \preceq_\mathcal{D} b \equiv a \oplus b = a$ is a total order.
Also, this order is
monotonic (or translation-invariant)
with respect to $\otimes$
(i.e., it is a total order that satisfies $a \preceq_\mathcal{D} b \Rightarrow a \otimes c \preceq_\mathcal{D} b \otimes c, \forall a,b,c \in W$).
\end{lemma}
\begin{proof}
First, we prove that $\preceq_\mathcal{D}$ is a total order.
(1)~To show it is transitive, suppose that $a \preceq_\mathcal{D} b$ and $b \preceq_\mathcal{D} c$ for $a,b,c \in W$, thus
$a \oplus b = a$ and $b \oplus c = b$.
By adding $a$ to both sides of the second equation we get $a \oplus (b \oplus c) = a \oplus b$
and by associativity, $(a \oplus b) \oplus c = a \oplus b$, hence $a \oplus c = a$ or $a \preceq_\mathcal{D} c$.
(2)~For antisymmetry, suppose that $a \preceq_\mathcal{D} b$ and $b \preceq_\mathcal{D} a$ for $a,b \in W$, thus
$a \oplus b = a$ and $b \oplus a = b$.
By commutativity, $a \preceq_\mathcal{D} b = b \preceq_\mathcal{D} a$, so it follows that $a=b$.
(3)~The fact that it is total (or connected), i.e.\ $a \preceq_\mathcal{D} b$ or $b \preceq_\mathcal{D} a$ for all $a,b \in W$, 
follows directly from selectivity.
(4)~Similarly, the fact that it is reflexive follows immediately from $a \oplus a = a$, hence $a \preceq_\mathcal{D} a$.

Next, we prove that the order is monotonic w.r.t. $\otimes$.
For all $a,b,c \in W$ we have $a \preceq_\mathcal{D} b \Rightarrow a \oplus b = a \Rightarrow 
c \otimes (a \oplus b) = c \otimes a
$
and by distributivity, $(c \otimes a) \oplus (c \otimes b) = c \otimes a$,
which implies that $a \otimes c \preceq_\mathcal{D} b \otimes c$.
\end{proof}

The converse is also true~\cite{bistarelli99semiring}: given a commutative monoid equipped with a total order that is monotonic with respect to $\otimes$,
there is a corresponding selective dioid where the $\oplus$ operator selects the least element according to the order.
Therefore, these two perspectives are equivalent.

\introparagraph{Algebraic structures as ranking functions}
A selective commutative dioid can be used as a ranking function where
$\oplus$ acts as a comparator and $\otimes$ as an aggregator of input weights.
The sum-of-weights model corresponds to $(\mathbb{R} \cup \{ \infty \}, \min, +, \infty, 0)$,
also called the \emph{tropical semiring}~\cite{pin98tropical}.
Notice the correspondence of
semiring multiplication
$\otimes$ to $+$ and
semiring addition
$\oplus$ to $\min$.
Another semiring we can use is $(\R \cup \{ \infty, -\infty \}, \min, \max, \infty, -\infty)$,
where we
essentially swap the $+$ operator for $\max$.
Ranked enumeration with this semiring prioritizes query answers whose maximum-weight witness
is as small as possible (regardless of the sum of weights).

\begin{definition}
Given a commutative selective dioid $\mathcal{D} = (W, \oplus, \otimes, \0, \1)$, we can define an algebraic~\footnote{To be precise, the aggregate ranking function is distributive~\cite{DBLP:journals/datamine/GrayCBLRVPP97} or self-decomposable~\cite{jesus15aggregation} in this case.} ranking function
$w^\mathcal{D}$
such that $w_A^\mathcal{D} = \bigotimes_{x \in X} x$ and $x \preceq y$ iff $x \oplus y = x$.
\end{definition}

\introparagraph{On commutativity}
Commutativity of the product operator $\otimes$ together with associativity
guarantees that the outcome of an aggregation does not depend on the order of the elements. 
It is required so that the ranking function defined from a selective dioid is an \emph{aggregate} function,
which is order-insensitive (\Cref{sec:ranked_enumeration}).
The reason that this property is important is because different join trees impose a different order of stages in DP,
thus a different order of weight aggregation.
Without commutativity, the weight of a query answer could be dependent on the chosen join tree.\footnote{In the earlier version of this article~\cite{tziavelis20vldb},
commutativity was not stated as an explicit requirement. This oversight is not important if the DP structure is fixed and the ranking function order agrees with it,
but it is necessary if we want the freedom to choose different join trees.
}

\begin{example}[Non-commutative structures]
The prime example of a non-commutative operator is matrix multiplication since for two matrices $A,B$ we could have
$AB \neq BA$.
However, we are not aware of any non-commutative selective dioid on infinite domains.
For example, if we assume that $\otimes$ is matrix multiplication,
there is no obvious way to define a selective operator on matrices so that
distributivity also holds.
Using the Mace4 tool~\cite{mccune03mace4}, 
we were able to identify a finite such structure on a domain of 5 elements
(see \Cref{appendix:dioid_example}).
\end{example}

\subsection{Comparison of Monotonicity Properties}
\label{sec:monotonicity}

In this section, we compare different notions of monotonicity of the ranking function and show how these can be derived from the properties of algebraic structures introduced previously.

\introparagraph{Holistic-Monotonicity}
A well-known notion of monotonicity is due to Fagin et al. \cite{fagin03}.
To avoid ambiguity in the use of the generic term ``monotonicity'', we use the term ``holistic-monotonicity''.

\begin{definition}[Holistic-Monotonicity~\cite{fagin03}]
A ranking function $w$ is holistic-monotone if
$x_i \preceq x_i', \forall i \in [\ell]$ implies
$w_A(\{x_1, \ldots, x_\ell \}) \preceq w_A(\{ x_1', \ldots, x_\ell' \})$.
\end{definition}

Intuitively, a holistic-monotone ranking function will yield a better (or equal) result
when all of its components are better (or equal).
We note that there is an \emph{equivalent} definition of this property where only one component changes value:
a ranking function $w$ is holistic-monotone if 
$x_i \preceq x_i'$ implies
$w_A(\{x_1, \ldots, x_i, \ldots, x_\ell \}) \preceq w_A(\{ x_1, \ldots, x_i', \ldots, x_\ell \})$.

\begin{figure}[tb]
\centering
\includegraphics[height=4.5cm]{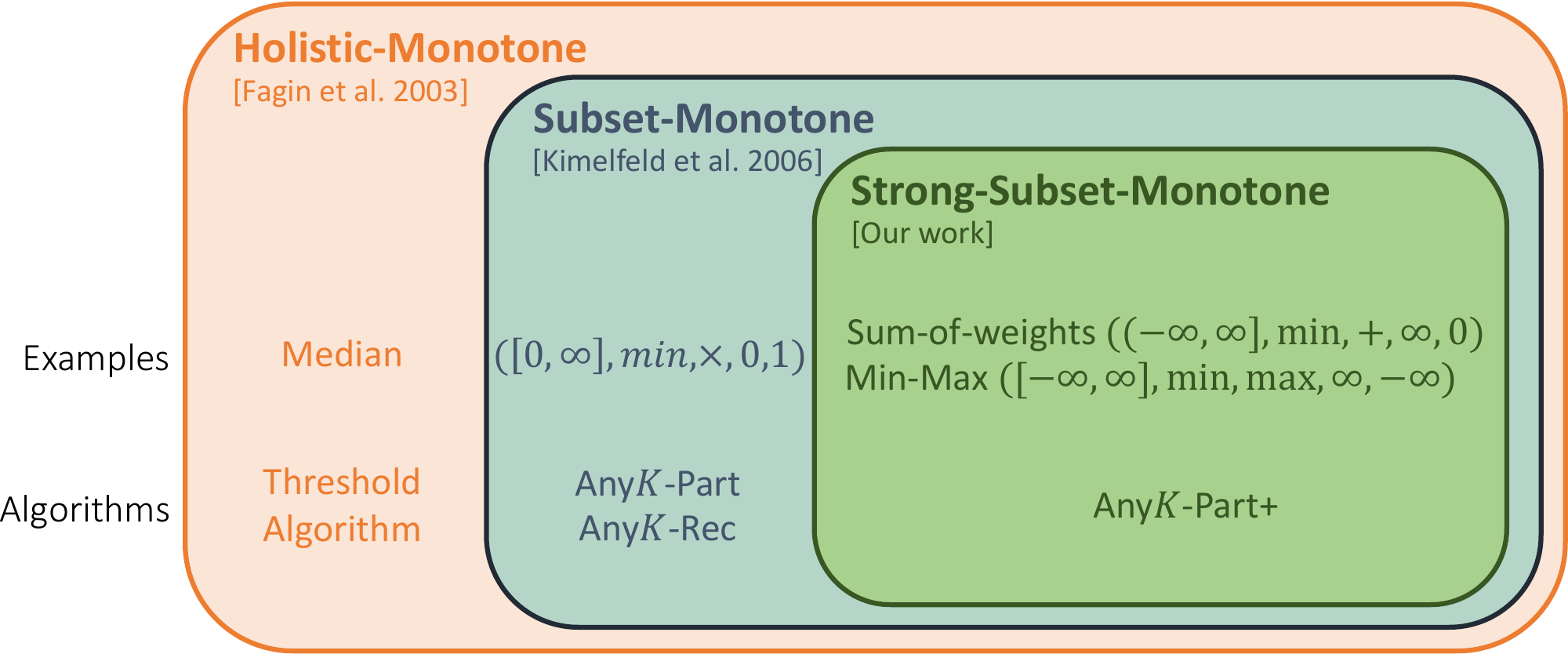}
\caption{Hierarchy of monotonicity properties. Depicted size does not correspond to actual class size.}
\label{fig:monotonicity}
\end{figure}

\introparagraph{Subset-Monotonicity}
We remind the reader that \smonotonicity allows us to compare two (sub)sets of elements ($w(Y_1) \preceq w(Y_2)$) to infer
their order when they are extended with the same (sub)set $X$ (\Cref{def:smonotone}). 
Compared to holistic-monotonicity, \smonotonicity is a stricter property that fewer ranking functions satisfy.
To see this, notice that holistic-monotonicity is implied if $|Y_1|=|Y_2|=1$.
Intuitively, holistic-monotonicity does not allow us to rely on the ordering of larger (than size-one) sets of elements
when we assemble larger sets.
Indeed, the Threshold Algorithm~\cite{fagin03} always makes comparisons either between individual elements
or between complete solutions. 
On the other hand, DP and our any-$k$ algorithms require the ability to meaningfully compare partial solutions (prefixes or suffixes).

\begin{example}[Not subset-monotone]
Consider as a ranking function the median over the real numbers.
It is holistic-monotone and
therefore supported by the Threshold Algorithm \cite{fagin03},
because substituting any individual element with a smaller one can only decrease the median.
However, it is \emph{not \smonotone}.
Consider $L_1 = \{1, 10, 100\}$ and $L_2 = \{19, 20, 21\}$.
Then $\textrm{median}(L_1) = 10 \leq 20 = \textrm{median}(L_2)$,
but if we union with $L = \{101, 102\}$, we have
$\textrm{median}(L_1 \cup L) = 100 \geq 21 = \textrm{median}(L_2 \cup L)$.
\end{example}

If the ranking function is defined via a commutative selective dioid, then \smonotonicity is guaranteed as it follows from the other algebraic properties.

\begin{lemma}
For a commutative selective dioid $\mathcal{D}$, the ranking function $w^\mathcal{D}$ is \smonotone.
\end{lemma}
\begin{proof}
For all multisets $Y_1, Y_2, X$, 
$w_A^\mathcal{D}(Y_1) \preceq w_A^\mathcal{D}(Y_2) \Rightarrow 
\bigotimes_{y_1 \in Y_1} y_1 \preceq_\mathcal{D} \bigotimes_{y_2 \in Y_2} y_2 \Rightarrow
(\bigotimes_{x \in X} x) \otimes (\bigotimes_{y_1 \in Y_1} y_1) \preceq_\mathcal{D} (\bigotimes_{x \in X} x) \otimes (\bigotimes_{y_2 \in Y_2} y_2) \Rightarrow
w_A^\mathcal{D}(X \uplus Y_1) \preceq w_A^\mathcal{D}(X \uplus Y_2)$,
where $\preceq_\mathcal{D}$ is the total order $a \preceq_\mathcal{D} b \equiv a \oplus b$
and the second-to-last implication follows from the monotonic property of that order (\Cref{lem:total_order}).
\end{proof}

\introparagraph{Strong-Subset-Monotonicity}
The \ssmonotonicity property (\cref{def:strongsubset}) is even stricter.
It enables us to swap a set $X_1$ in a comparison $w_A(X_1 \uplus Y_1) \preceq w_A(X_1 \uplus Y_2)$ with a different set $X_2$, provided that $w_A(X_1) \preceq w_A(X_2)$.
As previously mentioned, \smonotonicity is implied by setting $X_1 = \emptyset$.
The benefit of \ssmonotonicity is that it allows us to reuse the relative ranking of $Y_1$ and $Y_2$
even if we do not directly compare them, but instead compare their extension with $X_1$. 

We now give some intuition on how \ssmonotonicity differs from the weaker \smonotonicity.
If we could show that $w_A(Y_1) \preceq w_A(Y_2)$, then \ssmonotonicity would be guaranteed by \smonotonicity.
As we will show shortly, this is the case for cancellative algebraic structures where we can cancel out $X_1$. 
For example, in the tropical semiring, if $1 + 2 \leq 1 + 3$, then we can swap out $1$ with any other element, e.g., $5 + 2 \leq 5 + 3$ because we can infer $2 \leq 3$.
But \ssmonotonicity may still hold even if $w_A(Y_1) \npreceq w_A(Y_2)$.
An example is the min-max semiring where $\max(5, 3) \leq \max(5, 2)$ and $\max(i, 3) \leq \max(i, 2)$ for any $i \geq 5$, even though $3 \nleq 2$.
The condition $i \geq 5$ in the example is captured by the requirement that $w_A(X_1) \preceq w_A(X_2)$.
In terms of our \ANYKPARTP algorithm, this continues to hold as long as we continue to visit lower-ranked partial solutions.
We prove in more detail how this property guarantees the correctness of the algorithm in
\Cref{sec:rankingfctgeneralizing}.

\begin{example}[Not Strong-Subset-Monotone]
\label{ex:weaksubsetmonotone}
An example of a ranking function that is \smonotone but not \ssmonotone is the one defined by $([0,\infty], \min, \times, \infty, 1)$.
For example, we have $0 \times 2 \leq 0 \times 1$ but $5 \times 2 > 5 \times 1$ even though
$0 < 5$.
However, if we remove the non-cancellative element $0$, then the structure $((0,\infty], \min, \times, \infty, 1)$ gives us an \ssmonotone ranking function.
\end{example}	

We now show two algebraic properties sufficient to guarantee \ssmonotonicity:

\begin{lemma}
Given a commutative selective dioid $\mathcal{D}$ that is either cancellative 
or has a lattice structure,
the corresponding ranking function $w^\mathcal{D}$ is \ssmonotone.
\end{lemma}
\begin{proof}
Strong-subset-monotonicity requires that for multisets $X_1, X_2, Y_1, Y_2$,
if $w_A^\mathcal{D}(X_1 \uplus Y_1) \preceq w_A^\mathcal{D}(X_1 \uplus Y_2)$ 
and $w_A^\mathcal{D}(X_1) \preceq w_A^\mathcal{D}(X_2)$,
then $w_A^\mathcal{D}(X_2 \uplus Y_1) \preceq w_A^\mathcal{D}(X_2 \uplus Y_2)$.
By the definition of $w_A^\mathcal{D}$, this is equivalent to showing that if
$(x_1 \otimes y_1) \oplus (x_1 \otimes y_2) = x_1 \otimes y_1$
and $x_1 \oplus x_2 = x_1$,
then $(x_2 \otimes y_1) \oplus (x_2 \otimes y_2) = x_2 \otimes y_1$ for $x_1, x_2, y_1, y_2 \in W$.

First, let $\mathcal{D}$ be cancellative.
We apply distributivity to get $x_1 \otimes (y_1 \oplus y_2) = x_1 \otimes y_1$.
If $x_1 \neq \0$, we cancel it out to get $y_1 \oplus y_2 = y_1$ and by \Cref{lem:total_order},
$(x_2 \otimes y_1) \oplus (x_2 \otimes y_2) = x_2 \otimes y_1$.
If $x_1 = \0$, then because $x_1 \oplus x_2 = x_1$ we have $x_2 = \0$
and $(x_2 \otimes y_1) \oplus (x_2 \otimes y_2) = x_2 \otimes y_1$ because $x_2$ is absorptive.

Second, let $\mathcal{D}$ have a lattice structure.
Then, we multiply $(x_1 \otimes y_1) \oplus (x_1 \otimes y_2) = x_1 \otimes y_1$ by $x_2$ to get
$x_2 \otimes((x_1 \otimes y_1) \oplus (x_1 \otimes y_2)) = x_2 \otimes x_1 \otimes y_1$
and by distributivity,
$(x_2 \otimes x_1 \otimes y_1) \oplus (x_2 \otimes x_1 \otimes y_2) = x_2 \otimes x_1 \otimes y_1$.
Then, because $x_1 \oplus x_2 = x_1$, we have $x_1 \otimes x_2 = x_2$,
hence $(x_2 \otimes y_1) \oplus (x_2 \otimes y_2) = x_2 \otimes y_1$.
\end{proof}

\Cref{fig:monotonicity} illustrates the hierarchy between the three monotonicity notions.

\subsection{Correctness of Any-$k$ Algorithms}
\label{sec:rankingfctgeneralizing}

In this section, we elaborate on how our algorithms can work with ranking functions that are more general than
the sum-of-weights model that we have mostly assumed throughout the paper,
as long as these satisfy the monotonicity properties we have discussed.

\introparagraph{Correctness by monotonicity}
For \ANYKPART and \ANYKREC, the key property that is required for correctness
is that the order of suffixes from a node is maintained when those suffixes are grown to longer paths by appending
the same prefix in front.
This is required in (1) Dynamic Programming (\Cref{eq:DP_recursion,eq:TDP_recursion}),
(2) deviations for ranking (\Cref{lem:deviations}),
(3) the successor function for \ANYKPART (\Cref{sec:part}), and
(4) the generalized principle of optimality for \ANYKREC (\Cref{lem:generalizedopt}).
Fortunately, this property is an immediate consequence of subset-monotonicity.

\begin{lemma}
\label{lem:suffixes}
For an \smonotone ranking function $w$,
if $w(\sol_i(v)) \preceq w(\sol_j(v))$ for two suffixes $\sol_i(v), \sol_j(v)$
starting at node $v$,
then
$w(p \concat \sol_i(v)) \preceq w(p \concat \sol_j(v))$ for any path $p$ ending at $v$.
\end{lemma}

For \ANYKPARTP, we additionally need the property that from the order of complete solutions
popped from the priority queue $\Cand$,
we can obtain a ranking of suffixes that we reuse for later solutions.
This is not the same as \Cref{lem:suffixes} because we do not directly have access to the actual ranking of suffixes.
Here, we need the stricter \ssmonotonicity property.

\begin{lemma}
\label{lem:partp_correctness}
For an \ssmonotone ranking function $w$,
if $w(p_1 \concat \sol_i(v)) \preceq w(p_1 \concat \sol_j(v))$ for two suffixes $\sol_i(v), \sol_j(v)$
starting at node $v$ and the leading 
prefix $p_1$ of $v$,
then
for any path $p_2 \concat \sol_i(v)$ where $p_2$ is a prefix ending at $v$ with $p_2 \neq p_1$,
we have that
$w(p_2 \concat \sol_i(v)) \preceq w(p_2 \concat \sol_j(v))$.
\end{lemma}
\begin{proof}
For brevity, we write any $\Pi_k(v)$ as $\Pi_k$.
Since $p_1$ is the leading prefix, we have that $w(p_1 \concat \Pi_1)$ is smaller or equal
than all other $s-t$ paths going though $v$ for some suffix $\Pi_1$.
We will also be using a derivative property of strong subset-monotonicity:
if $w(Y_1 \uplus X_1) = w(Y_1 \uplus X_2)$ and $w(X_1) \preceq w(X_2)$, then 
$w(Y_2 \uplus X_1) = w(Y_2 \uplus X_2)$.
This is proven by rewriting the equality as the conjunction of two inequalities
$\preceq$ and $\succeq$ and applying \ssmonotonicity twice.
We consider four distinct cases.

\underline{Case 1:} $w(\Pi_i) \preceq w(\Pi_j)$. By \smonotonicity, 
$w(p_2 \concat \Pi_i) \preceq w(p_2 \concat \Pi_j)$.

\underline{Case 2:} $w(p_1) \preceq w(p_2)$. Starting from $w(p_1 \concat \Pi_i) \preceq w(p_1 \concat \Pi_j)$ 
and using \ssmonotonicity to replace $p_1$ with $p_2$,
we obtain $w(p_2 \concat \Pi_i) \preceq w(p_2 \concat \Pi_j)$.

\underline{Case 3:} $w(\Pi_j) \preceq w(\Pi_i)$, $w(p_2) \preceq w(p_1)$, and $w(\Pi_j) \preceq w(\Pi_1)$. 
We apply \smonotonicity to $w(p_2) \preceq w(p_1)$ to get $w(p_2 \concat \Pi_1) \preceq w(p_1 \concat \Pi_1)$.
Since $p_1$ is leading, we also have $w(p_1 \concat \Pi_1) \preceq w(p_2 \concat \Pi_1)$,
therefore $w(p_1 \concat \Pi_1) = w(p_2 \concat \Pi_1)$.
Now applying the equality version of \ssmonotonicity to replace $\Pi_1$ with $\Pi_j$ (recall that $w(\Pi_1) \preceq w(\Pi_j)$), we obtain
$w(p_1 \concat \Pi_j) = w(p_2 \concat \Pi_j)$.
Repeating the same process for $\Pi_i$ instead of $\Pi_j$ (by transitivity, $w(\Pi_1) \preceq w(\Pi_i)$), we obtain $w(p_1 \concat \Pi_i) = w(p_2 \concat \Pi_i)$.
We now use our initial assumption that $w(p_1 \concat \Pi_i) \preceq w(p_1 \concat \Pi_j)$
and prove $w(p_2 \concat \Pi_i) = w(p_1 \concat \Pi_i) \preceq w(p_1 \concat \Pi_j) = w(p_2 \concat \Pi_j)$.

\underline{Case 4:} $w(\Pi_j) \preceq w(\Pi_i)$, $w(p_2) \preceq w(p_1)$, and $w(\Pi_j) \preceq w(\Pi_1)$. 
We apply \smonotonicity to $w(\Pi_j) \preceq w(\Pi_1)$ to get $w(p_1 \concat \Pi_j) \preceq w(p_1 \concat \Pi_1)$.
Since $p_1$ is leading, we also have $w(p_1 \concat \Pi_1) \preceq w(p_1 \concat \Pi_j)$,
therefore $w(p_1 \concat \Pi_j) = w(p_1 \concat \Pi_1)$.
Similarly, we apply \smonotonicity to $w(p_2) \preceq w(p_1)$ to get $w(p_2 \concat \Pi_j) \preceq w(p_1 \concat \Pi_j)$.
Since $p_1$ is leading, we also have $w(p_1 \concat \Pi_1) \preceq w(p_2 \concat \Pi_j)$.
Up to this point, we have $w(p_2 \concat \Pi_j) \preceq w(p_1 \concat \Pi_j) = w(p_1 \concat \Pi_1) \preceq w(p_2 \concat \Pi_j)$.
Therefore, $w(p_2 \concat \Pi_j) = w(p_1 \concat \Pi_j)$.
We now apply the equality version of \ssmonotonicity to replace $\Pi_j$ with $\Pi_i$ (recall that $w(\Pi_j) \preceq w(\Pi_i)$) to obtain
$w(p_2 \concat \Pi_i) = w(p_1 \concat \Pi_i)$.
Similarly to Case 3, we derive 
$w(p_2 \concat \Pi_i) = w(p_1 \concat \Pi_i) \preceq w(p_1 \concat \Pi_j) = w(p_2 \concat \Pi_j)$.

\end{proof}

\Cref{lem:partp_correctness} shows that the order of suffixes discovered by the leading prefix $p_1$
can be reused for other prefixes.
Since the ranking function is commutative, it does not distinguish between prefixes and suffixes
and the dual statement also holds:
if $w(p_1 \concat \Pi_i(v)) \preceq w(p_2 \concat \Pi_i(v))$ 
for $p_2 \neq p_1$, then
$w(p_1 \concat \Pi_j) \preceq w(p_2 \concat \Pi_j)$.
This shows that a subscriber will never be needed before the leading prefix discovers the next-best suffix,
establishing the correctness of \ANYKPARTP.

\introparagraph{Any-$k$ without an inverse}
\ANYKPART and \ANYKPARTP 
compute the weights of $\O(\ell)$ new
candidates that are deviations of the current solution in $\O(\ell)$ time in every iteration. 
A careful reader might have noticed that while this is straightforward for sum-of-weights using subtraction,
other ranking functions might not have such an inverse operation.
In more detail, recall from 
\Cref{sec:complexity_measures}
that our analysis assumes an algebraic ranking function that allows us to aggregate 
partial weights $w(X)$ and $w(Y)$ into $w(X \uplus Y) = w(X) \otimes w(Y)$ with some operation $\otimes$ in $\O(1)$.
An inverse, denoted here by $\oslash$ allows us to reverse such an operation, i.e., 
$w(X) = w(X \uplus Y) \oslash w(Y)$.
Commutative monoids with this property are called \emph{Abelian groups}.
Typical examples of monoids that are not groups are the logical conjunction
$(\{0, 1\}, \wedge, 1)$
and the minimum over reals
$(\R, \min, \infty)$. 
In general, an inverse
allows us to
speed up the computation by reusing prior results.

In the context of our algorithms, \ANYKPART and \ANYKPARTP can use the inverse operation to calculate from a path
$\langle s, v_1, \ldots, v_j \rangle \concat \sol_1(v_j)$,
the weight of a deviation
$\langle s, v_1, \ldots, v_j' \rangle \concat \sol_1(v_j')$ as
$w(\langle s, v_1, \ldots, v_j \rangle \concat \sol_1(v_j)) \oslash w(v_{j-1}, v_j) \oslash w(\sol_1(v_j))
\otimes w(v_{j-1}, v_{j}') \otimes w(\sol_1(v_j'))$,
i.e., taking the current weight, ``subtracting'' the suffix weight and ``adding'' the new suffix weight.
We now discuss how to achieve the same $\O(\ell)$ computation (for $\O(\ell)$ new candidates)
per iteration without an inverse.
Note that \ANYKREC never uses an inverse since it always constructs solutions by appending one node to a suffix or a list of subtrees
(see \Cref{rec_line:insert} of \Cref{alg:rec}).

For DP, the modification that we have to make is that, as we expand a popped solution, 
we keep track of the weight of every prefix.
Using the weight of the prefix, we can compute the weight of a deviation
$\langle s, v_1, \ldots, v_j' \rangle \concat \sol_1(v_j')$
in $\O(1)$ as
$w(\langle s, v_1, \ldots, v_{j-1} \rangle)
\otimes w(v_{j-1}, v_{j}') \otimes w(\sol_1(v_j'))$.
For T-DP, the situation is more involved because a suffix of a solution consists of many disconnected subtrees.
As a consequence, the new node $v_j'$ only has access to the optimal weight of its subtree
and not the optimal weight of the suffix.
For the example of \Cref{fig:tdp}, suppose that our current solution is $\langle s, v_1, v_2, v_3 \rangle$
and we generate the deviation $\langle s, v_1, v_2, v_3' \rangle$.
Then, $\sol_1(v_3')$ does not contain the weight of the $\Sset_2-\Sset_4$ transition because it is in a different subtree.
Naively, computing the weight of the optimal suffix for all deviations of an iteration will take $\O(\ell^2)$.
We now describe how to bring this down to $\O(\ell^2)$.
First, recall that the stages are indexed in BFS order.
For each node $v_i$ of the current solution, let $\textrm{first}(v_i)$ be the smallest index of a node at
the same tree level and let $\textrm{last}(v_i)$ be the largest.
We now define the ``sums'' $\text{prefS}(v_i) = \bigotimes_{k \in [\textrm{first}(v_i), i]} \solW_1(v_k)$
and $\text{suffS}(v_i) = \bigotimes_{k \in [i, \textrm{last}(v_i)]} \solW_1(v_k)$.
These can be computed for every node in the current solution in $\O(\ell)$.
Using these data structures, the weight of a deviation
$\langle s, v_1, \ldots, v_j' \rangle \concat \sol_1(v_j')$ is
$w(\langle s, v_1, \ldots, v_{j-1} \rangle) \otimes \text{prefS}(v_{c-1}) \otimes \solW_1(v_j') \otimes \text{suffS}(v_{j+1})$, where
$c$ is the smallest index of a child of $v_j$
assuming that $c-1$ is at the same level as $c$ and $j+1$ is at the same level as $j$.
For the example of $\langle s, v_1, v_2, v_3 \rangle$ in \Cref{fig:tdp} we have $\langle s, v_1, v_2, v_3' \rangle = 
w(\langle s, v_1, v_2 \rangle) \otimes 
\text{prefS}(v_{4}) \otimes
\solW_1(v_3')$
because $3+1$ is not at the same level as $3$.

\subsection{Capturing Other Ranking Functions}

\subsubsection{Lexicographic orders}
\label{sec:lex}

As discussed in \Cref{sec:known}, any lexicographic order is supported in our framework, although with an additional logarithmic factor compared to work on constant-delay enumeration~\cite{bagan07constenum,bakibayev13fordering}.
To be consistent with the rest of the paper that assumes weights on the input tuples, 
we assume that the lexicographic order is given in terms of the relations rather than the query variables.
For this to be well-defined, self-joins need to be de-duplicated and
there must exist a total order on the tuples
within each relation. 

\introparagraph{Reduction to sum-of-weights}
One way to handle the lexicographic order is to reduce it to the sum-of-weights ranking function by defining appropriate weights on the input.
Given two consecutive tuples $r_i, r_{i+1}$ of some relation $R_j$, we want the weight of $r_{i+1}$ to be sufficiently large,
such that any sum $w(r_{i+1}) + \sum_{j' > j} r_{j'}, r_{j'} \in R_{j'}$ 
is
larger than any such sum where $w(r_{i+1})$ is replaced by $w(r_{i})$.
Assuming that all relations have size $n$,
this is guaranteed by setting the weight of the $i^\textrm{th}$ tuple, $i \in [0,n-1]$, in relation $R_j, j \in [0, \ell-1]$
to $i \cdot n^{\ell-1-j}$.

\introparagraph{Dioid}
An alternative way to handle a lexicographic order is to define an appropriate commutative selective dioid; however, this approach incurs an additional factor $\bigO(\ell)$ in time and memory.
We set the domain of the dioid to $W = \N^\stages$, i.e., each
tuple weight is an $\stages$-dimensional integer vector.
Without loss of generality, assume that the total order within a relation is represented
by natural numbers such that input tuple $r$ is associated with $w'(r) \in \N$.
Input tuple $r_j \in R_j$ has
weight $w(r_j) = (0,\ldots, 0, w'(r_j), 0,\ldots, 0)$ 
with zeros except for position $j$ that stores the ``local'' weight value in $R_j$.
The operator $\otimes$ is standard element-wise vector addition; therefore, the weight
of a query answer with witness $(r_1,\ldots, r_\stages)$ is $(w'(r_1),\ldots, w'(r_\stages))$.
To order two such vectors, 
the dioid addition
$\oplus$ 
returns the operand that comes first according to the lexicographic order
e.g.,
for $\stages=2$, $(a,b) \oplus (c,d) = (a,b)$ 
if $w'(a) < w'(c)$, 
or $w'(a) = w'(c)$ and $w'(b) < w'(d)$,
and $(c,d)$ otherwise.
The $\0$ and $\1$ elements of the dioid are
$(\infty,\ldots, \infty)$ and $(0,\ldots, 0)$, respectively.

\subsubsection{Attribute vs tuple weights}
We now discuss how to handle the case where the input weights are assigned to the domain values instead of the tuples.
All algorithms and complexity results also apply to that case since there is a simple linear-time reduction from domain-value weights to tuple weights
so that the weights of the query answers remain the same.
The reduction assigns each variable to one of the atoms that it appears in
and then computes the weight of a tuple by aggregating the weights of the attribute values that have weights assigned to them.
For the lower bound of \Cref{theorem:cq_data_comp},
note that the hardness proof relies only on the fact that ranked enumeration is at least as hard as
unranked enumeration regardless of the ranking function.

\section{Experiments}
\label{sec:experiments}

Since asymptotic complexity only tells part of the story, we compare all algorithms
in terms of actual running time.
The code to reproduce the experiments can be found in our project page
\url{https://northeastern-datalab.github.io/anyk/}. 

\introparagraph{Algorithms}
We implement our proposed algorithms in the same Java 11 environment.
We compare:
(1) \ANYKPART,
(2) \ANYKREC,
(3) \ANYKPARTP,
and (4) \BATCH which computes the full result using the Yannakakis
algorithm~\cite{DBLP:conf/vldb/Yannakakis81} followed by sorting.
Both \ANYKPART and \ANYKPARTP use the \QUICK variant which relies on Incremental Quicksort~\cite{paredes06iqs} to evaluate the successor function.
We also evaluate the performance of two database systems:
(5) \PSQL is PostgreSQL 9.5.20,
and (6) \SYSX is a commercial database system. 

\introparagraph{Queries}
We conduct our study on \emph{acyclic} queries,
since handling cyclic queries is an issue orthogonal to our work (see \Cref{sec:cycles}).
We use three types of queries,
parameterized by their number of atoms $\ell$.
\begin{enumerate}
    \item A \emph{path} query $Q_{P\stages }(\vec x) \datarule R_1(x_0, x_1), R_2(x_1, x_2), \ldots, R_\ell(x_{\stages-1}, x_{\stages})$ joins the relations in a chain and is a case of serial DP (\Cref{sec:cq_to_dp}).
    \item A \emph{star} query $Q_{2S\ell}(\vec x) \datarule R_1(x_0, x_1), R_2(x_0, x_2), \ldots, R_\stages(x_0, x_{\stages})$,
    where all relations join on $x_0$.
    Even though we could construct a path-structured join tree for this query, 
    we instead use a join tree with minimal depth and maximal degree so that it is structurally the opposite of the path query
    and requires treatment as a tree (T-DP).
    \item A \emph{branch} query $Q_{1B\ell}(\vec x) \datarule R_1(x_0, x_1), R_2(x_1, x_2), \ldots, R_{\ell-1}(x_{\stages-2}, x_{\stages-1}), R_\ell(x_{\stages-2}, x_{\stages})$ 
    is similar to a path, but requires treatment as a tree (T-DP) due to a single branch in the join tree.
    It has a low serial-decomposition width $\serialw=2$ (by placing the last 2 atoms together in the serial decomposition).
\end{enumerate}
For real datasets, query atoms refer to the same relation.
We set the ranking function to sum.

\introparagraph{Synthetic data}
Our synthetic data generator
creates input with regular structure 
and tunable parameters.
We generate relations $R_i(A_{i_1}, A_{i_2}, W_i), \, i \!\geq\! 1$, where the columns $A_{i_1}, A_{i_2}$ are used for joins,
while $W_i$ contains tuple weights.
The join distribution is controlled by the sampling process of the values
that populate the $A_{i_1}, A_{i_2}$ columns.
For a \emph{Uniform} distribution, we draw integers from $[0, |\dom|)$ 
uniformly at random with replacement for a given value $|\dom|$, which defaults to $n / 10$. 
(This means that a tuple joins, in expectation, with $10$ others in a joining relation.)
For a \emph{Gaussian} distribution, we round to integers the values drawn with a mean of $0$
and a given standard deviation, which defaults to $n / 10$.
Tuple weights are real numbers uniformly drawn from $[0, 10000]$.

\introparagraph{Real Data}
We use real-world networks where the output size of the joins typically
exceeds the input size.
In \Bitcoin~\cite{Bitcoin_dataset2,Bitcoin_dataset1}, edges
have weights that represent the degree of trust between users.
\Twitter~\cite{Twitter_dataset} models followership among users as edges.
Edge weight is set to the sum of the PageRanks~\cite{brin98pagerank}
of both endpoints. 
To control input size, we retain edges between users whose IDs are below a certain
threshold.
We also use two smaller networks where computing the entire join output is feasible.
\Friendship~\cite{Friendship_dataset,konect,konect:moody} is created from a student survey in which each participant indicated
their best friends.
In \Foodweb~\cite{Foodweb_dataset,konect,konect:foodweb}, an edge indicates
that a taxon uses another taxon as food with a given trophic factor.
\Cref{tab:datasets} summarizes relevant statistics.
Note that the size of our relations $n$ is equal to the number of edges.

\begin{figure}[!tb]
\small
\renewcommand{\tabcolsep}{1.0mm}
\begin{center}
\begin{tabular}{|l|r|r|c|c|}
\hline
Dataset                                                 & Nodes & Edges & Max/Avg Degree & Weights \\ \hline
\Twitter \cite{Twitter_dataset}                          & 131,072 & 3,615,171 & 30,105 / 55.2 & PageRank \\
\Bitcoin \cite{Bitcoin_dataset2,Bitcoin_dataset1}        & 5,881 & 35,592 & \,\,1,298 / 12.1 & Trust  \\
\Friendship \cite{Friendship_dataset,konect,konect:moody}                          & 2,539 & 12,969 & 36 / 10.2 & Interaction strength \\
\Foodweb \cite{Foodweb_dataset,konect,konect:foodweb}                          & 128 & 2,137 & 110 / 33.3 & Trophic factor (feeding level) \\
\hline
\end{tabular} 
\caption{Datasets used for experiments with real data.}
\label{tab:datasets}
\end{center}
\vspace{-4mm}
\end{figure}

\introparagraph{Implementation details}
All experiments are conducted on a machine with Ubuntu Linux 20.04.2, an
Intel Xeon E5-2643 CPU, and 128 GB RAM,
from which 100GB are allocated to the JVM.
Each measurement is the median of at least 20 separate JVM invocations.
To avoid the non-deterministic nature of garbage collection, we try to stay below
the available memory limit in each experiment.
As an optimization to all our algorithms, we initialize their data structures lazily when they are accessed for
the first time. For example, in \ANYKREC, we do not create the priority queue $\Choices_1(v)$ for a node $v$ until this node is visited for the first time.
This can significantly reduce $\TT(k)$ for small $k$.
Notice that our complexity analysis in \cref{sec:complexity}
assumes constant-time inserts for priority queues, which is important for algorithms
that push more elements than they pop per iteration.
This bound is achieved by data structures that are well known to perform poorly
in practice~\cite{cherkassky96shortest,LarkinSenTarjan2004:PQs}. 
Instead, we use the standard Java library binary heaps.

\introparagraph{Tuning the Database Systems}
For \PSQL, following standard methodology~\cite{bakibayev12fdb}, we remove the system overhead
as much as possible and make sure that the input relations are cached in memory by timing the second of two runs.
We turn off fsync, synchronous\_commit, full\_page\_writes, we set bgwriter\_delay to the maximum (10 sec), bgwriter\_lru\_maxpages to 0, checkpoint\_timeout to 1 hour and max\_wal\_size to a large value (1000 GB). 
We also give shared\_buffers and work\_mem 32 GB and set the isolation level to the lowest possible (READ UNCOMMITED).
We follow a similar approach for \SYSX, and also tune it for in-memory computation.

\introparagraph{Methodology}
First, we evaluate the different approaches on top-$k$ queries
where the value of $k$ is fixed to a relatively small value ($k=10^3$).
In this setting, we vary different parameters such as data size, query size, and
join distribution to show the advantage of any-$k$ algorithms
even when enumeration is not required.
For these experiments, we only show \ANYKPARTP as its performance is similar to
\ANYKPART and \ANYKREC because the enumeration phase is dominated by the 
preprocessing (i.e., Dynamic Programming).
Then, we move on to a study of any-$k$ where we compare the different approaches on $\TT(k)$ as the value of $k$ changes.
We look into the regime where the enumeration continues until the last query answer
and also into a regime where relatively fewer answers are returned ($k=n$),
yet the enumeration cost is not dominated by the preprocessing.

\begin{figure*}[h]

    \centering
    \begin{subfigure}{\linewidth}
        \centering
        \includegraphics[height=0.7cm]{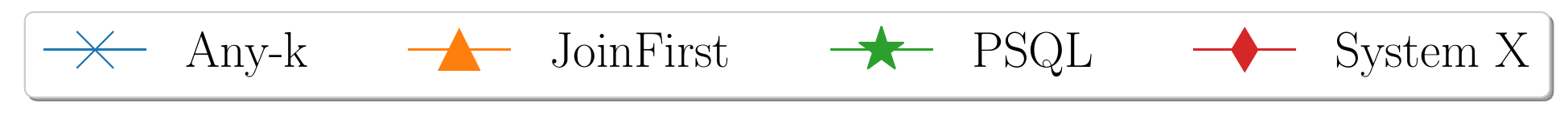}
    \end{subfigure}
    \vspace{-3mm}

    \begin{subfigure}{0.24\linewidth}
        \centering
        \includegraphics[width=\linewidth]{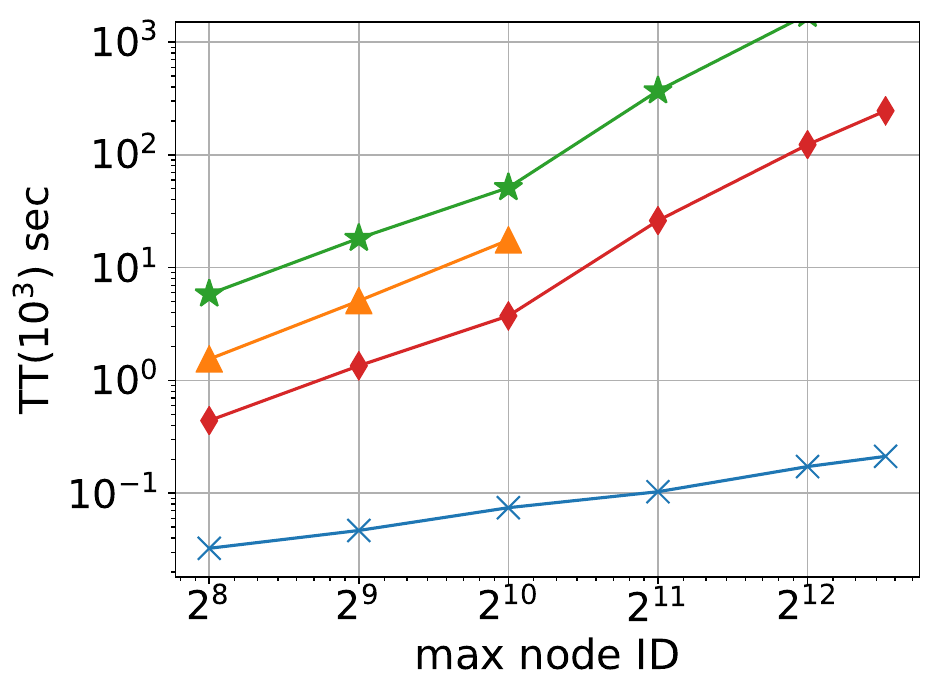}
        \caption{4-Path \Bitcoin\\(varying data size)}
		\label{exp:4path_bc}
    \end{subfigure}%
    \hfill
    \begin{subfigure}{0.24\linewidth}
        \centering
        \includegraphics[width=\linewidth]{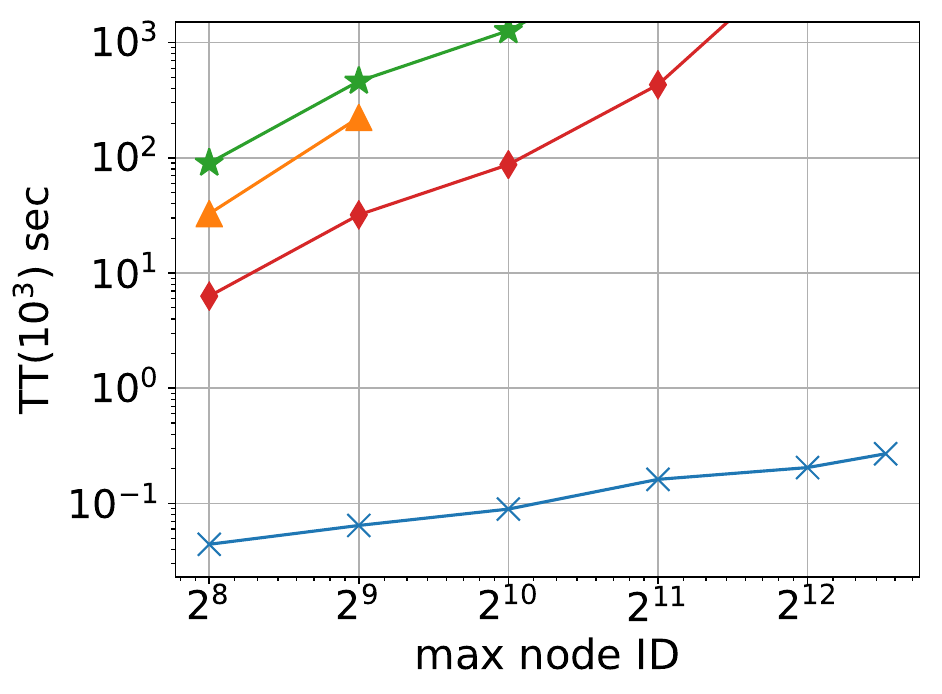}
        \caption{4-Star \Bitcoin\\(varying data size)}
		\label{exp:4star_bc}
    \end{subfigure}%
    \hfill
    \begin{subfigure}{0.24\linewidth}
        \centering
        \includegraphics[width=\linewidth]{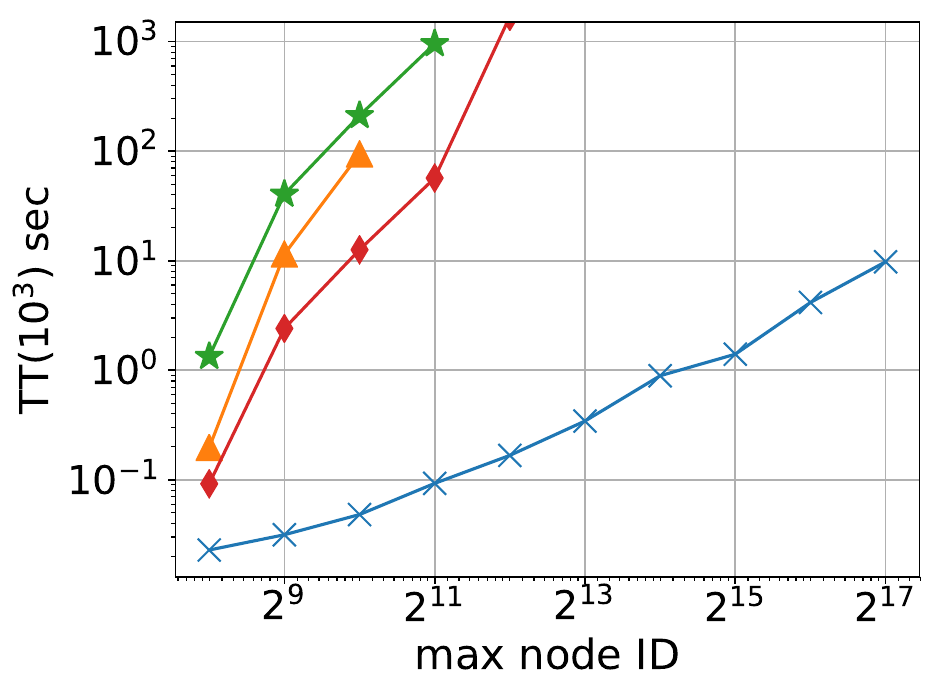}
        \caption{4-Path \Twitter\\(varying data size)}
		\label{exp:4path_tw}
    \end{subfigure}%
    \hfill
    \begin{subfigure}{0.24\linewidth}
        \centering
        \includegraphics[width=\linewidth]{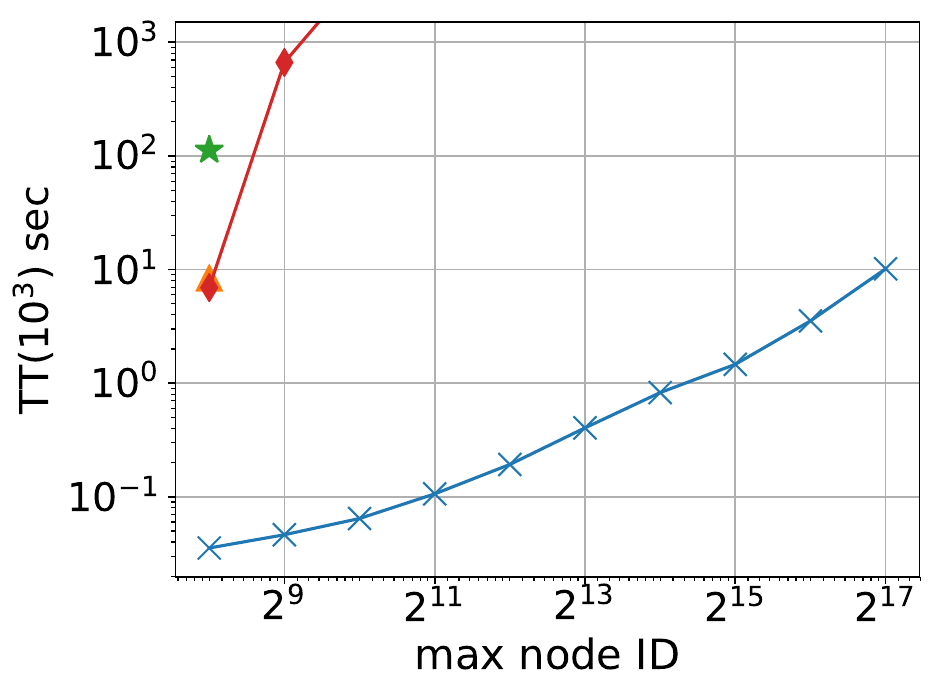}
        \caption{4-Star \Twitter\\(varying data size)}
		\label{exp:4star_tw}
    \end{subfigure}

    \begin{subfigure}{0.24\linewidth}
        \centering
        \includegraphics[width=\linewidth]{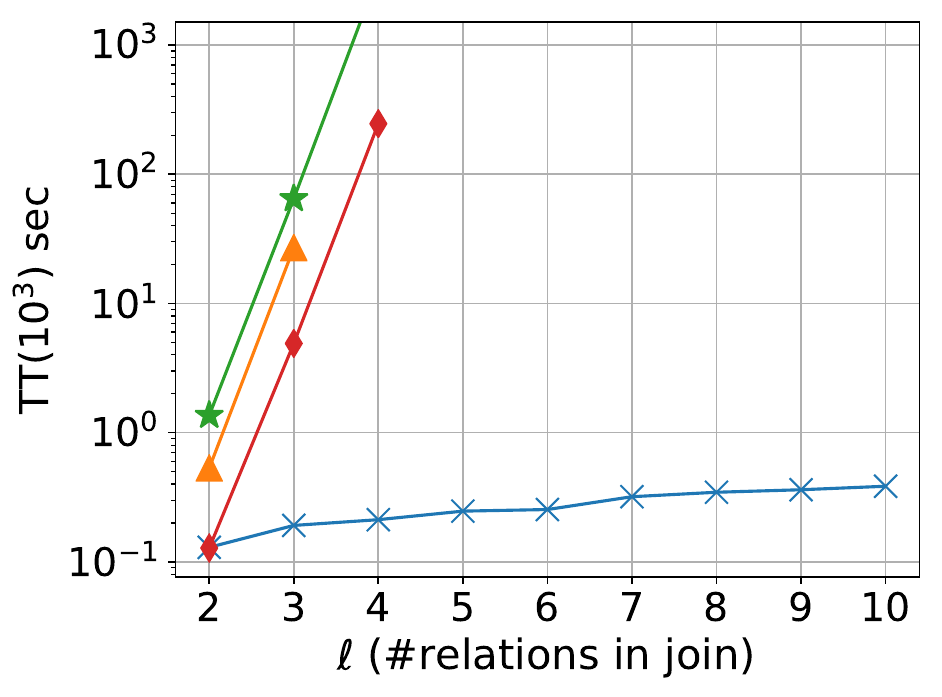}
        \caption{$\ell$-Path \Bitcoin\\(varying query size)}
		\label{exp:path_bc}
    \end{subfigure}%
    \hfill
    \begin{subfigure}{0.24\linewidth}
        \centering
        \includegraphics[width=\linewidth]{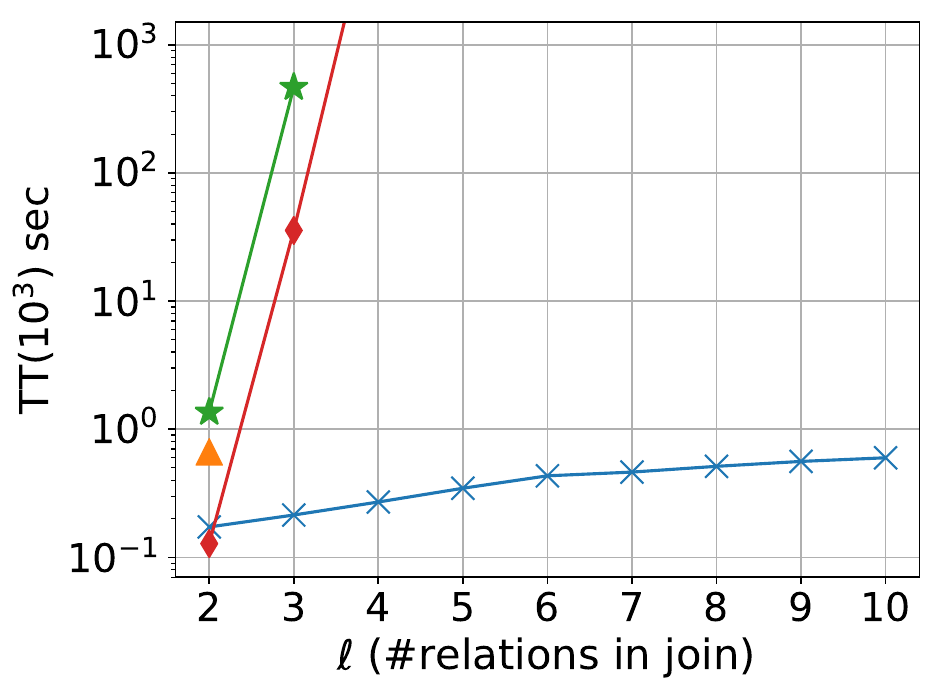}
        \caption{$\ell$-Star \Bitcoin\\(varying query size)}
		\label{exp:star_bc}
    \end{subfigure}%
    \hfill
    \begin{subfigure}{0.24\linewidth}
        \centering
        \includegraphics[width=\linewidth]{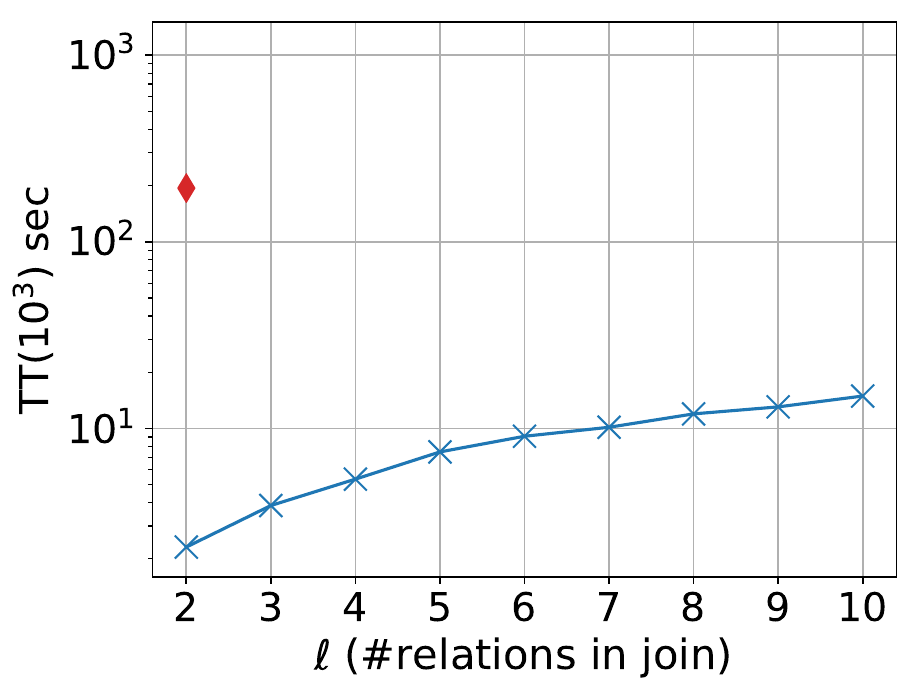}
        \caption{$\ell$-Path \Twitter\\(varying query size)}
		\label{exp:path_tw}
    \end{subfigure}%
    \hfill
    \begin{subfigure}{0.24\linewidth}
        \centering
        \includegraphics[width=\linewidth]{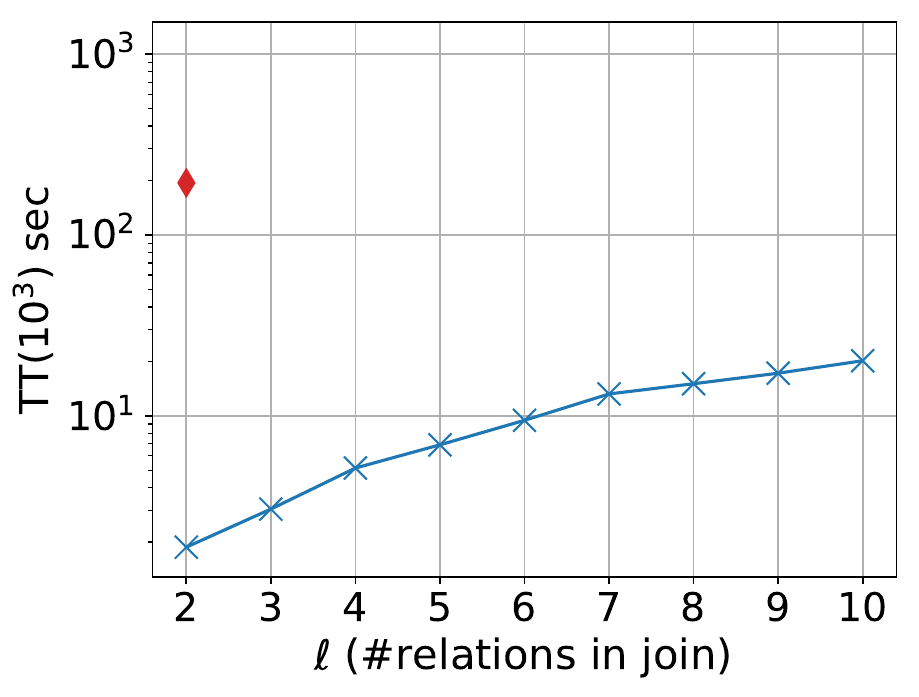}
        \caption{$\ell$-Star \Twitter\\(varying query size)}
		\label{exp:star_tw}
    \end{subfigure}

    \begin{subfigure}{0.24\linewidth}
        \centering
        \includegraphics[width=\linewidth]{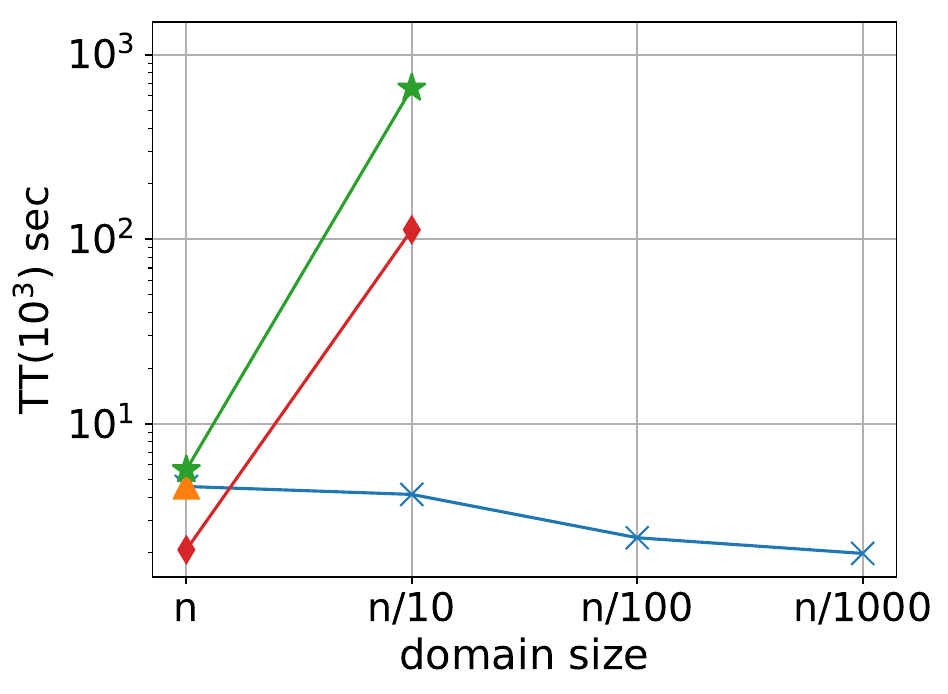}
        \caption{4-Path Synth $n=10^6$ Uniform (varying domain)}
		\label{exp:4path_synthetic_domain}
    \end{subfigure}%
    \hfill
    \begin{subfigure}{0.24\linewidth}
        \centering
        \includegraphics[width=\linewidth]{figs/experiments/path_n1000000_l4.pdf}
        \caption{4-Star Synth $n=10^6$ Uniform (varying domain)}
		\label{exp:4star_synthetic_domain}
    \end{subfigure}%
    \hfill
    \begin{subfigure}{0.24\linewidth}
        \centering
        \includegraphics[width=\linewidth]{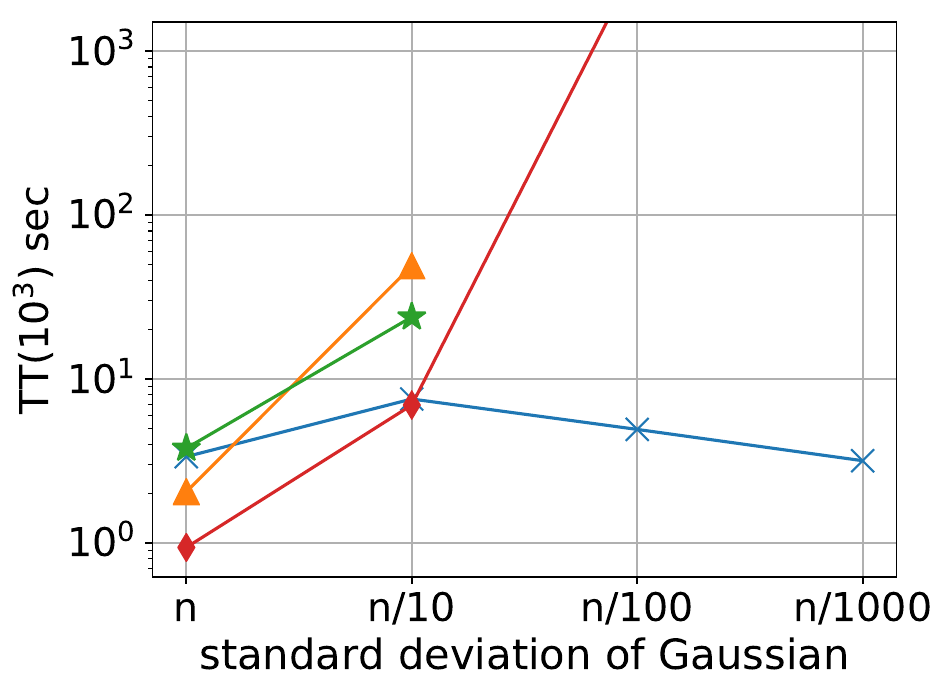}
        \caption{4-Path Synth $n=10^6$ Gaussian (varying std)}
		\label{exp:4path_synthetic_gauss}
    \end{subfigure}%
    \hfill
    \begin{subfigure}{0.24\linewidth}
        \centering
        \includegraphics[width=\linewidth]{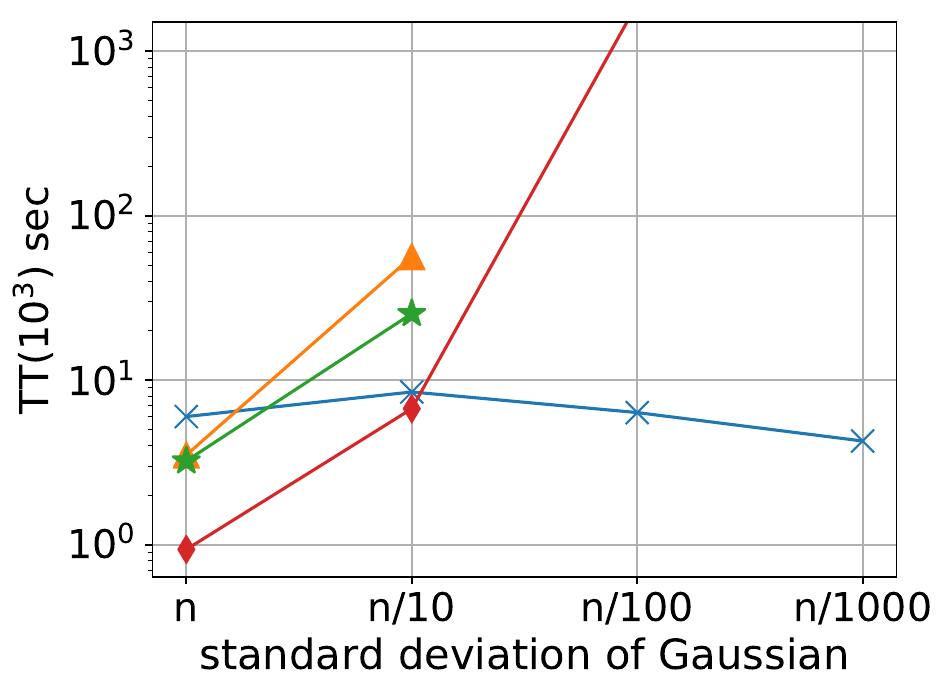}
        \caption{4-Star Synth $n=10^6$ Gaussian (varying std)}
		\label{exp:4star_synthetic_gauss}
    \end{subfigure}

    \caption{Experiments with fixed $k=10^3$.}
    \label{exp:topk}
\end{figure*}

\subsection{Results for Fixed \#Answers $k$}
\label{sec:experiments_topk}

\Cref{exp:topk} shows the time taken to retrieve the top $k=10^3$ answers for different settings.
In the real datasets \Bitcoin and \Twitter,
increasing the size of the data (\Cref{exp:4path_bc,exp:4star_bc,exp:4path_tw,exp:4star_tw})
or increasing the size of the query (\Cref{exp:path_bc,exp:star_bc,exp:path_tw,exp:star_tw})
results in the \BATCH approach becoming infeasible,
either running out of memory or exceeding the timeout of 1 hour.
\PSQL and \SYSX exhibit a similar behavior since they implement a similar evaluation strategy.

We observe a similar effect with synthetic data if we vary the join distribution (\Cref{exp:4path_synthetic_domain,exp:4star_synthetic_domain,exp:4path_synthetic_gauss,exp:4star_synthetic_gauss}).
Increasing the size of the join output,
either by decreasing the domain size for the Uniform distribution
or by decreasing the standard deviation for the Guassian distribution (while maintaining a fixed relation size $n=10^6$),
results in the other approaches faltering,
while any-$k$ is stable or even faster.
On the contrary, a relatively small join output size, e.g., when domain size equals relation size $n$,
favors the \BATCH strategy and the database systems that implement it.
The overhead of any-$k$ in these cases is not excessive; it remains within a factor of 2 from \BATCH in all cases.

\begin{figure*}[h]

    \centering
    \begin{subfigure}{\linewidth}
        \centering
        \includegraphics[height=0.6cm]{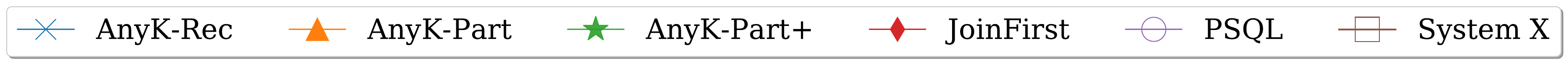}
    \end{subfigure}
    \vspace{-3mm}

    \begin{subfigure}{0.24\linewidth}
        \centering
        \includegraphics[width=\linewidth]{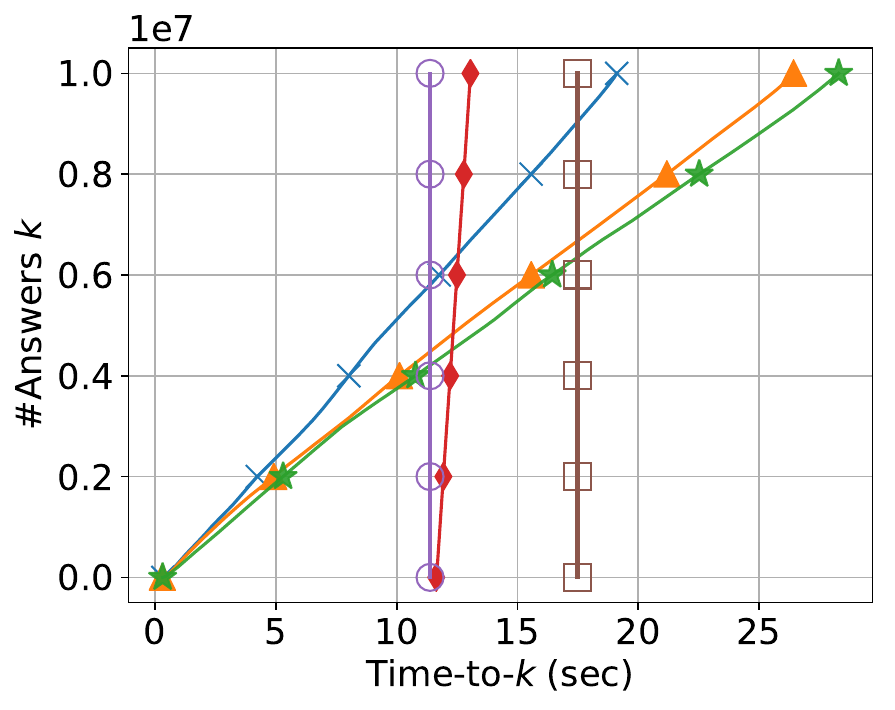}
        \caption{3-Path $n\!=\!10^5$}
		\label{exp:allk_3path}
    \end{subfigure}%
    \hfill
    \begin{subfigure}{0.24\linewidth}
        \centering
        \includegraphics[width=\linewidth]{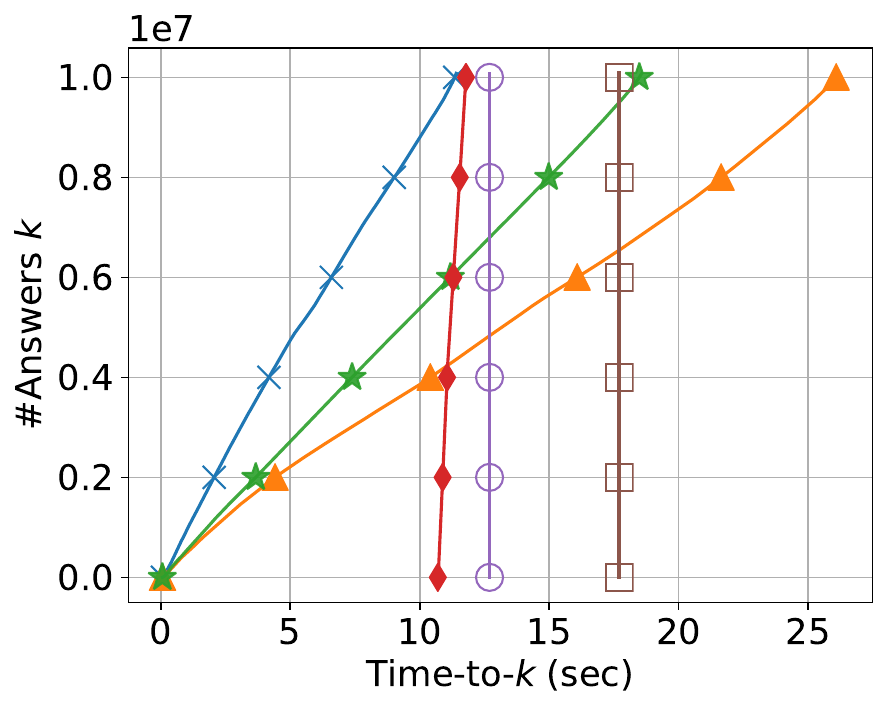}
        \caption{4-Path $n\!=\!10^4$}
		\label{exp:allk_4path}
    \end{subfigure}%
    \hfill
    \begin{subfigure}{0.24\linewidth}
        \centering
        \includegraphics[width=\linewidth]{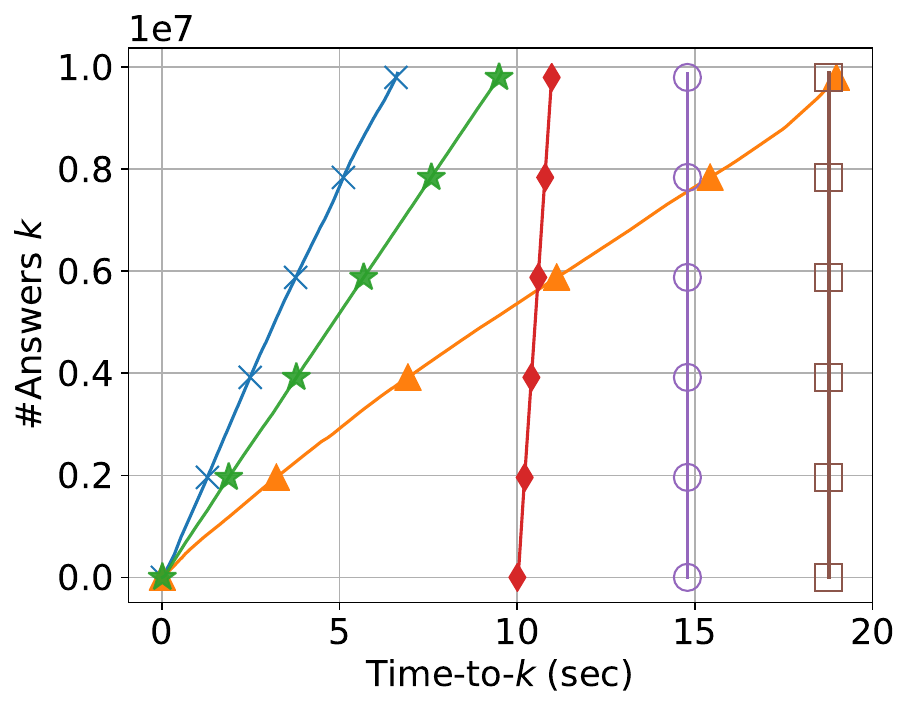}
        \caption{5-Path $n\!=\!10^3$}
		\label{exp:allk_5path}
    \end{subfigure}%
    \hfill
    \begin{subfigure}{0.24\linewidth}
        \centering
        \includegraphics[width=\linewidth]{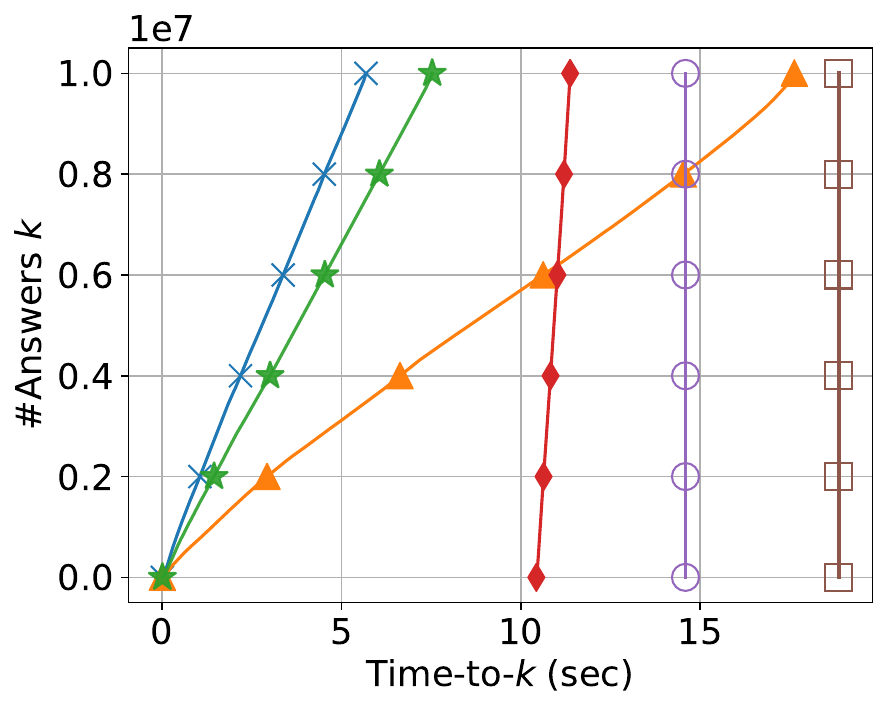}
        \caption{6-Path Synth $n\!=\!10^2$}
		\label{exp:allk_6path}
    \end{subfigure}

    \begin{subfigure}{0.24\linewidth}
        \centering
        \includegraphics[width=\linewidth]{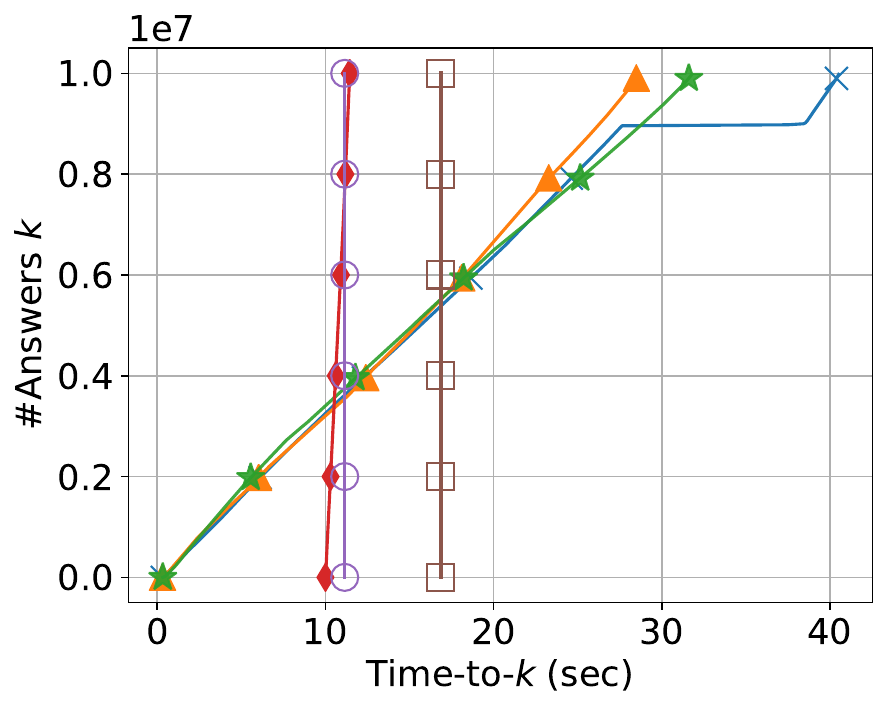}
        \caption{3-Star $n\!=\!10^5$}
		\label{exp:allk_3star}
    \end{subfigure}%
    \hfill
    \begin{subfigure}{0.24\linewidth}
        \centering
        \includegraphics[width=\linewidth]{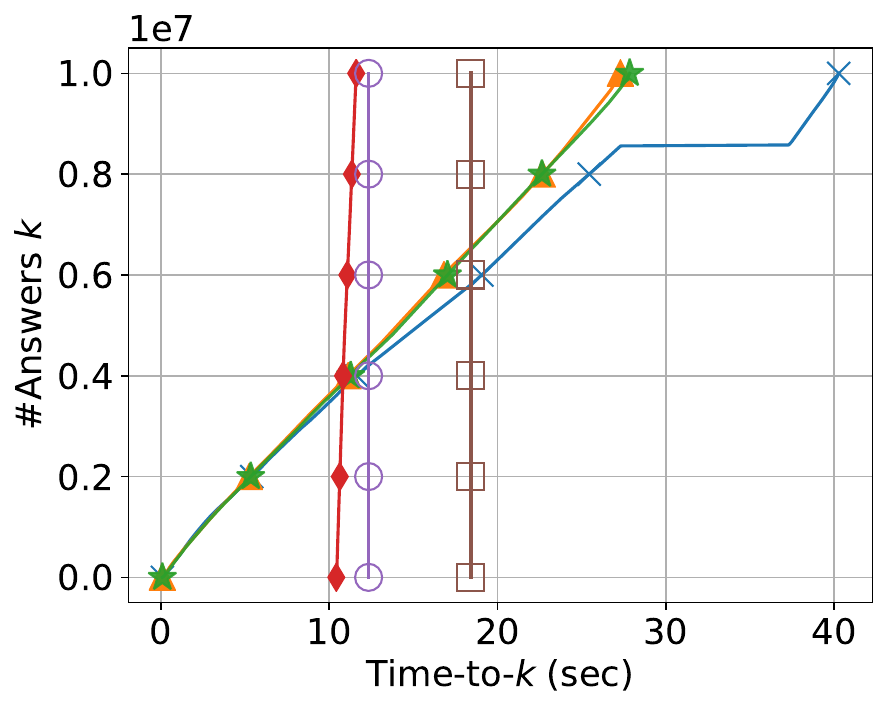}
        \caption{4-Star $n\!=\!10^4$}
		\label{exp:allk_4star}
    \end{subfigure}%
    \hfill
    \begin{subfigure}{0.24\linewidth}
        \centering
        \includegraphics[width=\linewidth]{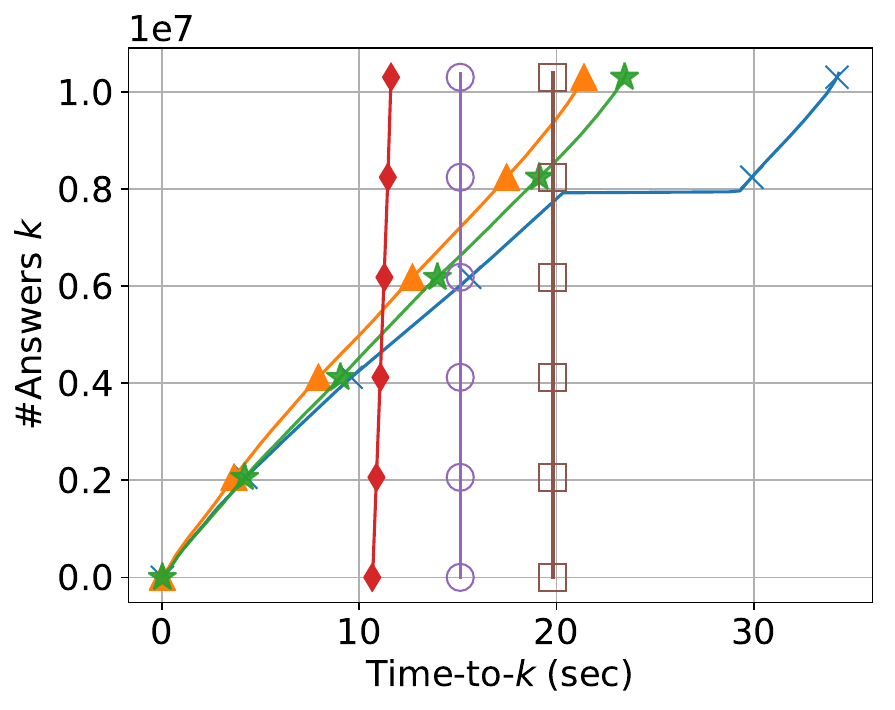}
        \caption{5-Star $n\!=\!10^3$}
		\label{exp:allk_5star}
    \end{subfigure}%
    \hfill
    \begin{subfigure}{0.24\linewidth}
        \centering
        \includegraphics[width=\linewidth]{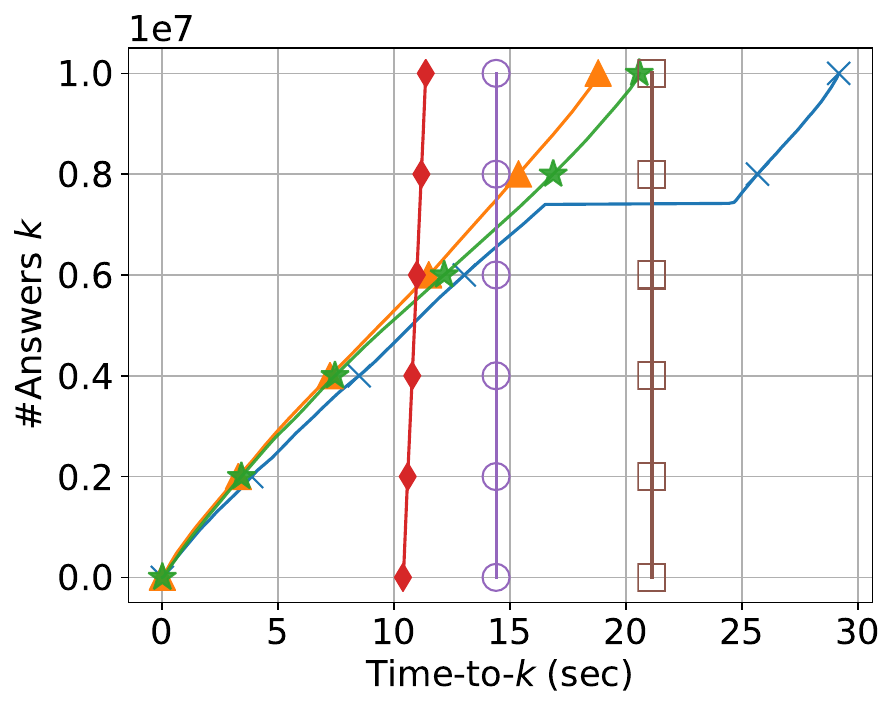}
        \caption{6-Star $n\!=\!10^2$}
		\label{exp:allk_6star}
    \end{subfigure}

    \begin{subfigure}{0.24\linewidth}
        \centering
        \includegraphics[width=\linewidth]{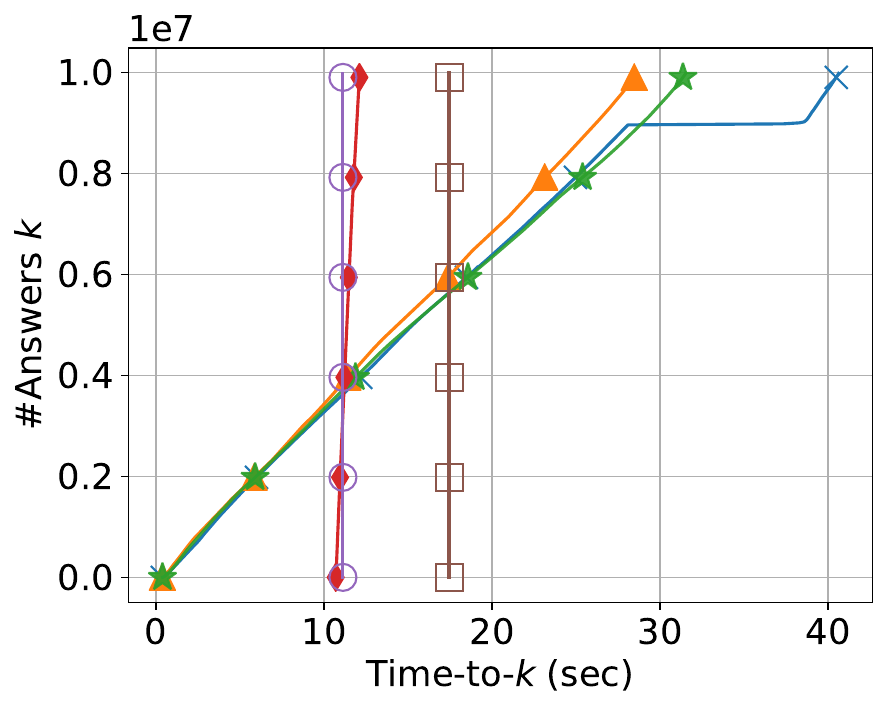}
        \caption{3-Branch $n\!=\!10^5$}
		\label{exp:allk_3branch}
    \end{subfigure}%
    \hfill
    \begin{subfigure}{0.24\linewidth}
        \centering
        \includegraphics[width=\linewidth]{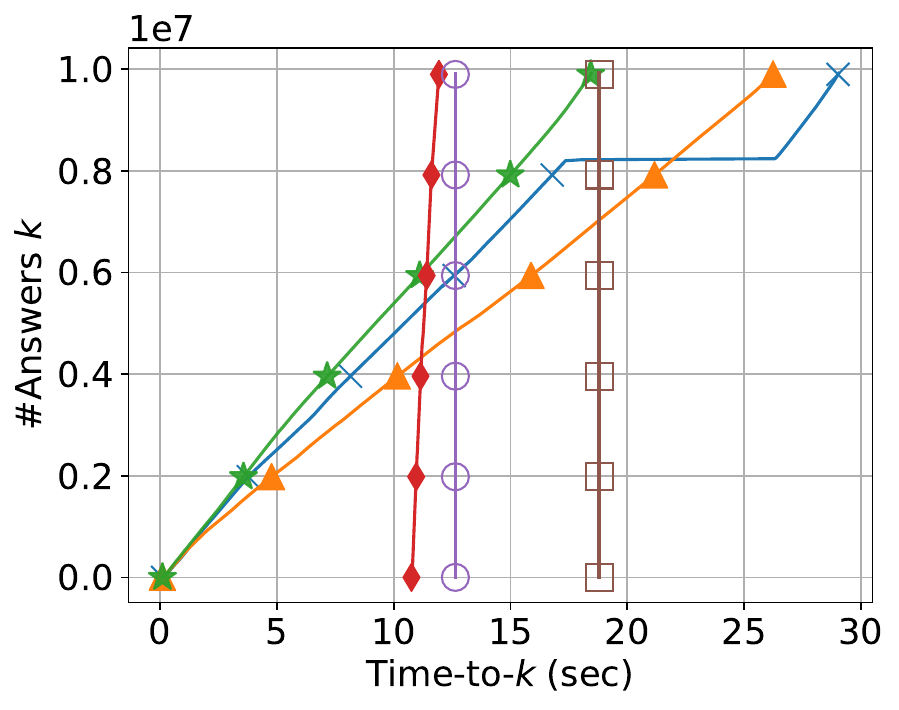}
        \caption{4-Branch $n\!=\!10^4$}
		\label{exp:allk_4branch}
    \end{subfigure}%
    \hfill
    \begin{subfigure}{0.24\linewidth}
        \centering
        \includegraphics[width=\linewidth]{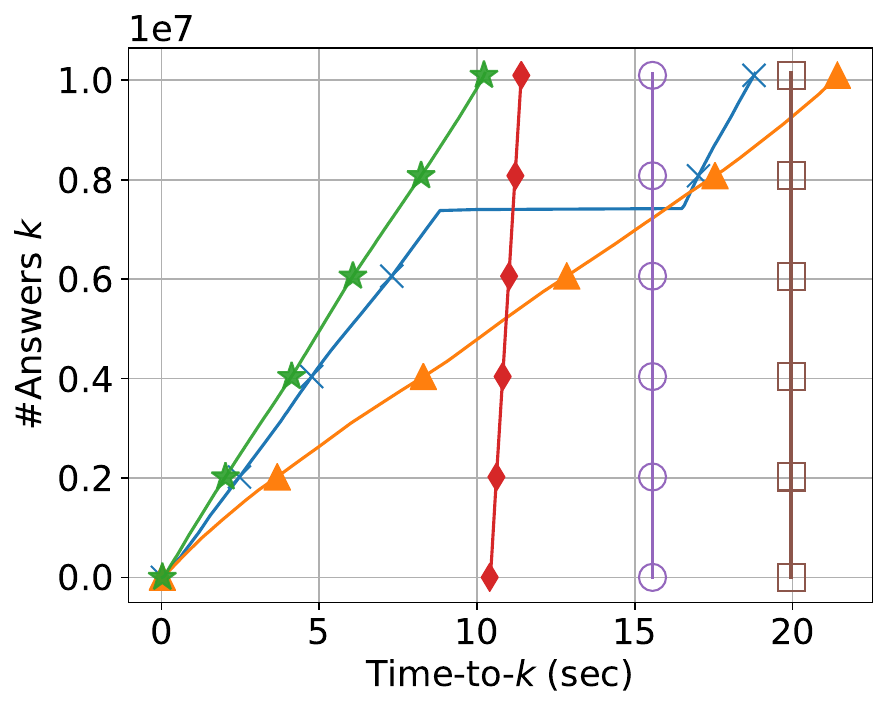}
        \caption{5-Branch $n\!=\!10^3$}
		\label{exp:allk_5branch}
    \end{subfigure}%
    \hfill
    \begin{subfigure}{0.24\linewidth}
        \centering
        \includegraphics[width=\linewidth]{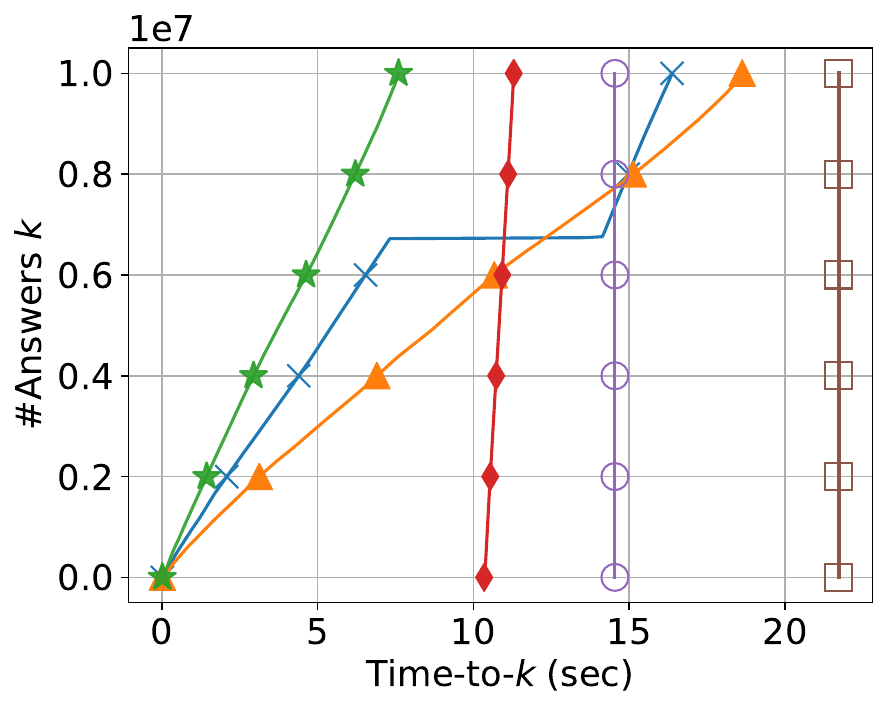}
        \caption{6-Branch $n\!=\!10^2$}
		\label{exp:allk_6branch}
    \end{subfigure}

    \caption{Experiments with varying $k$, enumerating \emph{all query answers} on Synthetic Uniform data ($|\dom|=n/10$).
    The value of $n$ is chosen so that output size is approximately the same ($\sim 10^7$) for different query sizes.}
    \label{exp:allk_synthetic}
\end{figure*}

\begin{figure*}[h]

    \centering
    \begin{subfigure}{\linewidth}
        \centering
        \includegraphics[height=0.6cm]{figs/experiments/legend2.pdf}
    \end{subfigure}
    \vspace{-3mm}
    
    \begin{subfigure}{0.24\linewidth}
        \centering
        \includegraphics[width=\linewidth]{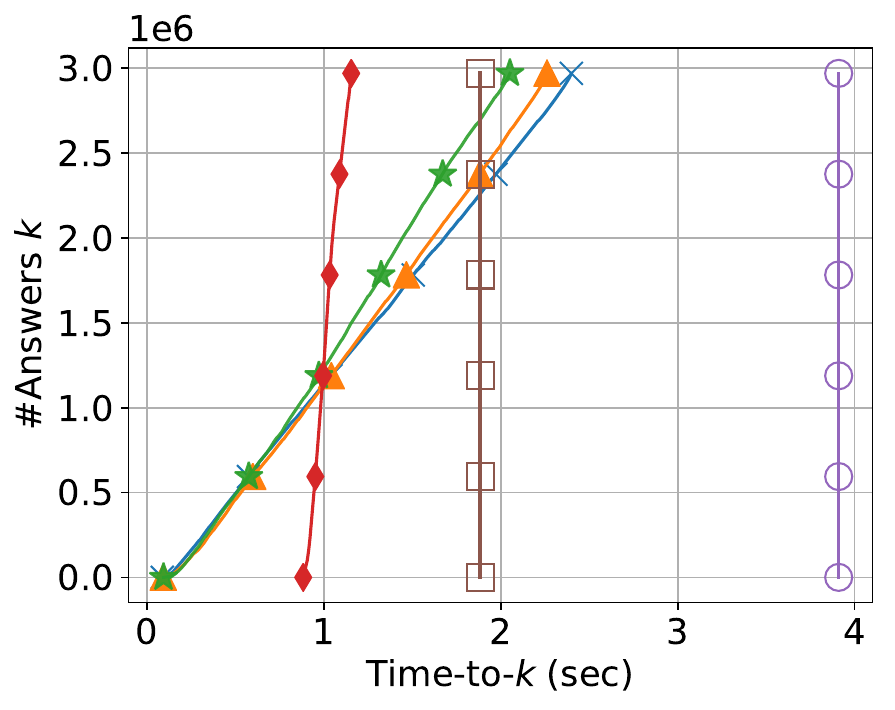}
        \caption{4-Path \Friendship}
		\label{exp:allk_4path_friendship}
    \end{subfigure}%
    \hfill
    \begin{subfigure}{0.24\linewidth}
        \centering
        \includegraphics[width=\linewidth]{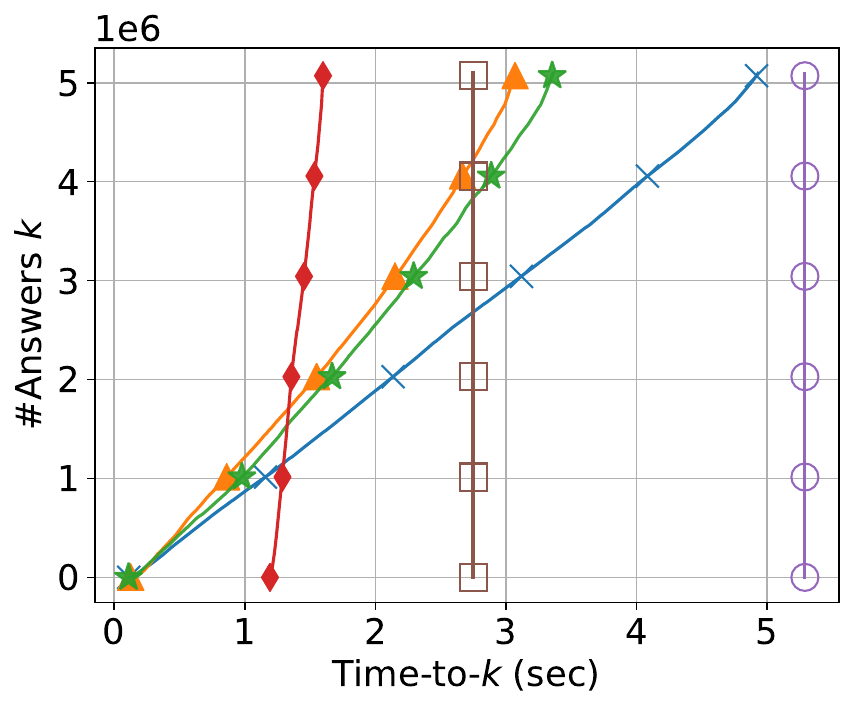}
        \caption{4-Star \Friendship}
		\label{exp:allk_4star_friendship}
    \end{subfigure}%
    \hfill
    \begin{subfigure}{0.24\linewidth}
        \centering
        \includegraphics[width=\linewidth]{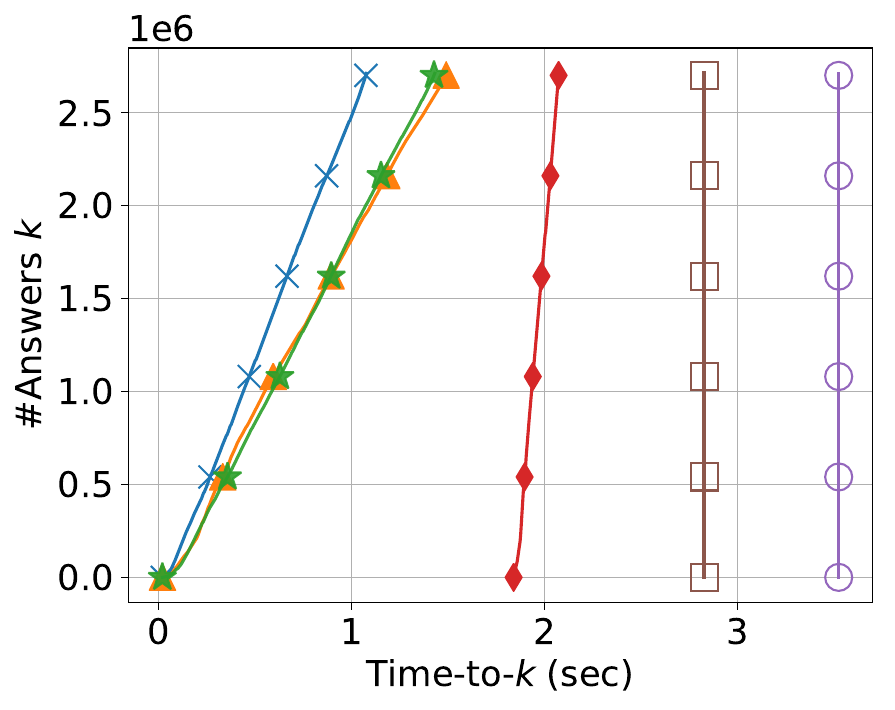}
        \caption{4-Path \Foodweb}
		\label{exp:allk_4path_foodweb}
    \end{subfigure}%
    \hfill
    \begin{subfigure}{0.24\linewidth}
        \centering
        \includegraphics[width=\linewidth]{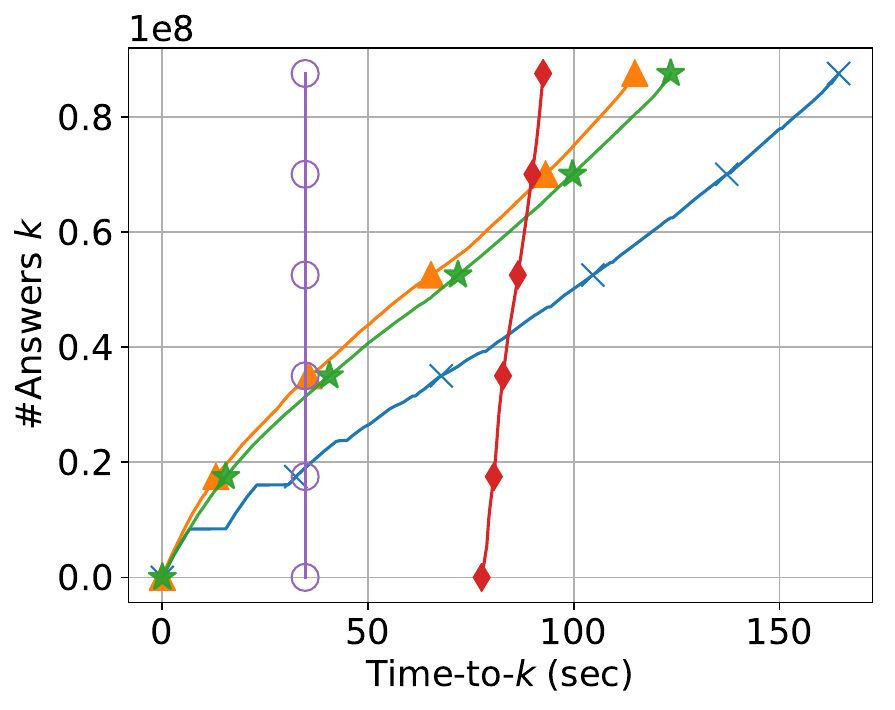}
        \caption{4-Star \Foodweb}
		\label{exp:allk_4star_foodweb}
    \end{subfigure}

    \caption{Experiments with varying $k$, enumerating \emph{all query answers} on real data.}
    \label{exp:allk_real}
\end{figure*}

\begin{figure*}[h]

    \centering
    \begin{subfigure}{\linewidth}
        \centering
        \includegraphics[height=0.6cm]{figs/experiments/legend2.pdf}
    \end{subfigure}
    \vspace{-3mm}
    
    \begin{subfigure}{0.24\linewidth}
        \centering
        \includegraphics[width=\linewidth]{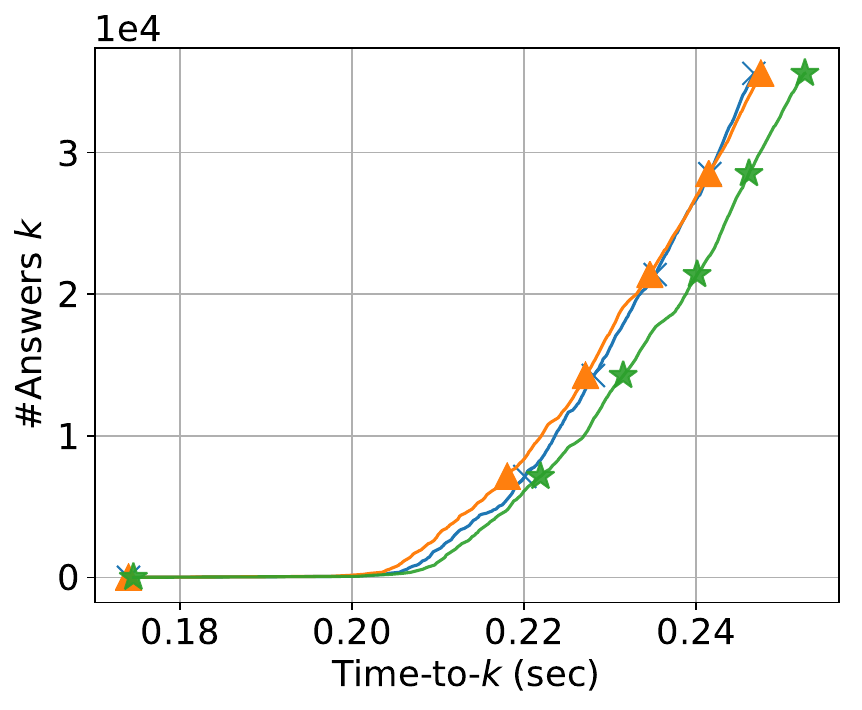}
        \caption{4-Path \Bitcoin \\output size $\sim 4 \!\cdot\! 10^{9}$}
		\label{exp:fewk_4path_bitcoin}
    \end{subfigure}%
    \hfill
    \begin{subfigure}{0.24\linewidth}
        \centering
        \includegraphics[width=\linewidth]{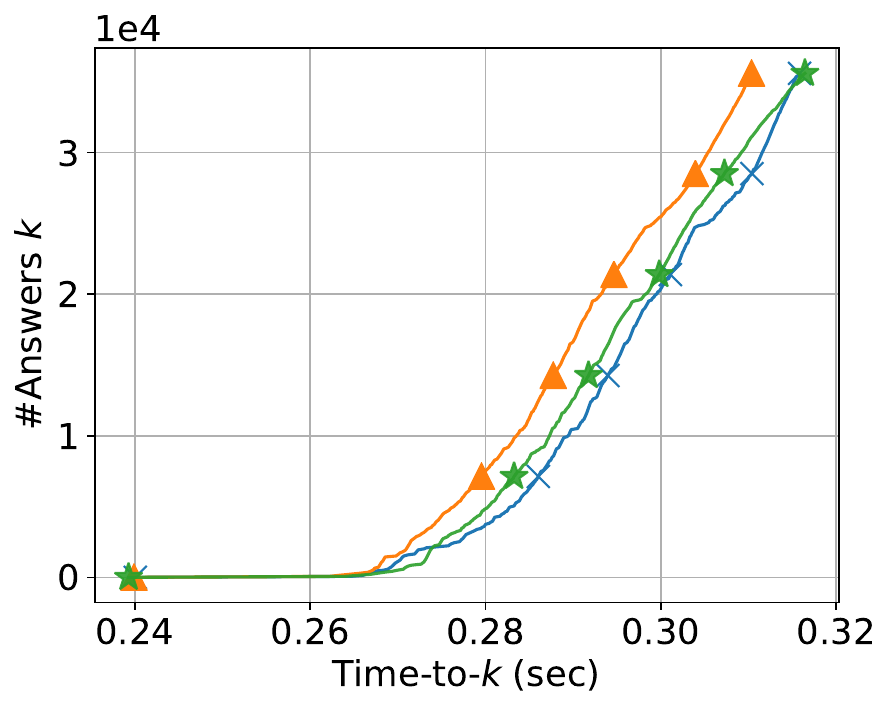}
        \caption{6-Path \Bitcoin \\output size $\sim 8 \!\cdot\! 10^{12}$}
		\label{exp:fewk_6path_bitcoin}
    \end{subfigure}%
    \hfill
    \begin{subfigure}{0.24\linewidth}
        \centering
        \includegraphics[width=\linewidth]{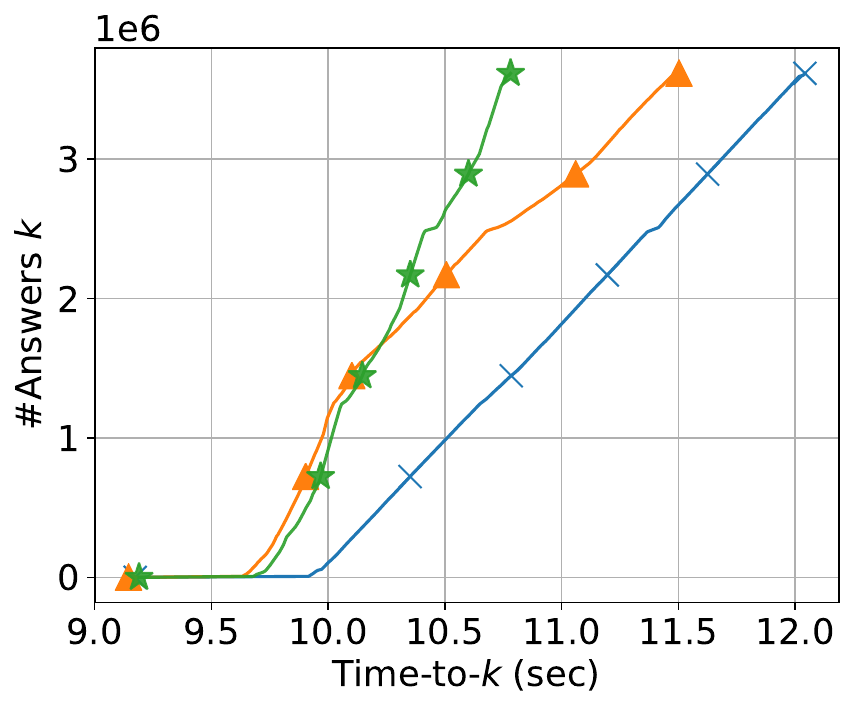}
        \caption{4-Path \Twitter \\output size $\sim 3 \!\cdot\! 10^{15}$}
		\label{exp:fewk_4path_twitter}
    \end{subfigure}%
    \hfill
    \begin{subfigure}{0.24\linewidth}
        \centering
        \includegraphics[width=\linewidth]{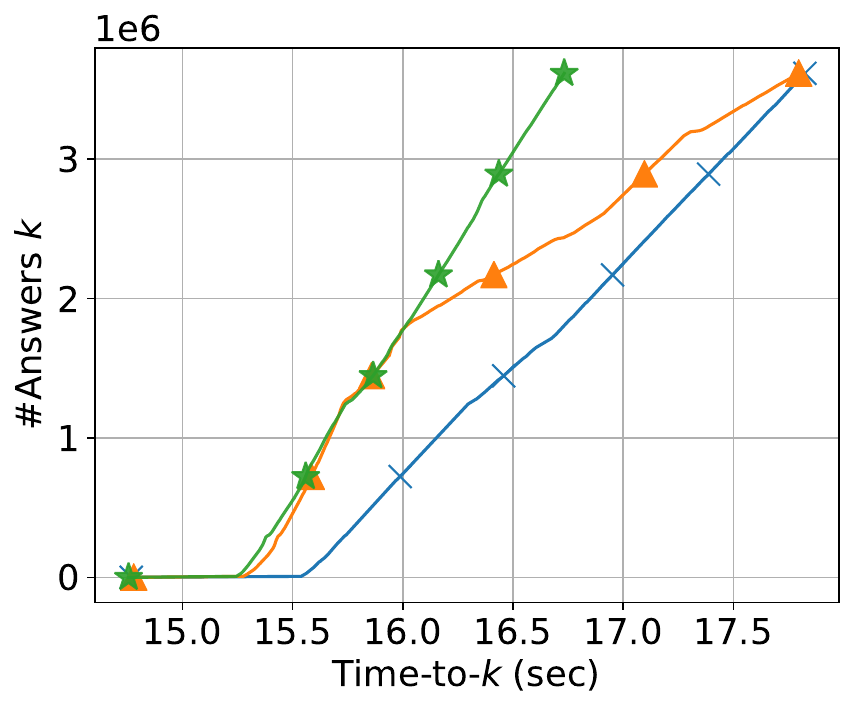}
        \caption{6-Path \Twitter \\output size $\sim 1 \!\cdot\! 10^{21}$}
		\label{exp:fewk_6path_twitter}
    \end{subfigure}

    \begin{subfigure}{0.24\linewidth}
        \centering
        \includegraphics[width=\linewidth]{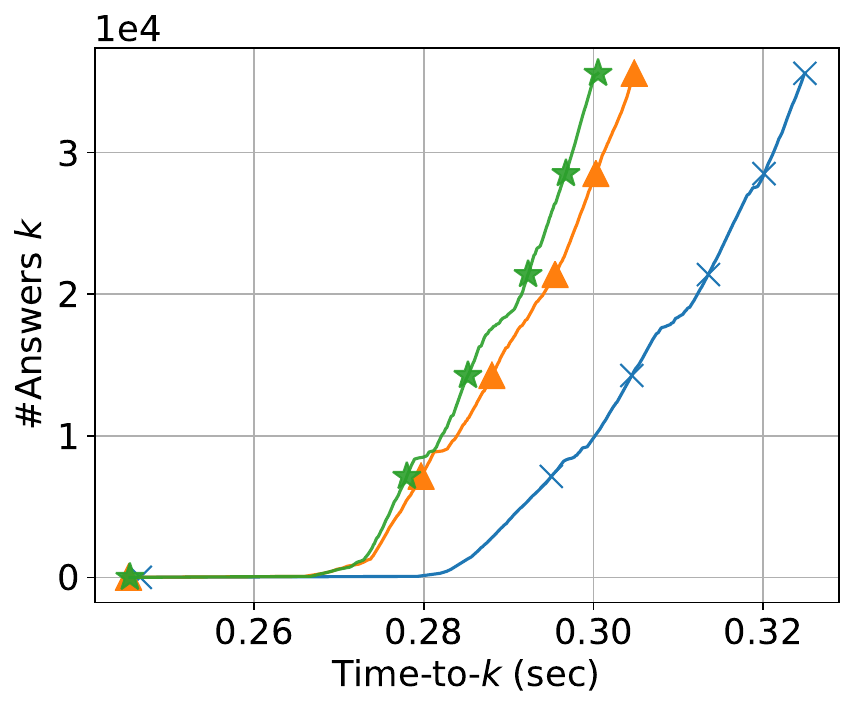}
        \caption{4-star \Bitcoin \\output size $\sim 5 \!\cdot\! 10^{11}$}
		\label{exp:fewk_4star_bitcoin}
    \end{subfigure}%
    \hfill
    \begin{subfigure}{0.24\linewidth}
        \centering
        \includegraphics[width=\linewidth]{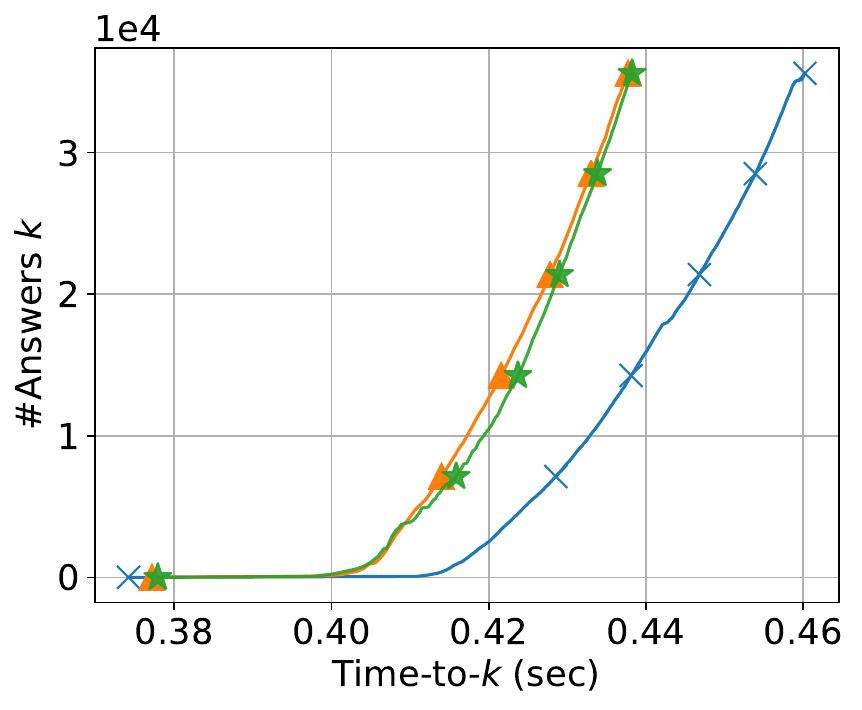}
        \caption{6-star \Bitcoin \\output size $\sim 2 \!\cdot\! 10^{17}$}
		\label{exp:fewk_6star_bitcoin}
    \end{subfigure}%
    \hfill
    \begin{subfigure}{0.24\linewidth}
        \centering
        \includegraphics[width=\linewidth]{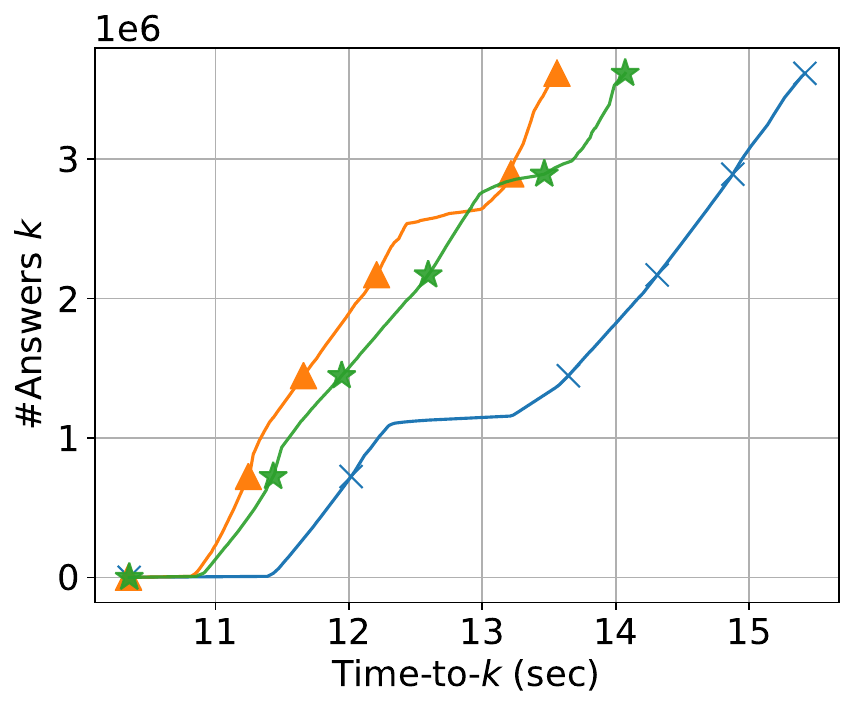}
        \caption{4-star \Twitter \\output size $\sim 5 \!\cdot\! 10^{17}$}
		\label{exp:fewk_4star_twitter}
    \end{subfigure}%
    \hfill
    \begin{subfigure}{0.24\linewidth}
        \centering
        \includegraphics[width=\linewidth]{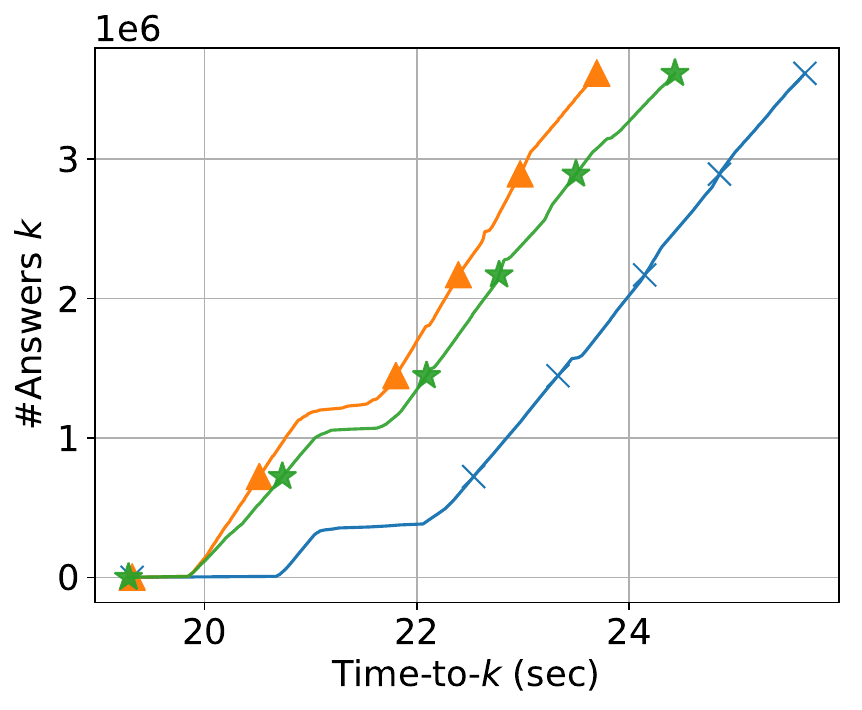}
        \caption{6-star \Twitter \\output size $\sim 6 \!\cdot\! 10^{25}$}
		\label{exp:fewk_6star_twitter}
    \end{subfigure}

    \caption{Experiments with varying $k$, enumerating \emph{few query answers} ($k=n$) on real data.}
    \label{exp:fewk_real}
\end{figure*}

\subsection{Results for Varying \#Answers $k$}
\label{sec:experiments_anyk}

We report the number of query answers returned in ranked order (y-axis) over time (x-axis).
In \Cref{exp:allk_synthetic}, we test the different approaches when all query answers are returned using synthetic data,
where we control the input to achieve a join output size in the order of $10^7$.
All three any-$k$ algorithms return the first few million answers
much faster in all cases.
For path queries with $\ell > 3$ (\Cref{exp:allk_4path,exp:allk_5path,exp:allk_6path}),
\ANYKREC not only returns the first answers quickly,
but also finishes the enumeration before \BATCH. 
The advantage becomes more evident as we increase the size of the query because we keep the output size fixed, which results in denser joins where the opportunities for memoization are greater.
\ANYKPARTP exhibits the same behavior.
The situation is different for star queries 
(\Cref{exp:allk_3star,exp:allk_4star,exp:allk_5star,exp:allk_6star}),
where we use a depth-2 join tree and these two algorithms
do not exploit the memoization of the suffix orders.
\ANYKREC is the slowest for \TTL out of the three any-$k$ options,
while the overhead of \ANYKPARTP compared to \ANYKPART is not significant. 
Interestingly, \ANYKPARTP retains the \TTL advantage over \BATCH for branch queries
(\Cref{exp:allk_5branch,exp:allk_6branch})
whose join tree is very close to a path.

We next test the entire output enumeration on the two smaller real datasets.
We observe that the \TTL advantage over \BATCH
only occurs for \Foodweb (\Cref{exp:allk_4path_foodweb}) and not for the less-dense \Friendship
(\Cref{exp:allk_4path_friendship}).
Similarly to the synthetic case,
\ANYKREC underperforms for star queries (\Cref{exp:allk_4star_friendship,exp:allk_4star_foodweb}).
\SYSX did not manage to run in-memory for the star query on \Foodweb
because of buffer size limits.

Finally, we conduct experiments where the join output is very large
and ask only for $k=n$ answers.
\Cref{exp:fewk_real} shows the results
for \Bitcoin and \Twitter,
noting the output size of the join in each case.
In this regime, \BATCH, \PSQL, and \SYSX
exceed the timeout and are not depicted.
\ANYKPART wins over \ANYKREC for these smaller values of $k$.
\ANYKPARTP is close to \ANYKPARTP
and in some cases overtakes it 
thanks to its memoization strategy(\Cref{exp:fewk_4path_twitter,exp:fewk_6path_twitter}).

\resultbox{%
\introparagraph{Summary}
\BATCH is impractical for joins with large output
and is outperformed by any-$k$ by orders of magnitude.
\PSQL and \SYSX follow a strategy similar to \BATCH.
\ANYKREC is usually the fastest for path queries when the enumeration is carried out to the end,
and may even beat \BATCH for \TTL,
but does not work as well for non-path-like queries.
\ANYKPART is marginally faster when fewer answers are returned.
\ANYKPARTP achieves the best of both worlds, 
having a competitive \TTL for all queries,
and a similar or better performance to \ANYKPART
for small $k$.
}

\section{Related Work}

\introparagraph{Top-$k$ joins}
Top-k queries have received significant attention in the database community
\cite{Agrawal:2009:CJA:1546683.1547478,akbarinia11topk,DBLP:conf/vldb/BastMSTW06,bruno02,chaudhuri99,ilyas08survey,issa20topk,rahul19topk,DBLP:journals/vldb/TheobaldBMSW08,tsaparas03topk}.
Much of that work relies on the value of $k$ given in advance in order to
\emph{prune the search space}. Besides, the cost model introduced by the seminal
Threshold Algorithm (TA)~\cite{fagin03} only accounts for the \emph{cost of fetching} input tuples
from external sources. Later approaches such as
J*~\cite{natsev01},
Rank-Join~\cite{ilyas04},
LARA-J* \cite{mamoulis07lara},
and a-FRPA \cite{finger09frpa}
generalize TA to complex joins like the ones we consider,
yet still focus on minimizing the number of accessed input tuples.
Thus, they are
\emph{sub-optimal when accounting for all steps of
the computation}, including intermediate result size
(see \Cref{appendix:topk_suboptimal}).
In a recent tutorial \cite{tziavelis20tutorial}, we examined
the relationship between top-k joins and the ranked enumeration paradigm discussed in this work.
Bonifati et al.~\cite{bonifati21threshold} study the impact of thresholds (i.e., LIMIT clauses) in query evaluation, but without ranking.

\introparagraph{Optimal join processing}
\emph{Acyclic} Boolean CQs can be evaluated optimally in $\O(n)$ data complexity by the Yannakakis algorithm \cite{DBLP:conf/vldb/Yannakakis81}.
The AGM bound~\cite{AGM}, a tight upper bound on the output size for full CQs,
motivated worst-case optimal join algorithms
\cite{arroyuelo22,
ngo18open,
ngo2018worst,
Ngo:2014:SSB:2590989.2590991,
veldhuizen14leapfrog}
and was generalized to take into account functional dependencies~\cite{gottlob12fds}
and degree constraints~\cite{abo2016degree,khamis17panda}. 
The complexity of answering \emph{cyclic} Boolean CQs has improved over the years thanks to decompositions
that achieve ever smaller width-measures~\cite{chekuri97querydecomp,
GottlobLS:2002,
RobertsonS:1986},
such as fractional hypertree width~\cite{grohe14fhtw}
and submodular width~\cite{Marx:2013:THP:2555516.2535926}.
Work has also been done in the direction of achieving instance-optimality~\cite{alway21optimal,Khamis:2016:JVG:3014437.2967101,ngo14mine}.
Our approach can leverage prior work on decompositions
to obtain the top-ranked answer in the same time bound as the Boolean CQ. 
We also show that it is possible to achieve better $\TTL$ complexity 
than sorting the output of a join algorithm if the ranking is integrated into the join.

\introparagraph{Factorization}
Our efficient encoding of the CQ answers as a (T-)DP instance with intermediate nodes 
in-between bicliques~\cite{feder91cliques} is a type of factorization,
which has been studied systematically in the context of factorized databases~\cite{bakibayev12fdb,DBLP:conf/icdt/KaraO18,
olteanu16record,olteanu12ftrees,olteanu15dtrees}.
The main idea is to represent the query answers compactly while supporting
different types of tasks,
such as enumeration~\cite{kara21triangles,olteanu15dtrees},
aggregation~\cite{AboKhamis:2019:FAQ:3294052.3319694,abo16faq},
training machine learning models~\cite{khamis17ml,schleich16ml,shaikhha21ml},
uniform sampling~\cite{carmeli20random}
and directly accessing ranked answers~\cite{carmeli23direct}.
We build on that body of work, showing that ranked enumeration can also be performed
efficiently for certain ranking functions.
Through factorization, CQs with ``short'' inequality predicates 
(i.e., between adjacent relations in the join tree)
can be reduced to CQs without inequalities over (poly)logarithmically larger relations~\cite{tziavelis21inequalities};
this extends the cases where our ranked enumeration techniques apply.

\introparagraph{Unranked enumeration}
Significant effort has been made to identify the classes
of queries whose answers can be enumerated
\emph{in arbitrary sequence}
with linear preprocessing and constant delay in data complexity~\cite{amarilli21spanners,
bagan07constenum, 
Berkholz:2017:ACQ:3034786.3034789,
carmeli21ucqs,
idris20dynamic_theta,
DBLP:journals/sigmod/Segoufin15,
qichen23change}.
Much of the focus is on the less-practically-relevant measure of delay (see \Cref{sec:complexity_measures}),
sometimes exploring its tradeoff with preprocessing~\cite{deep21projections,deep18compressed}
instead of focusing on $\TT(k)$.
Constant-delay enumeration is possible for cyclic CQs with a higher preprocessing cost by 
decomposing them into acyclic CQs~\cite{berkholz19submodular,olteanu15dtrees}.
Our work shows that adding ranking to enumeration of CQ answers (with an \smonotone ranking function) only incurs a logarithmic cost factor in data complexity.
On the other hand, unranked enumeration has been shown to be tractable for  
CQs with ``long'' inequality predicates~\cite{qichen22comparisons} (i.e., between non-adjacent relations in the join tree), 
where it is unclear whether ranked enumeration can be performed efficiently.

\introparagraph{Ranked enumeration}
The ranked enumeration paradigm can be seen as sorting done \emph{incrementally}.
For a set of elements that is given, Paredes and Navarro~\cite{paredes06iqs} propose a simple modification
to the well-known quicksort algorithm~\cite{hoare61qsort} for ranked enumeration,
which we leverage as a subprocedure for the \QUICK variant of \ANYKPART.
For graph-pattern queries that are less general than CQs,
Chang et al.~\cite{chang15enumeration} 
introduce the idea behind the \LAZY variant of \ANYKPART, while
Yang et al.~\cite{yang2018any} describe \MIN.
The idea of applying our any-$k$ algorithms to cyclic CQs using (multiple) hypertree decompositions
appeared in a very preliminary version of this work~\cite{YangRLG18:anyKexploreDB}.
The generalization of the REA algorithm~\cite{jimenez99shortest}
from paths to trees is achieved by \ANYKREC,
as well as the concurrent work of Deep and Koutris~\cite{deep21}.
In comparison with that work, we
(1) reveal the general DP structure of the problem and deep relationships to classic work on $k$-shortest 
paths,
(2) provide a more fine-grained analysis that shows important differences between the algorithms and surprising results such as asymptotically faster time-to-last than sorting,
(3) propose an asymptotically better algorithm that also applies to the underlying problem of path enumeration in a DAG,
(4) include CQs with projections in our study,
and (5) implement and experimentally compare different algorithms.
On the other hand, Deep and Koutris~\cite{deep21} define and support monotonicity properties that are sensitive to the structure of the join tree (or the tree decomposition);
they allow for example the ranking function $f(x,y) + f(z,u)$ for arbitrary $f$
by always placing $x,y$ together in one bag and $z,u$ together in another bag of the decomposition.
For CQs with arbitrary projections (i.e., not necessarily free-connex),
Bagan et al.~\cite{bagan07constenum} provide an algorithm with linear preprocessing and linear delay.
For arbitrary projections and \smonotone ranking functions, Kimelfeld and Sagiv \cite{KimelfeldS2006} give one with linear preprocessing and polynomial delay.
Interestingly, a closer look at these algorithms reveals that the former
can support any lexicographic order, while the latter has linear
(instead of polynomial) delay.
For lexicographic orders and sum-of-weights, these guarantees were later matched by Deep et al.~\cite{deep22ranked}.
The algorithm of Kimelfeld and Sagiv \cite{KimelfeldS2006} uses the Lawler-Murty procedure like \ANYKPART but does not exploit
the DP structure of the problem (i.e., the deviations) since that is not possible for arbitrary CQs.
ContourJoin~\cite{ding21progressive} is a ranked-enumeration algorithm for binary joins,
but, similarly in spirit to top-$k$ joins, does not offer any non-trivial guarantees in the RAM cost model.
Bourhis et al.~\cite{bourhis21ranked} study ranked enumeration for queries in Monadic Second Order (MSO) logic which is more general than CQs, 
but the word data model is simpler since it cannot describe arbitrary relations.
Ranked enumeration under updates has been considered by Berkholz et al.~\cite{berkholz21ranked} for probabilistic databases.

\introparagraph{$k$-shortest paths}
The literature is rich in algorithms for finding the $k$-shortest paths~\cite{eppstein16kbest}
with the sum-of-weights model (i.e., the tropical semiring)
where typically the graph can be cyclic.
Many variants focus on the loopless version where a path cannot visit a node twice~\cite{Chang15topkpaths,hershberger07paths,katoh82kshortest,yen1971finding}, 
a property always satisfied in DAGs. 
Hoffman and Pavley \cite{hoffman59shortest}
introduce the idea of deviations (see \Cref{sec:part}).
Building on that idea, Dreyfus~\cite{dreyfus69shortest}
proposes an algorithm that can be seen as a modification to the procedure of
Bellman and Kalaba~\cite{bellman60kbest}. The \emph{Recursive Enumeration Algorithm}
(REA)~\cite{jimenez99shortest} uses the same set of equations as Dreyfus,
but applies them in a top-down recursive manner. 
Our \ANYKREC builds upon REA.
To the best of our knowledge, prior work has ignored the fact that this algorithm
reuses computation in a way that can asymptotically outperform sorting in some cases.
In another line of research, Lawler \cite{lawler72} generalizes an earlier algorithm of 
Murty \cite{murty1968} and applies it to $k$-shortest paths. 
Aside from $k$-shortest paths, the Lawler-Murty procedure has been widely used for a variety of problems in the database community~\cite{golenberg11parallel}.
This procedure and the Hoffman-Pavley deviations are the main ingredients
of \ANYKPART.
Eppstein's algorithm~\cite{eppstein1998finding} has the
best-known complexity $\TT(k) = \O(N + k \log k)$, which does not depend on
the $\O(\ell)$ length of returned paths. 
To achieve that, it returns an $\O(1)$ implicit representation of a path, i.e., its weight
together with a pointer 
to the start of the path from which it can be reconstructed.
For the task of explicitly enumerating paths in a DAG, an $\O(k \ell)$ term is unavoidable and
our \ANYKPART achieves the same complexity as Eppstein's algorithm~\cite{eppstein1998finding}
with a much simpler construction;
\ANYKPARTP asymptotically improves upon this.
In \Cref{sec:map_related_work}, we further discuss the literature on $k$-shortest paths.

\introparagraph{X+Y}
The open X+Y sorting problem asks whether the pairwise sums of two $n$-size sets can be sorted faster than the naive $\O(n^2 \log n)$ algorithm~\cite{bremner06xy,demaine06sum,fredman76xy,harper75xy}.
This corresponds precisely to sorting the answers to $Q_{XY}(x, y) \datarule R_1(x), R_2(y)$.
The problem of selection where only the $k^{\textrm{th}}$ answer needs to be returned
has also been studied for X+Y~\cite{frederickson82selection} as well as
for more sets $X_1 + X_2 + \cdots + X_m$~\cite{johnson78xy}.
Our work shows that it is possible to improve the logarithmic factor in sorting
$X_1 + X_2 + \cdots + X_m$.

\introparagraph{CSPs and Homomorphisms}
The connections between conjunctive query evaluation, constraint satisfaction (CSP), 
and homomorphisms between relational structures are well known
\cite{DBLP:journals/siamcomp/FederV98,gyssens94csp,KOLAITIS2000302,vardi00csp}.
Ranked enumeration using the Lawler-Murty procedure in a
way similar to Kimelfeld and Sagiv~\cite{KimelfeldS2006} for CQs
has also been proposed for CSPs
\cite{DBLP:journals/jcss/GottlobGS18,DBLP:conf/cp/GrecoS11}.
Our results apply directly to these problems as well;
for example, they
generalize
minimum-cost homomorphism problems~\cite{gutin06homomorphism,hell12homomorphism}
to ranked enumeration of homomorphisms.

\section{Conclusions and Future Work}

We described a framework for ranked enumeration of answers to full acyclic CQs and, by extension, free-connex CQs and cyclic CQs.
More generally, it applies to a wide range of DP problems that can be expressed through semirings,
such as DNA sequence alignment~\cite{dpv08book}
or Viterbi decoding~\cite{seshadri94viterbi}.
Given that special cases of this problem had been studied in isolation in the past,
our general framework aims to avoid piecemeal rediscovery in the future.
The algorithms we proposed, such as \ANYKPARTP which asymptotically dominates all previously known alternatives,
directly apply to all these problems and improve the state-of-the-art upper bounds.

Although our work thoroughly resolves ranked enumeration for cases where a DP structure can be identified in a problem,
it remains to be seen whether ranked enumeration can be done efficiently
in cases where this structure is not clear, as is the case
in \emph{Unions of CQs}~\cite{carmeli21ucqs} or CQs with ``long'' inequality predicates~\cite{qichen22comparisons}.
Furthermore, our analysis could be extended by going beyond the worst-case and studying the \emph{average-case} behavior~\cite{SA0G18} of our algorithms.
On the practical side, we showed the advantage of any-$k$ over the traditional DBMS approach,
but we have left for future work to determine the best way to integrate our techniques with a DBMS
query processor and optimizer.
To that end, an external-memory adaptation and analysis of our algorithms might be necessary.

%% file: appendix.tex
\clearpage
\appendix

\section{Nomenclature}

\begin{table}[h]
\centering
\small
\begin{tabularx}{\linewidth}{@{\hspace{0pt}} >{$}l<{$}  @{\hspace{2mm}}  X @{}}
\hline
\textrm{Symbol}	& Definition 	\\
\hline
[m]             & Integers $\{ 1, \ldots, m \}$ \\
\phantom{}[m]_0           & Integers $\{ 0, \ldots, m \}$ \\
D               & Database instance \\
n               & Maximum size of a relation \\
Q               & Conjunctive Query $Q(\vec X) \datarule R_1(\vec{X_1}),\ldots, R_\stages(\vec{X_\stages})$ \\
\vec X, \vec Y, \vec Z & Lists of variables \\
x, y, z         & Variables \\
\free(Q)        & The free variables of CQ $Q$ \\
R_1, \ldots, R_\ell & Relations \\
\ell            & Number of atoms of a CQ \\
\alpha          & Maximum arity of an atom of a CQ \\
\dom            & The domain of the relations \\
Q(D)            & Set of query answers of $Q$ over $D$ \\
q \in Q(D)      & A query answer \\
\oplus, \otimes & Operators of a commutative selective dioid \\
\diam(Q)        & Diameter of the hypergraph of CQ $Q$ \\
G(V, E)         & Graph with nodes $V$ and edges $E$ \\
\nodenum        & Number of nodes in a graph \\
s, t            & Source node $s$, target node $t$ \\
S_1, \ldots, S_\ell & Stages of a node-partitioned graph \\
w               & Weight function for input tuples/query answers/graph edges/paths \\
\solW_k(v)      & Weight of $k^\textrm{th}$ shortest path from $v$ \\
\sol_k(v)       & $k^\textrm{th}$ shortest path from $v$ \\
\langle v_1, \ldots, v_\lambda \rangle & A path of length $\lambda$ \\
\suc(v_i, v_{i+1}) & Next best child of $v_i$ after $v_{i+1}$ according to the optimal weight needed to reach $t$ \\
\Choices_1(v)   & The paths compared by DP at node $v$: thy go to a child of $v$ and then optimally to $t$ \\
E_{pc}          & Edges (decisions) from stage $\Sset_p$ to stage $\Sset_c$ in T-DP \\
\Ch(\Sset_p), \Ch(v_p) & Indexes of children stages of $\Sset_p$ or $v_p \in \Sset_p$ in T-DP \\
\parent(\Sset_c), \parent(v_c) & Index of parent stages of $\Sset_c$ or $v_c \in \Sset_c$ in T-DP \\
\llceil \Sset_i \rrceil, \llceil v_i \rrceil & Index of all stages inn the subtree of $\Sset_i$ or $v_i \in \Sset_i$ excluding $i$ in T-DP \\
\serialw        & Width of serial decomposition of a T-DP problem \\
\hline
\end{tabularx}
\end{table}

\section{Order-By Queries in Factorized Databases}
\label{appendix:fdb_order}

Factorized databases (FDBs)~\cite{bakibayev12fdb,olteanu16record,olteanu12ftrees,olteanu15dtrees}
support constant-delay enumeration of query answers according to a desired lexicographic order on the attributes~\cite{bakibayev13fordering}.
Lexicographic orders are a special case of the ranking function considered in this paper and our approach supports them (see \Cref{sec:lex}) but with logarithmic delay.
Here we look closer at the differences between the two approaches for this special case of lexicographic orders and show that our approach can be asymptotically better in certain cases despite the logarithmic delay. 

First, we provide a very short description of the main idea behind factorized databases and we refer the reader to the original works for a deeper understanding. 
To achieve a succinct representation, factorized databases repeatedly apply the distributivity law in an order described by a tree structure whose nodes are the attributes \cite{olteanu12ftrees}. 
Intuitively, if $X$ is the attribute of a node of the tree and $\mathrm{anc}(X)$ are its ancestor attributes, then every value $x \in X$ is represented at most once for each combination of values of $\mathrm{anc}(X)$.
D-representations~\cite{olteanu15dtrees} provide further succinctness by making the dependencies of each attribute in the tree explicit. 
This means that some attributes in $\mathrm{anc}(X)$ might not actually determine what the possible $X$ values are.
Truly dependent ancestor attributes of a node are denoted as $\mathrm{key}(X)$.
Each value $x \in X$ is then represented at most once for each combination of values of $\mathrm{key}(X)$.

Such a factorized representation allows constant-delay unranked enumeration of the query answers.
Yet for specific lexicographic orders, there are two conditions that have to be met: 
($i$) the order-by attributes have to be ``at the top'' of the tree and 
($ii$) the tree order has to agree with the lexicographic order.
If the tree order is not in agreement 
(e.g., we want $A$ before $B$ but $A$ is a child of $B$ in the tree), 
then the whole representation has to be restructured. 
The restructuring operation takes an input representation and transforms it to an 
output representation consistent with the lexicographic order 
in time linear (ignoring log factors) in the input \emph{and output representation sizes}.
However, the output representation itself could be very large.
We next illustrate the simplest example where an ill-chosen lexicographic order results in a quadratic representation for a simple binary join.

\begin{example}[Lexicographic orders]\label{ex:fdbs_lexicographic}
Consider the 2-path query $Q_{P2}(A, B, C) \datarule R(A, B), S(B, C)$. 
As usual, $n$ is the maximum number of tuples in a relation.
Ideally, we would want to factorize it using a tree that has $B$ as the root and $A, C$ as its children.
That way, every $A$ and $C$ value in the query result would be represented independently for each $B$ value.
However, for the lexicographic order $A \rightarrow C \rightarrow B$ this factorization is not in agreement since $B$ comes after $A$ and $C$.
The \emph{only possible} tree that satisfies the condition ($ii$) above is a path from $A$ to $C$ to $B$. 
Note that the tree with $A$ as the root and $B$, $C$ as the children is not possible because of the ``path condition'' in factorized databases: attributes that belong to the same relation ($B$ and $C$ here) are dependent and have to lie in the same root-to-leaf path.
In the only valid tree, $\mathrm{key}(B) = \{ A, C \}$.
According to Lemma 7.20 in \cite{olteanu15dtrees}, there exist arbitrarily large databases such that the number of $B$ values in the representation is at least $n^{\rho^{*}(B \cup key(B))}$, where $\rho^{*}$ is the fractional edge cover of the query, thus $\Omega(n^2)$.

\Cref{fig:FDBs} presents a concrete instance where this happens.
For this database, the single $B$-value $1$ will be represented once for each combination of $A, C$ values and there are $n^2$ of them. 
In contrast, our approach begins the enumeration after only linear time preprocessing.
Thus in this case, the preprocessing step of FDBs takes $\O(n^2)$ after which results can be enumerated in constant time.
In contrast, our approach has $\TT(k) = \O(n + k \log k)$.
\end{example}

\begin{figure}[t]
	\small
\centering
	\centering
	\setlength{\tabcolsep}{0.4mm}
		\hspace{0mm}
				\mbox{
				\begin{tabular}[t]{ >{$}c<{$} | >{$}c<{$}  >{$}c<{$} }
	 			\mathbf{R}	&  A 	& B \\
				\hline
				& $1$ & \markZwicky(R1){$1$} 	\\
				& $2$ & \markZwicky(R2){$1$} 	\\
				& \ldots & \ldots \\
				& $n$ & \markZwicky(R3){$1$}
				\end{tabular}
		}
		\hspace{0mm}
		\mbox{
				\begin{tabular}[t]{ >{$}c<{$} | >{$}c<{$}  >{$}c<{$} }
	 			\mathbf{S}	&  B 	& C \\
				\hline
				& \markZwicky(S1){$1$} & $1$ 	\\
				& \markZwicky(S2){$1$} & $2$	\\
				& \ldots & \ldots \\
				& \markZwicky(S3){$1$} & $n$
				\end{tabular}
		}
\caption{\Cref{ex:fdbs_lexicographic}:
Example database showing sub-optimality of factorized databases for lexicographic order $A \rightarrow C \rightarrow B$.
}
\label{fig:FDBs}
\tikzZwicky[blue](R1.east)(S1.west)
\tikzZwicky[blue](R1.east)(S2.west)
\tikzZwicky[blue](R1.east)(S3.west)
\tikzZwicky[blue](R2.east)(S1.west)
\tikzZwicky[blue](R2.east)(S2.west)
\tikzZwicky[blue](R2.east)(S3.west)
\tikzZwicky[blue](R3.east)(S1.west)
\tikzZwicky[blue](R3.east)(S2.west)
\tikzZwicky[blue](R3.east)(S3.west)
\end{figure}

\section{Top-k joins in our cost model}
\label{appendix:topk_suboptimal}

Consider the database $I_2$ from \cref{fig:TA} with $\ell=3$ relations and $n=10$ tuples per relation.
The top output tuple is marked in blue; it consists of the lightest tuples from the first $\ell-1$ relations and the heaviest tuple from $R_\ell$.
J*~\cite{natsev01} and Rank-Join~\cite{ilyas04} access 
the tuples in the input relations by decreasing weight. 
Their cost model takes into account only the number of database accesses, 
hence they try to minimize the depth up to which the sorted relations have to be accessed in order to find the top-k results. 
In this case, 
\emph{both J* and Rank-Join will consider the $(n-1)^{\ell-1}$ combinations between $R$ and $S$
before getting to the 
the top-1 tuple $(r_0, s_0, t_0)$}. 
This happens because J* over-estimates their weight by using the large weight of $t_0$ to upper-bound them, 
while Rank-Join by default joins each newly encountered tuple with all the other ones seen so far.
In contrast, our approach achieves $O(n)$ for the top ranked result (in data complexity).

\begin{figure}[t]
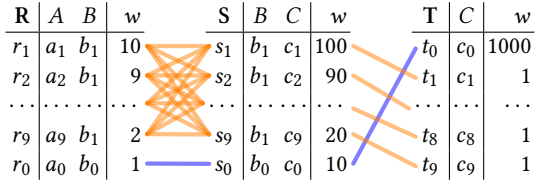

	\small
\centering
	\centering
	\setlength{\tabcolsep}{0.4mm}
		\hspace{0mm}
				\mbox{
				\begin{tabular}[t]{ >{$}c<{$} | >{$}c<{$}  >{$}c<{$} | >{$}r<{$} }
	 			\mathbf{R}	&  A 	& B 	& w\\
				\hline
				r_1		& a_1 	& b_1 	& \markZwicky(1-1){$10$}		\\			
				r_2		& a_2 	& b_1 	& \markZwicky(2-1){$9$}\\
				\cdots 	& \cdots & \cdots 	& \markZwicky(3-1){$\cdots$\vphantom{\large $0$}} \\
				r_9		& a_9 	& b_1 	& \markZwicky(4-1){$2$}	\\				
				r_{0}	& a_{0} & b_0 	& \markZwicky(0-1){$1$}								
				\end{tabular}			
		}
		\hspace{5mm}
		\mbox{
				\begin{tabular}[t]{ >{$}c<{$} | >{$}c<{$} >{$}c<{$} | >{$}r<{$}}
	 			\mathbf{S}	&  B 	& C 	& w\\
				\hline
				\markZwicky(1-2){\vphantom{\large $0$}$s_1$}	& b_1 & c_1 	& \markZwicky(1-3){$100$}	\\			
				\markZwicky(2-2){\vphantom{\large $0$}$s_2$}	& b_1 & c_2 	& \markZwicky(2-3){$90$}  	\\
				\markZwicky(3-2){\vphantom{\large $0$}$\cdots$} & \cdots  & \cdots 	& \markZwicky(3-3){$\cdots$} \\
				\markZwicky(4-2){\vphantom{\large $0$}$s_9$}	& b_1 & c_9 	& \markZwicky(4-3){$20$}	\\				
				\markZwicky(0-2){\vphantom{\large $0$}$s_0$}	& b_0 & c_0 	& \markZwicky(0-3){$10$}								
				\end{tabular}	
		}
		\hspace{5mm}
		\mbox{
				\begin{tabular}[t]{ >{$}c<{$} | >{$}c<{$} | >{$}r<{$}}
	 			\mathbf{T}	& C & w		\\
				\hline
				\markZwicky(0-4){$t_0$}	& c_0	& 1000\\	
				\markZwicky(1-4){$t_1$}	& c_1	& 1\\						
				\markZwicky(2-4){$\cdots$} & \cdots & \cdots \\				
				\markZwicky(3-4){$t_8$}	& c_8	& 1\\								
				\markZwicky(4-4){$t_9$}	& c_9	& 1
				\end{tabular}			
		}
\caption{
Database $I_2$ showing sub-optimality of J* and Rank-Join. (\Cref{appendix:topk_suboptimal})}
\label{fig:TA}
\tikzZwicky[blue](0-1.east)(0-2.west)
\tikzZwicky[blue](0-3.east)(0-4.west) %

\tikzZwicky[orange](1-1.east)(1-2.west)
\tikzZwicky[orange](1-1.east)(2-2.west)
\tikzZwicky[orange](1-1.east)(3-2.west)
\tikzZwicky[orange](1-1.east)(4-2.west)
\tikzZwicky[orange](2-1.east)(1-2.west)
\tikzZwicky[orange](2-1.east)(2-2.west)
\tikzZwicky[orange](2-1.east)(3-2.west)
\tikzZwicky[orange](2-1.east)(4-2.west)
\tikzZwicky[orange](3-1.east)(1-2.west)
\tikzZwicky[orange](3-1.east)(2-2.west)
\tikzZwicky[orange](3-1.east)(3-2.west)
\tikzZwicky[orange](3-1.east)(4-2.west)
\tikzZwicky[orange](4-1.east)(1-2.west)
\tikzZwicky[orange](4-1.east)(2-2.west)
\tikzZwicky[orange](4-1.east)(3-2.west)
\tikzZwicky[orange](4-1.east)(4-2.west)

\tikzZwicky[orange](1-3.east)(1-4.west)
\tikzZwicky[orange](2-3.east)(2-4.west)
\tikzZwicky[orange](3-3.east)(3-4.west)
\tikzZwicky[orange](4-3.east)(4-4.west)

\end{figure}

\section{Minimal example of a non-commutative selective dioid}
\label{appendix:dioid_example}

In \Cref{sec:algebra}, we mentioned that commutativity of ``multiplication'' in selective dioids
is an important algebraic property for us since non-commutative structures would not allow
any order of aggregation, restricting the space of possible join trees for CQs.
However, all typical examples of selective dioids on infinite domains are commutative,
and it can be difficult to find an example of a non-commutative structure.
\Cref{fig:nc_dioid} show such an example on a domain of 5 elements that we identified with
the Mace4 tool~\cite{mccune03mace4}.

\begin{figure*}[t]
    \begin{subfigure}{0.48\linewidth}
        \centering

        \begin{tabular}{| c | c | c | c | c | c |}
        \multicolumn{6}{c}{ \textbf{Operator $\oplus$} } \\ \hline
        & \textbf{0} & \textbf{1} & \textbf{2} & \textbf{3} & \textbf{4} \\ \hline
        \textbf{0} & 0 & 0 & 0 & 0 & 0 \\ \hline
        \textbf{1} & 0 & 1 & 1 & 1 & 1 \\ \hline
        \textbf{2} & 0 & 1 & 2 & 2 & 2 \\ \hline
        \textbf{3} & 0 & 1 & 2 & 3 & 3 \\ \hline
        \textbf{4} & 0 & 1 & 2 & 3 & 4 \\ \hline
        \end{tabular}
        \caption{Definition of $x \oplus y$.}
        \label{fig:nc_dioid-1}

    \end{subfigure}%
    \hfill
    \begin{subfigure}{0.48\linewidth}
        \centering

        \begin{tabular}{| c | c | c | c | c | c |}
        \multicolumn{6}{c}{ \textbf{Operator $\otimes$} } \\ \hline
        & \textbf{0} & \textbf{1} & \textbf{2} & \textbf{3} & \textbf{4} \\ \hline
        \textbf{0} & 0 & 1 & 2 & 3 & 4 \\ \hline
        \textbf{1} & 1 & 1 & 2 & 3 & 4 \\ \hline
        \textbf{2} & 2 & 2 & 2 & 3 & 4 \\ \hline
        \textbf{3} & 3 & 4 & 4 & 4 & 4 \\ \hline
        \textbf{4} & 4 & 4 & 4 & 4 & 4 \\ \hline
        \end{tabular}
        \caption{Definition of $x \otimes y$}
        \label{fig:nc_dioid-2}

    \end{subfigure}%
    \hfill
    \caption{Example of non-commutative selective dioid. Evaluation tables of the operators are shown with the first operand $x$ as the row and the second operand $y$ as the column. The order induced by $\oplus$ is 
    $0 \preceq 1 \preceq 2 \preceq 3 \preceq 4$. The operator $\otimes$ is non-commutative since 
    $1 \otimes 3 \neq 3 \otimes 1$.}
    \label{fig:nc_dioid}
\end{figure*}

\section{\ANYKPART Variants}
\label{sec:part_variants}

As we discussed in \Cref{sec:part}, the specific implementation of the successor function
$\suc(v_i, v_{i+1})$ for an edge $(v_i, v_{i+1})$
gives rise to a number of variants that have appeared in prior work.
In the conference version of this article~\cite{tziavelis20vldb},
we analyzed and compared these in more detail.
They can broadly be categorized as ``strict'', which follow the successor definition
adopted in the main body of our paper,
and ``relaxed'' which follow a more general definition.

\smallsection{Strict approaches}
A natural implementation of the successor function returns precisely the next-best choice.

Eager Sort (\EAGER):
Since a state might be reached repeatedly through different prefixes, it
may pay off to pre-sort all choice sets by weight and add pointers from each choice
to the next one in sort order. 
Then $\Suc(v_i, y)$ returns the next-best choice at $v_i$ in constant time
by following the next-pointer from $v_{i+1}$.

Lazy Sort (\LAZY):
To lower pre-processing cost, we can leverage the approach Chang et al.~\cite{chang15enumeration}
proposed in the context of graph-pattern search.
Instead of sorting a choice set, it constructs a binary heap in linear time.
Since all but one of the successor requests in a single repeat-loop execution
are looking for the second-best 
choice\footnote{During each execution of the repeat-loop, only the first iteration
of \Cref{line:for1} looks for a lower choice.}, 
the algorithm already pops the top two choices off the heap and moves them into a sorted list. 
For all other choices, the first access popping them from the heap will append them to the
sorted list that was initialized with the top-2 choices. 
As the algorithm progresses, the heap of choices gradually empties out, filling the sorted list
and thereby converging to \EAGER.

\QUICK: Using the Incremental Quicksort algorithm~\cite{paredes06iqs}, we can sort the children of
each node incrementally and obtain the same guarantees as \LAZY in expectation.

\smallsection{Relaxed approaches}
Instead of finding the \emph{single true successor}
of a choice,  what if the algorithm could return a set of \emph{potential successors}?
Correctness is guaranteed, as long as the true successor is contained in this set
or is already in $\Cand$. (Adding potential successors early to $\Cand$ does not
affect correctness, because they have higher weight and would not be popped
from $\Cand$ until it is ``their turn.'')
This relaxation may enable faster successor finding, but inserts
candidates earlier into $\Cand$.

All choices (\MIN):
This approach is based on a construction that Yang et al.~\cite{yang2018any} proposed
in the context of graph-pattern search.
Instead of trying to find the true successor of a choice, 
\emph{all} but the top choice are returned by $\Suc$. 
While this avoids any kind of pre-processing overhead, it inserts
$\O(n)$ potential successors into $\Cand$.

\HEAP:
This approach aims to keep pre-processing at a minimum (like \MIN),
but also return a small number of successors fast (like \EAGER). 
To achieve this, we organize each choice set
as a binary heap. In this tree structure, the root node is the minimum-weight choice and the
weight of a child is always 
greater than
its parent. Function $\Suc(v_i, v_{i+1})$ returns the two children
of $v_{i+1}$ in the tree.
Unlike \LAZY, we never perform a pop operation and the heap stays intact for the entire
operation of the algorithm; it only serves as a partial order on the choice set, pointing
to two successors every time it is accessed.
Also note that the true successor does not necessarily
have to be a child of node $v_{i+1}$. Overall, returning two successors is asymptotically
the same as returning one and heap construction time is linear~\cite{Cormen:2009dp},
hence this approach achieves lower delay.

In \Cref{tab:complexity_part}, we summarize the differences in complexity among these variants.
We remind the reader that in the analysis in the main body of this article, 
we assume the \LAZY variant which is simple and achieves the best $\TT(k)$.
As we argue in \Cref{sec:complexity_measures}, the lower delay of \HEAP may not be practically relevant.

\definecolor{colorbest}{RGB}{77, 175, 74}

\begin{figure*}[t]
\centering
\footnotesize
\renewcommand{\tabcolsep}{1.3pt}
\begin{tabular}{|l|l|l|l|l|l|}
\hline
Algorithm 	& $\TT(k)$ 	& $\Del(k)$	& $\MEM(k)$ \\ 
\hline

\EAGER          &$\bigO(|G| \log |G| + k (\log k + \ell))$
                &\cellcolor{colorbest!20}$\bigO(\log k + \stages)$ 
				&\cellcolor{colorbest!20}$\bigO(|G| + k \stages)$ \\

\LAZY           &\cellcolor{colorbest!20}$\bigO(|G| + k (\log k + \ell))$
                &\cellcolor{colorbest!20}$\bigO(\log k + \stages + \log \nodenum)$ 
				&\cellcolor{colorbest!20}$\bigO(|G| + k \stages)$ \\

\QUICK          &\cellcolor{colorbest!20}$\bigO(|G| + k (\log k + \ell))$
                &\cellcolor{colorbest!20}$\bigO(\log k + \stages + \log \nodenum)$ 
				&\cellcolor{colorbest!20}$\bigO(|G| + k \stages)$ \\

\MIN           &$\bigO(|G| + k (\log k + \ell \nodenum))$
                &$\bigO(\log k + \ell \nodenum)$ 
				& $\bigO(|G| + \min\{k \nodenum, \out\} \stages)$ \\
    
\HEAP           &\cellcolor{colorbest!20}$\bigO(|G| + k (\log k + \ell))$
                &\cellcolor{colorbest!20}$\bigO(\log k + \stages)$ 
				&\cellcolor{colorbest!20}$\bigO(|G| + k \stages)$ \\

\hline
\end{tabular} 
\caption{Complexity of \ANYKPART variants for ranked-enumeration of $s$-$t$ paths in a DAG.
Best performing algorithms are colored in green.
$|G|$ is the graph size, $\nodenum$ is the number of nodes, 
$\out$ is the total number of paths, and
$\ell$ is the maximum length of a path.
The guarantees of \QUICK are randomized and hold in expectation.
}
\label{tab:complexity_part}
\end{figure*}

\section{Map of related work for $k$-shortest paths}
\label{sec:map_related_work}

\begin{figure}[tb]
\centering
\includegraphics[height=6.5cm]{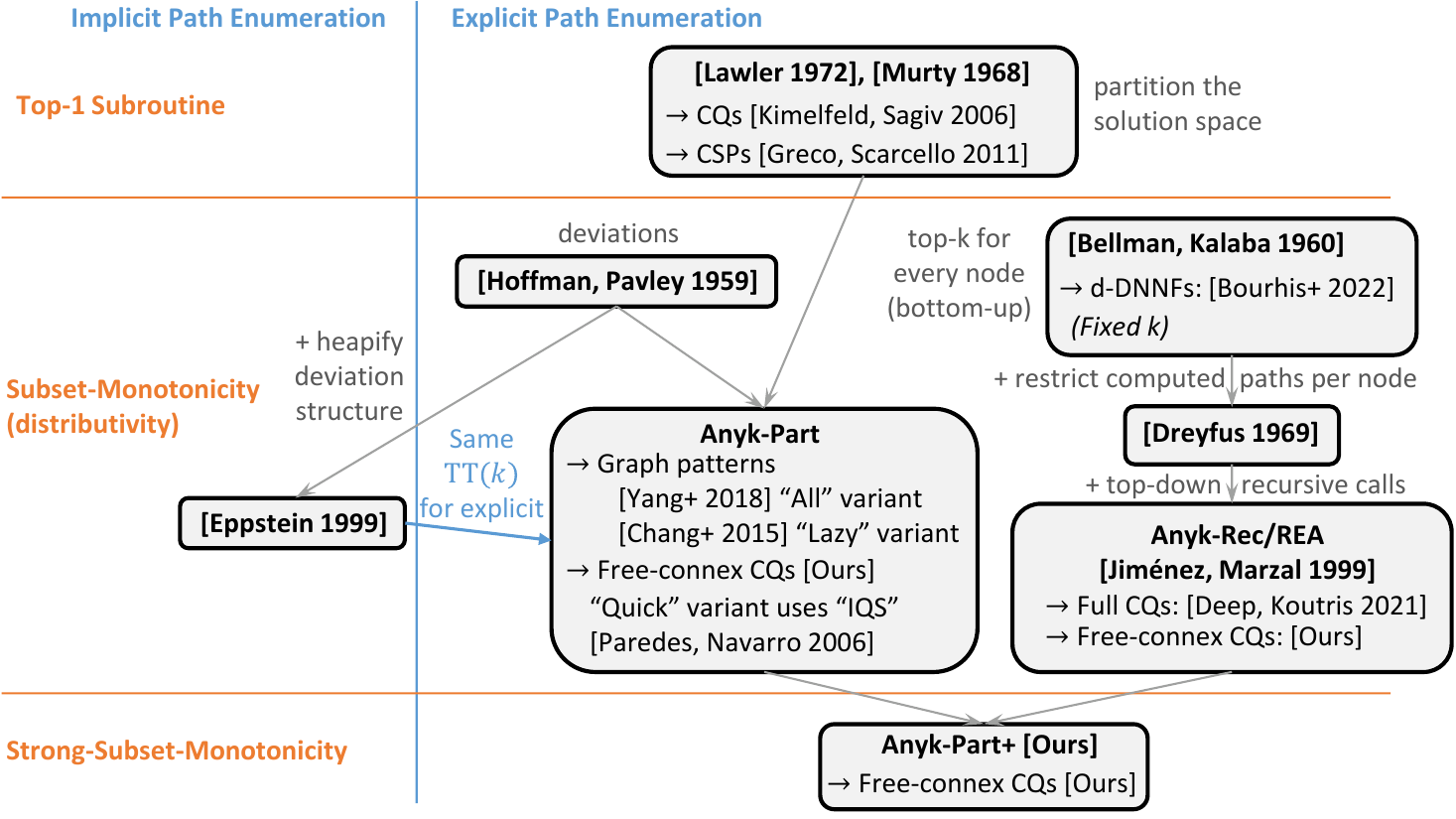}
\caption{Map of related work for the problem of $k$-shortest paths on a DAG.
Arrows indicate the progression of ideas historically.
Within each node, we note the application of these algorithms to different areas.
Implicit path enumeration~\cite{eppstein1998finding} (left) refers to the problem where an $\O(1)$ representation of each path is returned and the lower bound of $\Omega(k \ell)$ no longer applies.
In this work, we focus on explicit path enumeration (right) where each path has to be returned as a list of edges.
Most algorithms in the literature operate on a sum-of-weights model and can be generalized to \smonotone ranking functions (middle). 
The Lawler-Murty procedure~\cite{murty1968,lawler72} only requires an (efficient) sub-routine for solving the top-1 problem (top).
Our \ANYKPARTP algorithm requires the stronger \ssmonotone property (bottom).
}
\label{fig:map_related}
\end{figure}

\Cref{fig:map_related} depicts important ideas and the connections between them in
the history of the $k$-shortest paths problem on a DAG that we study in this work.